\documentclass[pdflatex,sn-mathphys-ay]{sn-jnl}

\usepackage{graphicx}%
\usepackage{multirow}%
\usepackage{amsmath,amssymb,amsfonts}%
\usepackage{amsthm}%
\usepackage{mathrsfs}%
\usepackage[title]{appendix}%
\usepackage{xcolor}%
\usepackage{textcomp}%
\usepackage{manyfoot}%
\usepackage{booktabs}%
\usepackage{algorithm}%
\usepackage{algorithmicx}%
\usepackage{algpseudocode}%
\usepackage{listings}%
\usepackage{subfig}
\usepackage{url}
\usepackage{xspace}
\usepackage{hyperref}

\theoremstyle{thmstyleone}%
\newtheorem{theorem}{Theorem}
\theoremstyle{thmstyletwo}%
\newtheorem{remark}{Remark}%
\newtheorem{lemma}{Lemma}

\theoremstyle{thmstylethree}%

\newcommand{\bs}{\boldsymbol}
\newcommand{\ve}[1]{\mbox{\boldmath ${#1}$}}
\newcommand{\vesub}[2]{\mbox{{\boldmath ${#1}$}$_{#2}$}}
\newcommand{\vesup}[2]{\mbox{{\boldmath ${#1}$}$^{#2}$}}
\newcommand{\vess}[3]{\mbox{{\boldmath ${#1}$}$_{#2}^{#3}$}}

\newcommand{\tve}[1]{\tilde{\ve{#1}}}
\newcommand{\tvesub}[2]{\tilde{\ve{#1}}_{#2}}
\newcommand{\tvesup}[2]{\tilde{\ve{#1}}^{#2}}
\newcommand{\tvess}[3]{\tilde{\ve{#1}}_{#2}^{#3}}

\newcommand{\bve}[1]{\underline{\ve{#1}}}

\newcommand{\bvesup}[2]{\underline{\ve{#1}}^{#2}}

\newcommand{\hve}[1]{\hat{\ve{#1}}}
\newcommand{\hvesub}[2]{\hat{\ve{#1}}_{#2}}
\newcommand{\hvesup}[2]{\hat{\ve{#1}}^{#2}}
\newcommand{\hvess}[3]{\hat{\ve{#1}}_{#2}^{#3}}

\newcommand{\diag}{\mbox{diag}}
\newcommand{\tr}{\mbox{tr}}
\newcommand{\Var}{\mbox{Var}}
\newcommand{\Cov}{\mbox{Cov}}
\newcommand{\E}{\mbox{E}}
\newcommand{\argmax}{\arg\!\max}

\newcommand{\fall}{\mathbf{f}}
\newcommand{\gall}{\mathbf{g}}

\begin{document}

\title[Analyzing Functional Data with a Mixture of Covariance Structures Using a CBSS]{Analyzing Functional Data with a Mixture of Covariance Structures Using a Curve-Based Sampling Scheme}


\author[1]{\fnm{Yian} \sur{Yu}}\email{yianyu@umich.edu}

\author*[2]{\fnm{Bo} \sur{Wang}}\email{bo.wang@le.ac.uk}

\author*[1,3]{\fnm{Jian Qing} \sur{Shi}}\email{shijq@sustech.edu.cn}

\affil[1]{\orgdiv{Department of Statistics and Data Science, College of Science}, \orgname{Southern University of Science and Technology}, \orgaddress{\city{Shenzhen}, \country{China}}}

\affil[2]{\orgdiv{School of Computing and Mathematical Sciences}, \orgname{University of Leicester}, \orgaddress{\city{Leicester}, \postcode{LE1 7RH}, \country{UK}}}

\affil[3]{\orgname{National Center for Applied Mathematics}, \orgaddress{\city{Shenzhen}, \country{China}}}


\abstract{Motivated by distinct walking patterns in real-world free-living gait data, this paper proposes an innovative curve-based sampling scheme for the analysis of functional data characterized by a mixture of covariance structures. Traditional approaches often fail to adequately capture inherent complexities arising from heterogeneous covariance patterns across distinct subsets of the data. We introduce a unified Bayesian framework that integrates a nonlinear regression function with a continuous-time hidden Markov model, enabling the identification and utilization of varying covariance structures. One of the key contributions is the development of a computationally efficient curve-based sampling scheme for hidden state estimation, addressing the sampling complexities associated with high-dimensional, conditionally dependent data. This paper details the Bayesian inference procedure, examines the asymptotic properties to ensure the structural consistency of the model, and demonstrates its effectiveness through simulated and real-world examples.}

\keywords{Conditionally dependent data, Continuous-time hidden Markov model, Gaussian process prior, Nonlinear functional regression model, Varying Covariance structures}



\maketitle

\section{Introduction}
\label{sec:Intro}
Functional data analysis has emerged as a vibrant area of research over the past several decades, primarily due to the increasing prevalence of high-resolution data collected over a continuum; see the detailed discussions in \citet{ramsay2005functional} and among others. In functional regression models, we mainly concern mean function \citep{CaiOptimal2011,chamidah2024estimation} and covariance function \citep{rice1991estimating,Wang2020low}. Within the context of functional data, structural changes in the curves often exist. However, most existing studies focus mainly on differences in mean functions, and treat these changes primarily as abrupt ``shocks" \citep{aue2009estimation,horvath2022change}. In contrast, as illustrated in Figure \ref{fig:FoG-real-intro} and discussed in Section \ref{sec:App}, our analysis of real accelerator signal data indicates that the mean curve remains nearly constant around zero, but the covariance structures associated with different walking patterns vary. This requires identifying specific time periods during which participants experience abnormal episodes.

To address this problem, we propose a Bayesian nonlinear regression model with a Markov chain that robustly models the evolution of underlying hidden state processes. Specifically, continuous-time hidden Markov models (HMMs) are used to define underlying state processes \citep{Rabiner1986AnIntro,cappe2005inference}. Recent developments can be found in, e.g., \citet{liu2015efficient}, \citet{Zhou2020Continuous}. However, almost all of the previous works assume that the response variables are independent given the hidden states. The dependencies among response variables conditioned on the hidden states have not been well studied, e.g., \citet{Hamada2016Modeling}, \citet{jung2020HMMGP}, and \citet{Li2023HMMGPFR}, a property that is natural under the functional data framework and in the context of long-term, continuously monitored real data.

Estimating hidden state realizations is very challenging when accounting for these conditional dependencies. The conventional forward-backward algorithm \citep{Rabiner1989Atutorial} used in HMMs is not directly applicable here due to violations of the conditional independence assumption required in standard HMMs. When conditional dependencies exist, the possible values of the state sequence\textemdash the vector consisting of the realized values of the hidden state process at each recorded time point\textemdash increase exponentially with the number of data points (sample size) $n$, specifically as $M^n$, where $M$ denotes the number of hidden states. Sample size $n$ is usually very large for functional data, for instance, in the accelerometer data discussed in Section \ref{sec:App}, datasets collected at the rate of $100$ Hz resulting in a sample size exceeding hundreds even for a walking period lasting just a few seconds. To tackle these computational challenges, we propose an innovative curve-based sampling scheme. This approach transforms the challenge of enumerating various state sequence combinations into estimating the underlying probabilities of the curve-based process, thus circumventing the computational burdens arising from combinatorial explosion.

The nonlinear functional data we modeled is defined by a Gaussian process (GP) prior. As shown, for example, in \citet{bernardo1998regression}, \citet{williams2005gaussian}, and \citet{shi2011gaussian}, the GP regression framework naturally accommodates regression problems involving multidimensional functional covariates, allowing the modeling of nonlinear regression model as well as the covariance structures of functional data. By specifying a prior of the covariance kernel, this Bayesian nonparametric regression technique addresses a wide class of nonlinear effects while avoiding the issues of the curse of dimensionality due to multidimensional covariates. Recent advancements, such as \citet{datta2016hierarchical} and \citet{wilson2021pathwise}, have underscored computational efficiency as an important research hotspot in GP prior studies, particularly because conventional methods become computationally intractable when applied to complex dependent functional data. In our specific problem, these computational challenges are even more acute, rendering standard approaches impractical and necessitating the development of innovative algorithmic solutions. In this paper, we make the following contributions: (1) We propose a new Bayesian nonlinear functional regression model that employs a continuous-time HMM with a GP prior, specifically designed for dependent functional data with a mixture of covariance structures. (2) We develop an innovative curve-based sampling scheme to address computational challenges, an approach that can also be applied to similar issues in functional data analysis. (3) We provide proofs for the associated theoretical framework.

The article is organized as follows. Section \ref{sec:The Model} introduces the model with a continuous-time HMM, emphasizing conditional dependencies and variations in the covariance structure. Further details of the parameter estimation and prediction through Bayesian inference are also provided in this section. The curve-based sampling scheme for estimating the hidden states is presented in Section \ref{sec:CBSS}. Section \ref{sec:Theory} examines the identifiability and asymptotic properties that ensure the structural consistency of the model. Simulation results are reported in Section \ref{sec:Simu}, and the motivating examples with multi-output modeling are discussed in Section \ref{sec:App}. Section \ref{sec:Discussion} is the discussion. Some technical details and theoretical proofs are provided in the appendices.

\section{Nonlinear Functional Regression Model with Varying Covariance Structures}\label{sec:The Model}
\subsection{The Model}
A general functional regression model can be expressed as $y(t)=\mu (\ve{x}(t) )+f(\ve{x}(t)\big)+\epsilon(t),$
where $\left\{y(t),\, t\in \mathcal{T} \subset \Re \right\}$ represents a functional or longitudinal response variable and its value is dependent at different points in time $t$, $\ve{x}(t)$ denotes a $d$-dimensional vector of functional or scalar covariates, and $\epsilon(t)\stackrel{i.i.d.}{\sim} N(0\, ,\,\sigma^2)$ is the measurement error. The nonlinear relationship between the functional response $y(t)$ and the covariates $\ve{x}(t)$ is modeled by a mean structure $\mu(\ve{x}(t))$ and a nonlinear random effects term $f(\ve{x}(t))$. The covariance structure of $y(t)$ is usually defined via $f(\ve{x}(t))$. The detailed discussion can be found in, e.g., \citet{shi2007gaussian}, \citet{Yi2011Penalized}, and \citet{wang2014generalized}. However, as illustrated in Figure \ref{fig:FoG-real-intro}, the functional response can exhibit continuously alterable covariance structures. We therefore need a new model to fit this type of data.
\begin{figure}
\centerline{\includegraphics[width=1\linewidth]{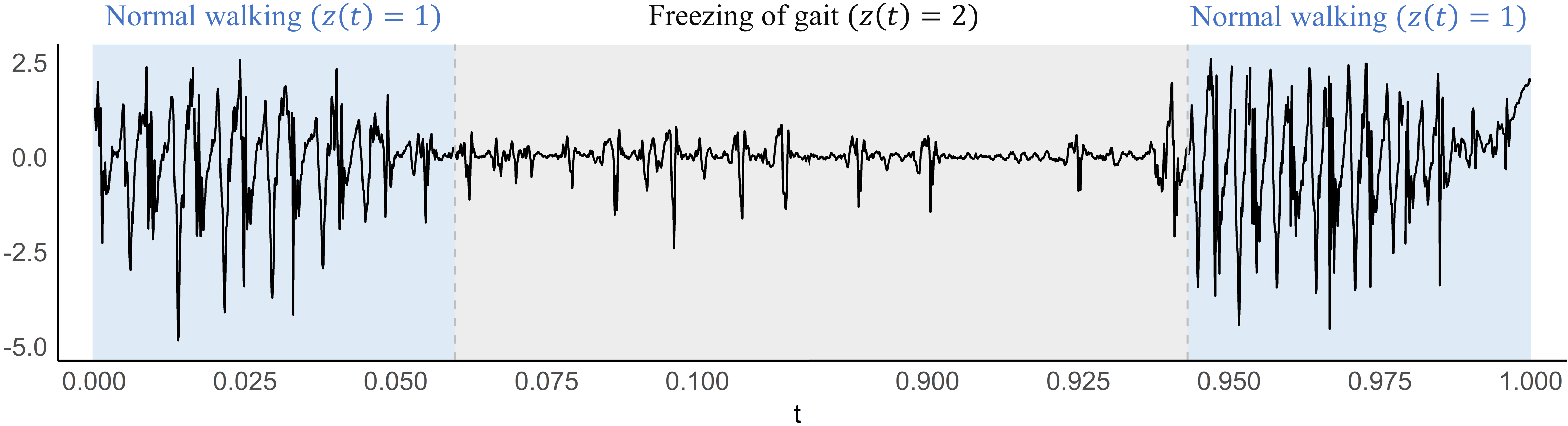} }
	\caption{The vertical acceleration data with variations in the covariance structure associated with different walking patterns (normal walking and freezing of gait) collected from patient of Parkinson's disease. }
	\label{fig:FoG-real-intro}
\end{figure}

The key idea is to introduce a latent process with finite states, where each state is associated with a distinct covariance structure. We propose the following nonlinear functional regression model with continuous-time hidden states: \begin{equation}
	\begin{aligned}
		y(t)=&\mu_{z(t)}(\ve{x}(t))+f_{z(t) }(\ve{x}(t))+\epsilon(t),\ \mu_{z(t)}(\ve{x}(t))=\mu_{m}(\ve{x}(t)),\\
		f_{z(t)}(\ve{x}(t) )=&f_m(\ve{x}(t)), \text{ when } z(t)=m,\ m \in \{1,\cdots, M\}.
	\end{aligned}\label{eq:Model-part1}
\end{equation} Here, $z(t)$ is the hidden state process, where $m$ specifies one of the $M$ (potentially unknown) distinct states. Referring to Figure \ref{fig:FoG-real-intro}, we have two states, i.e., $M=2$, where $z(t)=1$ indicates normal walking and $z(t)=2$ is the freezing of gait. We assume that the hidden state process follows a continuous-time Markov model (CTMM), which will be detailed later. The CTMM model accounts for the dependency of each state not only on the preceding state but also on the time duration spent in that state \citep{ross1996stochastic}.

In contrast to conventional HMMs that impose the output independent assumption\textemdash where the observation $y(t_i)$ depends solely on the current state $z(t_i)$ and is conditionally independent of other observations $\{y(t),t\neq t_i\}$ given $z(t)$, our framework allows responses to exhibit continuous dependency given the hidden states. This is critical for modeling gradual state transitions, as depicted in Figure \ref{fig:FoG-real-intro}, where covariance structures evolve smoothly rather than abruptly. 

The covariance of $y(t)$ is governed by: \begin{equation}
\Cov\left(y(t_i),y(t_j)\right)=\Cov\left(f_{z(t_i)}(\ve{x}(t_i)),f_{z(t_j)}(\ve{x}(t_j))\right)+\sigma^2\mathbb{I}(i=j),\label{eq:cov-cross}
\end{equation} where $\mathbb{I}\left(\cdot\right)$ is an indicator variable that equals $1$ if the condition in parentheses is true and $0$ otherwise. The dependencies among response variables are induced by the covariance structure among functions $\{f_m(t)\}_{m=1}^M$. The correlations for data exist not only within the same state but also across different states. We apply a multivariate GP prior for the combined functions $\ve{f}(\ve{x}(t))$, where $\ve{f}(\ve{x}(t))=\left(f_1(\ve{x}(t)),\cdots, f_M(\ve{x}(t))\right)^T$.  {To ensure that the covariance matrix for the multivariate GP is at least non-negative definite, a tractable way is to use a convolution process \citep{shi2020regression,shi2011gaussian}. With the specified smooth kernel as detailed in Appendix \ref{SM-sec:CPGP}, we construct the squared exponential covariance function for $\ve{f}(\ve{x}(t))$ with positive parameters $\vesub{\theta}{f}=\{ \upsilon_{m0},\upsilon_{m1}, \vesub{A}{m0}, \vesub{A}{m1}\}_{m=1}^M$, where $\upsilon$ is the amplitude and $\ve{A}$ characterizes the lengthscale or the spread of the covariance.} For $\vesub{x}{i}=\ve{x}(t_i)$, $\vesub{x}{j}=\ve{x}(t_j)$, the covariance and cross-covariance between these functions are defined as $\Cov \left(f_{m}(\vesub{x}{i}), f_{m}(\vesub{x}{j} )\right)=k_{mm}(\vesub{x}{i},\vesub{x}{j}),\, \Cov \left(f_{m}(\vesub{x}{i}), f_{h}(\vesub{x}{j} )\right)=k_{mh}(\vesub{x}{i},\vesub{x}{j})=k_{hm}(\vesub{x}{j},\vesub{x}{i})$ for $m,h \in \{1,\cdots, M\}$, $h\neq m$, and calculated referring to 
Equation \eqref{SM-eq:cov-cross} in Appendix \ref{SM-sec:CPGP}.

In this study, we primarily focus on cases where different states are associated with different covariance structures. Hence, we consider the scenario in which the mean functions are zero, i.e. $\mu_{z(t)}(\ve{x}(t))=0$. Other mean structures, including concurrent forms of the covariates $\ve{x}(t)$, can also be employed; see e.g., \citet{ramsay2005functional} and \citet{shi2007gaussian}. Moreover, for notational simplicity we assume a univariate covariate $\ve{x}(t)$ and consider it as a rescaled time variable, denoted by $t$, i.e., $\ve{x}(t)=t$. This setting is particularly useful in our analysis of accelerometer signal data as it allows the model to capture systematic time-related structure in the response trajectory. 

The hidden state process $z(t)$ is assumed to follow a CTMM$(\ve{Q},\ve{\pi})$:
\begin{equation}
	z(t) \sim \text{CTMM}(\ve{Q},\ve{\pi}),\ z(t) \in \{1,\cdots, M\}, \label{eq:Model-CTMM_z}
\end{equation} where $\ve{\pi}=(\pi_1,\cdots, \pi_M)^T$ is the initial state probability vector, with $\pi_m=p(z(0)=m)$ and $\sum_{m=1}^M\pi_m=1$. The matrix $\ve{Q}=\{q_{mh}\}_{m,h=1}^{M}\in \Re^{M\times M}$ is the infinitesimal generator of the process, or the transition rate matrix, that satisfies $q_{m h } >0,$ for $h\neq m$, and $	q_{m m}=-\sum_{h \neq m} q_{m h}$. For any state $m,h \in \{1,\cdots, M\}$ and for any time point $t_i$ and time duration $ (t_j-t_i)$, the transition probability from state $m$ at time point $t_i$ to state $h$ at time point $t_j$ is denoted as $p\left(z(t_j)=h\mid z(t_i)=m\right)$. 

Assuming that the Markov process is time-homogeneous and the transition probabilities are differentiable functions at any time point $t$, these probabilities satisfy the following Chapman-Kolmogorov forward equation \citep{karlin1981second}: $\ve{P}^\prime(t)=\ve{P}(t)\ve{Q}.$ Here, $\ve{P}(t)$ represents the matrix with components $\ve{P}_{m\, h}(t)=p\left(z(t)=h\mid z(0)=m\right)$. The off-diagonal components $q_{mh}$ of $\ve{Q}$ are interpreted as the rate of change in transition probabilities over an infinitesimal time interval, leading to 
$q_{m h}=\lim _{t \rightarrow 0} \frac{p(z(t)=h \mid z(0)=m)}{t},$ $h \neq m.$ Thus, from time $0$ to $t$, the time-homogeneous process $z(t)$ has a transition probability matrix given by $\ve{P}(t)=\exp\big[ \ve{Q} t \big]$, which is solved via the matrix exponential of the above forward equation. Then the state transition probability can be expressed as $\ve{P}_{m\, h}(t)=\exp\big[ \ve{Q} t \big]_{m\,h}$, where $\exp[\ \cdot\ ]_{m\,h}$ denotes the $(m,h)$-th entry of the matrix $\exp[\ \cdot\ ]$.

We use Equations \eqref{eq:Model-part1}, \eqref{eq:cov-cross}, and \eqref{eq:Model-CTMM_z} to formally define the functional regression model with varying structures and continuous-time hidden states (FRVS).

\subsection{Bayesian Inference}\label{sect:bayesian_inf}

Suppose that we have $n$ observations summarized as $\mathcal{D}=\{ y_i,\, t_i \}_{i=1}^n$, where $y_i$ denotes the observation at time point $t_i$, such that $y_i=y(t_i)$. The vector of observations is denoted as $\ve{y}=\left(y_1,\cdots, y_n\right)^T$. Accordingly, we denote $z_i=z(t_i)$ and the hidden state sequence as $\ve{z}=\left(z_1,\cdots, z_n\right)^T$, and $\vesub{f}{m}=\left( f_{m}(t_1),\cdots, f_{m}(t_n)\right)^T$. The $nM$-dimensional vector $\fall=\left(\ve{f}(t_1)^{T} ,\cdots, \ve{f}(t_n)^{T} \right)^T$ with $\ve{f}(t)=\left(f_1(t),\cdots, f_M(t)\right)^T$ denotes the collections of the functions $\{f_m(t)\}_{m=1}^M$ at all time points, which follows a multivariate Gaussian distribution with zero mean and covariance matrix $\ve{K}=\{\ve{\mathrm{K}}(t_i,t_j)\}_{i,j=1}^n$. The $(m,h)$-th entry of the block matrix $\ve{\mathrm{K}}(t_i,t_j)$ is $k_{mh}(t_i,t_j)$, $m,h=1,\cdots,M$, calculated using Equations \eqref{eq:cov-cross} based on the parameters $\vesub{\theta}{f}$. We denote $\vesub{f}{\bs{z}}=\left(f_{z_1}(t_1),\cdots, f_{z_n}(t_n)\right)^T$, a realization of $f_{z(t)}(t)$ associated with $\ve{z}$, and $\vesub{f}{\text{-}\bs z}$ as the remaining elements of $\fall$ after $\vesub{f}{\bs z}$ is taken away. This subvector of $\fall$ also follows a multivariate Gaussian distribution with zero mean and covariance matrix $\vesub{K}{\bs{z}}=\{k_{z_i z_j}(t_i,t_j)\}_{i,j=1}^n$, which is the crucial part that captures the underlying curve of the observations while incorporating information from the hidden state process.

The model for observed data can be represented as follows:\begin{equation*}
	\begin{aligned}
		y_i&=f_{z_i}(t_i)+\epsilon_i,\, \epsilon_i\stackrel{\text { i.i.d. }}{\sim} N(0\, ,\, \sigma^2),\, i=1,\cdots, n,\\
		f_{z_i }(t_i )&=f_m(t_i), \text{ when } z_i=m,\ m \in \{1,\cdots, M\},\
		\vesub{f}{\bs{z}}  \sim N\left(\vesub{0}{n} \, ,\, \vesub{K}{\bs{z}}\right).
	\end{aligned}
\end{equation*}
Here, $\vesub{0}{n}$ is the $n$-dimensional zero vector. Given the hidden state sequence $\ve{z}$, the marginal conditional distribution of the observations $\ve{y}$ is also a multivariate Gaussian distribution with zero mean and covariance matrix $\left(\vesub{K}{\bs{z}}+\sigma^2 \vesub{I}{n} \right)$, where $ \vesub{I}{n}$ is the identity matrix of dimension $n$.

Given data $\mathcal{D}$, we assume that the number of hidden states $M$ is known and that the initial state is $z_0=1$ with $\pi_{1}=1$. The primary unknowns of interest for inference in this model include the combined functions under different states across the entire time points $\fall$, the hidden state sequence $\ve{z}$, and the parameters ${\Theta}=\{\sigma^2, \vesub{\theta}{z}, \vesub{\theta}{f} \}$. Here, $\vesub{\theta}{z}=\{q_{m h}\}^M_{m,h=1, h\neq m}$ are the transition rates of the CTMM. The covariance function of the nonlinear functions $\{f_m(\cdot)\}^M_{m=1}$ is governed by the hyper-parameters $\vesub{\theta}{f}$. These hyper-parameters can be predetermined based on prior knowledge or estimated by an empirical Bayesian approach.

In this paper, we adopt a Bayesian paradigm for inference. The joint posterior density of all unknowns can be expressed as: 
\begin{equation}
	p\left(\fall, \ve{z},\Theta \mid\ve{y}\right)\propto 	p\left(\ve{y}\mid\fall,\ve{z},\Theta \right) p\left(\fall\mid\ve{z}, \Theta \right)p\left(\ve{z}\mid\Theta \right)p\left(\Theta\right).\label{eq:joint-posterior-distri}
\end{equation} 

Bayesian inference is performed using the Gibbs sampler \citep{robert1999gibbs,gelfand2000gibbs}, a standard technique for complex hierarchical models. The $r$-th iteration comprises the following steps: Step 1: Draw $\fall^{(r)} \sim p\left (\fall\mid\vesup{z}{(r-1)},\Theta^{(r-1)}, \ve{y}\right)$; Step 2: Draw $\Theta^{(r)}\sim p\left (\Theta \mid\fall^{(r)},\vesup{z}{(r-1)}, \ve{y}\right)$; Step 3: Draw $\vesup{z}{(r)}\sim p\left (\ve{z}\mid\fall^{(r)},\Theta^{(r)},\ve{y} \right)$. After $R_0$ iterations (referred to as ``burn-in"), the random samples $\{\fall^{(r)}, \vesup{z}{(r)},\Theta^{(r)}\}_{r=R_0+1}^R$ are treated as samples approximately distributed according to the joint posterior distribution $p\left(\fall, \ve{z},\Theta \mid\ve{y}\right)$ under mild regularity conditions \citep{robert1999gibbs,meyn2012markov}. For the estimation of the unknowns, the posterior means are calculated from the post-burn-in samples. The conditional probabilities $\hat{p}\left(z_i \mid \fall, \Theta, \ve{y}\right)$, for $i=1,\cdots,n$, are averaged from the posterior samples $\{{z}_i^{(r)}\}_{r=R_0+1}^R$. The estimated state $z_i$ is then determined as $\hat{z}_i=\argmax_m \hat{p}\left(z_i=m \mid \fall, \Theta, \ve{y}\right)$. Other quantities can be determined similarly.
\subsubsection{Conditional Distribution $p\left (\fall\mid\ve{z},\Theta, \ve{y}\right)$}
In our sampling procedure, Step 1 and Step 2 can be easily tracked with analytical conditional distributions. The conditional distribution of $\fall$ given the others can be expressed as: $$p\left(\fall\mid \ve{z},\Theta,\ve{y}\right)\propto p\left(\ve{y}|\fall,\ve{z},\sigma^2\right)p(\fall\mid\ve{z},\vesub{\theta}{f})=p\left(\ve{y}|\vesub{f}{\bs{z}},\ve{z},\sigma^2\right)p(\vesub{f}{\text{-}\bs{z}}\mid\vesub{f}{\bs{z}},\ve{z},\vesub{\theta}{f})p(\vesub{f}{\bs{z}}\mid\ve{z},\vesub{\theta}{f}),$$ where $\fall=\left(\vess{f}{\bs{z}}{T}, \vess{f}{\text{-}\bs{z}}{T}\right)^T$ as defined in the previous section. We can therefore obtain random samples of $\fall$ by sampling $\vesub{f}{\bs{z}}$ and $\vesub{f}{\text{-}\bs{z}}$ separately.

Since the observations $\ve{y}$ conditioned on $\fall$ and $\ve{z}$ follows a multivariate Gaussian distribution, it can be easily shown that the conditional distribution of $\vesub{f}{z}$ is also Gaussian: \begin{equation}
	\left(\vesub{f}{\bs{z}}\mid \ve{z},\Theta,\ve{y}\right) \sim N\left( \vesub{K}{\bs{z}}\left(\vesub{K}{\bs{z}}+\sigma^2\vesub{I}{n}\right)^{-1}\ve{y} \, ,\, \sigma^2 \vesub{K}{\bs{z}} \left(\vesub{K}{\bs{z}}+\sigma^2\vesub{I}{n}\right)^{-1} \right).\label{eq:combine_fz_y}
\end{equation}

After sampling $\vesub{f}{\bs{z}}$, we subsequently draw the remaining part $\vesub{f}{\text{-}\bs{z}}$ from the conditional distribution of $ p(\vesub{f}{\text{-}\bs{z}} \mid\vesub{f}{\bs{z}}, \ve{z},\Theta,\ve{y} )$. Denote the covariance matrix between $\vesub{f}{\bs{z}}$ and $\vesub{f}{\text{-}\bs{z}}$ as $\vess{K}{\bs{z}}{\ast} \in \Re^{n \times (nM-n)}$, and the covariance matrix for $\vesub{f}{\text{-}\bs{z}}$ itself as $\vesub{K}{\text{-}\bs{z}} \in \Re^{(nM-n)\times (nM-n)}$. With a similar derivation procedure, we have $$(\vesub{f}{\text{-}\bs{z}} \mid\vesub{f}{\bs{z}}, \ve{z},\Theta,\ve{y} )=(\vesub{f}{\text{-}\bs{z}} \mid\vesub{f}{\bs{z}}, \ve{z},\Theta )\sim N( \vess{K}{\bs{z}}{\ast T}\vess{K}{\bs{z}}{-1}\vesub{f}{\bs{z}} \, , \vesub{K}{\text{-}\bs{z}}-\vess{K}{\bs{z}}{\ast T} \vess{K}{\bs{z}}{-1}\vess{K}{\bs{z}}{\ast} ).$$ 

\begin{remark}\label{remark_NNGP}
Despite the existence of these analytical forms, directly working with or evaluating the density is computationally expensive for large sample size $n$ as the inverse and determinant of the covariance matrix necessitates $n^3$ floating point operations and requires storage in the order of $n^2$. To mitigate these computational costs, the nearest neighbor GP (NNGP) method is employed; see     Appendix \ref{SM-sec:NNGP} for the details.  
\end{remark}

\subsubsection{Conditional Distribution $p\left (\Theta\mid\fall,\ve{z}, \ve{y}\right)$}
In Step 2, for the noise variance $\sigma^2$, the prior distribution is chosen as a usual conjugate prior, an inverse-Gamma distribution, proportional to $
\left (1/\sigma^2\right )^{\alpha+1}\exp\left(-\beta/\sigma^2\right),$ with $\sigma^2>0$ and two hyper-parameters $\alpha,\ \beta>0$. Consequently, the full conditional distribution of $\sigma^2$ remains an inverse-Gamma distribution with parameters $\left(\alpha +n/2\right)$ and $(\beta +\frac{1}{2}\left\|\ve{y}-\vesub{f}{\bs{z}}\right\|^2)$.

The conditional distributions of $\vesub{\theta}{z}$, along with the hyper-parameters $\vesub{\theta}{f}$, are not analytically tractable. We utilize a random walk Metropolis-Hastings (MH) algorithm \citep{metropolis1953equation,Hastings1970Monte} for these two sets of hyper-parameters, respectively. As an alternative, an empirical Bayes Gibbs sampler can be implemented, where the values of the hyper-parameters are obtained by maximizing the marginal conditional likelihood \citep{casella2001empirical}.
\subsubsection{Conditional Distribution $p\left (\ve{z}\mid\fall,\Theta,\ve{y}\right)$} \label{sec:cond_pz}
Step 3, which involves sampling the state sequence $\ve{z}$, poses challenges. At each time point, the state can take any integer value from 1 to $M$, leading to potentially $M^n$ combinations for $\ve{z}$. Calculating the probabilities for all possible values of the state sequence becomes infeasible when $n$ is large. Moreover, the traditional forward-backward algorithm \citep{Rabiner1986AnIntro,bishop2006pattern} used in the HMMs, cannot be directly applied in this context due to the violation of the conditional independence assumption among the observations $\big(y_1,y_2,\cdots, y_{i-1}\big)$ and $\big( y_{i+1},y_{i+2}$, $\cdots, y_n\big)$ given the current state $z_i$.

Another approach for estimating the hidden state sequence is to sample each state variable $z_i$ based on the previously sampled state sequence $\vesub{z}{\text{-}i}=\ve{z}\backslash z_i$, updating $z_i \sim p\left(z_i \mid \fall, \vesub{z}{\text{-}i}, \Theta,\ve{y} \right)$ from a categorical distribution. While this approach is an effective method for handling high-dimensional sampling, it may converge very slowly when $n$ is large and often yield extreme scenarios where all state variables become identical. For our FRVS model, it is usually stuck at the initial values. To address this challenge, we introduce an innovative curve-based sampling scheme, which will be detailed in Section \ref{sec:CBSS}.  

\section{Curve-Based Sampling Scheme for Updating Hidden States} \label{sec:CBSS}

{The key idea of our novel approach to address the sampling problem discussed in Section \ref{sec:cond_pz} is to introduce a set of augmented continuous variables $\{g_m(t) \}_{m=1}^{M-1}$, which are defined as $g_m(t)= \log p\left(z(t)=m \mid \fall,\Theta,\ve{y} \right)- \log p\left(z(t)=M \mid\fall,\Theta,\ve{y} \right).$

Instead of sampling $z_i$ at each discrete time points (thus leading to intractable computational problems due to complexity of $M^n$), we sample continuous random functional variables $\{g_m(t)\}_{m=1}^{M-1}$. We can then sample the states using the following categorical distribution for $i=1,\cdots, n$ with $\alpha_{im}=\exp\left(g_m(t_i)\right)\left/\left({\sum_{h=1}^{M-1} \exp\left( g_h(t_i)\right)+1 }\right) \right.$: \begin{equation}
		\left( z_i \mid \fall, \Theta, \ve{y}\right)=\left( z_i \mid g_1(t_i),\cdots, g_{M-1}(t_i) \right)\sim \text{Cat}\left(\alpha_{i1},\cdots, \alpha_{i M-1},1-\sum_{m=1}^{M-1} \alpha_{im}\right). \label{eq:z- categorical}
\end{equation}

With these auxiliary functions, we can generate multiple possible state sequences in a curve-based manner, circumventing the direct calculation of probabilities while still obtaining samples that follow the target distribution. The implementation of this strategy is illustrated in Figure \ref{fig:Allsteps}. Hence, the problem of enumerating the state sequence combinations is transformed into estimating their underlying curves of probability. 
\begin{figure}
	\centering
\includegraphics[width=1\linewidth]{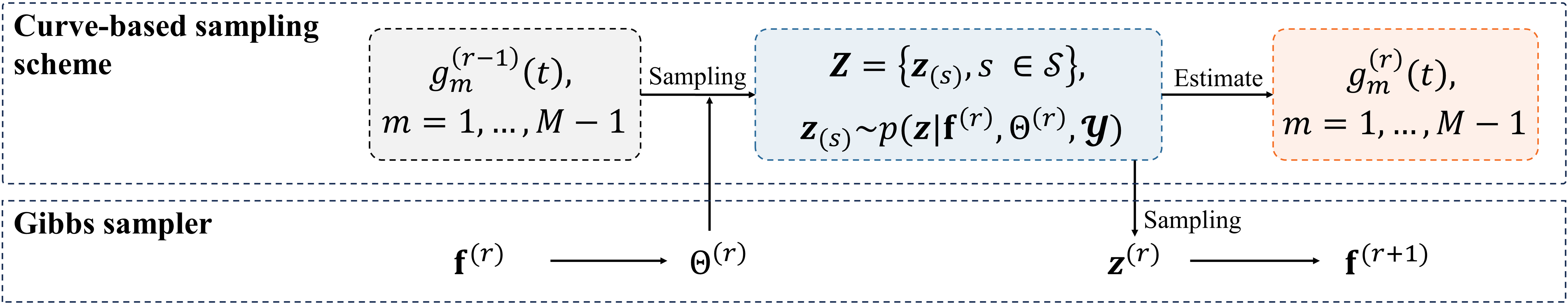}
	\caption{Curve-based sampling scheme for sampling the hidden state sequence. State sequence $\ve{z}$ is updated via augmented curves of probability $\{g_m(t)\}_{m=1}^{M-1}$. }
	\label{fig:Allsteps}
\end{figure}

\subsection{Sampling Procedure for $\ve{z}$ }\label{sec:detailed for Cpd-z} 
The auxiliary continuous functions $\left\{g_m(t)\right\}_{m=1}^{M-1}$ are the key component of our curve-based sampling scheme. For implementation, each of these functions is modeled via a GP prior characterized by a zero mean function and a squared exponential covariance kernel $c_m(\cdot, \cdot)$. Specifically, for $m=1,\cdots, M-1$:\begin{equation}
	\begin{aligned}
		g_m(t) \sim \text{GP}\left( 0, c_m(\cdot,\cdot)\right),\ 	c_m(t_i,t_j)=\gamma_m\exp\left(-\frac{1}{2}\omega_m\left(t_i-t_j\right)^2\right),
	\end{aligned} \label{eq:Cov_g}
\end{equation} where the hyper-parameters for each $g_m(t)$ are denoted as $\vesub{\theta}{gm}=\{\gamma_m,\omega_m\}$.

Concatenating all realizations into one vector, we define the joint vector $\gall=\\ \left(\ve{g}{(t_1)}^{T},\cdots,\ve{g}{(t_n)}^{T} \right)^T \in \Re^{(M-1)n}$ with $\ve{g}(t)=\left(g_1(t),\cdots, g_{M-1}(t)\right)^T$. The vector $\gall$ follows a multivariate Gaussian distribution with a zero mean vector and a block-diagonal covariance matrix $\ve{C}=\left\{\ve{\mathrm{C}}(t_i,t_j)\right\}_{i,j=1}^n$ with each block $\ve{\mathrm{C}}(t_i,t_j)=\diag \left( c_1(t_i,t_j),\cdots, c_{M-1}(t_i,t_j)\right)$. The collective hyper-parameters are denoted as $\vesub{\theta}{g}=\{\vesub{\theta}{gm}\}_{m=1}^{M-1}$.

In our curve-based sampling scheme, both $\gall$ and $\ve{z}$ are updated in Step 3. After the burn-in period, the sample of $\gall$ produced in the previous iteration closely resembles its posterior distribution (in the left panel of     Figure \ref{fig:curve-based}). Similarly, the state sequence $\ve{z}$ generated from this sample of $\ve{g}(t)$ is also very close to its posterior distribution (in the middle panel of     Figure \ref{fig:curve-based}), making it easier to accept in the MH algorithm. We will discuss MH algorithm below, and have the discussion on how to update $\gall$ and the related hyper-parameters to the next subsection. 

The target distribution $p\left(\ve{z}\mid\fall,\Theta,\ve{y} \right)$ is proportional to $\mathcal{P}(\ve{z})= p\left(\ve{y}\mid\vesub{f}{\bs{z}},\ve{z},\sigma^{2}\right)p\left(\ve{z}\mid\vesub{\theta}{z}\right)$. Thus an independent MH algorithm \citep{robert1999gibbs,liu2001monte} can be applied. An easy-to-simulate proposal density is $\mathcal{Q}(\ve{z})=p\big(\ve{z}\mid \hat{\gall}\big)$, as given in Equation \eqref{eq:z- categorical}, where $\hat{\gall}$ is the estimated value of $\gall$ from the previous iteration.

Therefore, for $s=1,\cdots, S$, the probability of accepting a new sample $\vesup{z}{\ast}\sim \mathcal{Q}(\ve{z})$ given the previous sample $\vesub{z}{(s-1)}$ is $\min\left\{ 1, \frac{\mathcal{P}\big(\vesup{z}{\ast}\big)\mathcal{Q}\big(\vesub{z}{(s-1)}\big)}{\mathcal{P}\big(\vesub{z}{(s-1)}\big)\mathcal{Q}\big(\vesup{z}{\ast}\big)} \right\}.$ The chain converges to the target distribution as $S\rightarrow +\infty$. To reduce the correlations among the samples in $\ve{Z}$, we could select a subset from $\{1,\cdots, S\}$, which is denoted as $\mathcal{S}$ with size $\mathbb{S} \ll S$. In practice, parallel processing can be used. 
\begin{remark}
 {The proposal set size $\mathbb{S}$ determines the number of candidate state sequences generated in the curve-based sampling scheme and thus affects the trade-off between computational efficiency and sampling performance. In practice, we find that even relatively small values of $\mathbb{S}$ are sufficient to ensure stable performance. As demonstrated in the simulation studies, the results obtained with $\mathbb{S}=10$ and $\mathbb{S}=20$ are very similar in terms of state classification accuracy and model fitting. Therefore, selecting $\mathbb{S}$ in a moderate range (e.g., $\mathbb{S}=10–20$) typically provides a good balance between computational efficiency and sampling effectiveness.}
\end{remark}
\begin{remark}
	The ensemble MCMC (EnMCMC) algorithm is a method related to the multiple try Metropolis to provide a better exploration of the sample space. The main idea involves considering various potential candidates at each iteration and larger jumps are facilitated. When the ensemble size is large enough, the probability of accepting a new sample approaches $1$ and the correlation among the generated state vanishes \citep{liu2000multiple}. The details of the algorithm and the simulation comparisons are presented in     Appendix \ref{SM-sec:EnMCMC}. Our curve-based sampling scheme yields nearly identical results to EnMCMC requiring less computational time. This is probably because we updated both $\ve{g}(t)$ and $\ve{z}$ in our algorithm.  
\end{remark}
\subsection{Update $\gall$ and $\vesub{\theta}{g}$}\label{sec:update thetag and g}

We assume that $S$ samples of state sequences $\ve{Z}=\{\vesub{z}{(s)}\}_{s=1}^S$ are obtained. The hyper-parameters $\vesub{\theta}{g}$ are estimated by maximizing the marginal density function $	p\left(\ve{Z} \mid\vesub{\theta}{g} \right)$, where
\begin{align*}
	p\left(\ve{Z} \mid\vesub{\theta}{g} \right)&=\int p\left(\ve{Z} \mid\gall \right)p\left(\gall |\vesub{\theta}{g} \right) d \gall=\int \prod_{s\in\mathcal{S}} p\left(\vesub{z}{(s)} \mid\gall \right)p\left(\gall |\vesub{\theta}{g} \right) d\gall\\
	&=\int \prod_{s\in\mathcal{S}} \left(\prod_{i=1}^{n}p\left(z_{(s)_i} \mid\ve{g}(t_i) \right) \right)\left(2\pi\right)^{-\frac{(M-1)n}{2}}|\ve{C}|^{-\frac{1}{2}}\exp\left\{ -\frac{1}{2}\gall^{T} \vesup{C}{-1}\gall\right\} d \gall
\end{align*} with $$
	p\left(z_{(s)_i} \mid\ve{g}(t_i) \right)=\frac{\exp\left( g_{z_{(s)_i}}(t_i)\mathbb{I}\left(z_{(s)_i} \neq M\right) \right)}{\sum_{h=1}^{M-1}\exp\left( g_h(t_i)\right)+1}.$$

However, the integral involved in the above marginal density is analytically intractable. An alternative, as described in \citet{wang2014generalized}, is to approximate $p\left(\ve{Z} \mid\vesub{\theta}{g} \right)$ by $\tilde{p}\left(\ve{Z} \mid\vesub{\theta}{g} \right)\triangleq \left.\frac{p\left(\ve{Z},\gall \mid\vesub{\theta}{g} \right)}{	\tilde{p}_G\left(\gall \mid\ve{Z},\vesub{\theta}{g} \right)}\right|_{\gall=\tilde{\gall}},$ where $\tilde{p}_G\left(\gall \mid\ve{Z},\vesub{\theta}{g} \right)$ is the Gaussian approximation to the conditional density $p\left(\gall \mid\ve{Z},\vesub{\theta}{g}\right)$, and $\tilde{\gall}$ is the mode of the conditional density of $\gall$, also being the estimation for $\gall$. Here, $$ p\left(\ve{Z},\gall\mid \vesub{\theta}{g}\right)=p\left(\ve{Z} \mid\gall \right)p\left(\gall |\vesub{\theta}{g} \right)
=\exp\left\{ \sum_{s\in\mathcal{S}} \sum_{i=1}^n \log p\left(z_{(s)_i} \mid\ve{g}(t_i) \right) + \log p\left(\gall |\vesub{\theta}{g} \right) \right\}, $$ and $\log p\left(z_{(s)_i} \mid\ve{g}(t_i) \right)$ is approximated by a Taylor expansion to the second order around $\tve{g}(t_i)=\big(\tilde{g}_{1}(t_i),\cdots, \tilde{g}_{M-1}(t_i)\big)^T$. Denoting $\ve{d}$ and $\ve{D}$ as the collection of first and second order derivative of $\ell_{z_{(s)_i}}\left(\ve{g}(t_i)\right)$ with respect to $\ve{g}(t_i)=\tve{g}(t_i)$, for $i=1,\cdots,n$,
the probability density function $p\left( \gall\mid \ve{Z}, \vesub{\theta}{g}\right)$ is proportional to $p\left(\ve{Z},\gall\mid \vesub{\theta}{g}\right)$, given by \begin{align*}
	p\left( \gall\mid \ve{Z}, \vesub{\theta}{g}\right)	\propto p\left(\ve{Z},\gall\mid \vesub{\theta}{g}\right)&\propto \exp\left\{\ve{d}\left(\gall-\tilde{\gall}\right) -\frac{1}{2}\left(\gall-\tilde{\gall}\right)^T \ve{D}\left(\gall-\tilde{\gall}\right) -\frac{1}{2}\gall^{T} \vesup{C}{-1} \gall\right\}
	\\ &\propto \exp\left\{ -\frac{1}{2}\gall^{T} \left(\vesup{C}{-1}+ \ve{D} \right)\gall+\left(\ve{d}+\tilde{\gall}^{T}\ve{D} \right) \gall \right\}. 
\end{align*} The detailed calculation is presented in     Appendix \ref{SM-sec:predict_new_z_g}. 

Given $\vesub{\theta}{g}$, the mode $\tilde{\gall}$ is obtained using the Fisher scoring algorithm (detailed in     Algorithm \ref{alg::thetag_g-estimation} in Appendix \ref{SM-sec:predict_new_z_g}). The Gaussian approximation $\tilde{p}_G\left( \gall \mid\ve{Z},\vesub{\theta}{g} \right)$ is the density function of the multivariate Gaussian distribution $N\left(\tilde{\gall}\, ,\, \left(\vesup{C}{-1}+\ve{D} \right)^{-1}\right)$. Then $\vesub{\theta}{g}$ is updated by maximizing the approximated marginal likelihood $\tilde{p}\left(\ve{Z}\mid \vesub{\theta}{g}\right)$, where $$
\tilde{p}\left(\ve{Z}\mid \vesub{\theta}{g}\right) \propto \left|\ve{C}\right|^{-\frac{1}{2}} \left|\vesup{C}{-1}+\ve{D}\right|^{-\frac{1}{2}}\exp\left\{- \frac{1}{2}\tilde{\gall}^{T} \vesup{C}{-1}\tilde{\gall} \right\}.
$$

In practice, estimating $\gall$ becomes computationally challenging when the sample size $n$ is large, due to the time complexity of $\mathcal{O}(n^3)$ in Algorithm \ref{alg::thetag_g-estimation}. To address the computational demand associated with this large dataset, we propose selecting a considerably smaller subset of time points $\bve{t}$ from the full set $\{t_i\}_{i=1}^n$ with $\underline{n}$ ($\underline{n}\ll n$) time points. The sampled set of state sequences $\ve{Z}$ at these tunable time points, denoted as $\bve{Z}$, will be utilized to estimate hyper-parameters $\vesub{\theta}{g}$ and the corresponding $\underline{\gall}$ at the selected time points. Subsequently, the powerful predictive capabilities of GP modeling will be leveraged to predict the functions at the remaining time points. For any time point $t^\ast$ that are not included in $\bve{t}$, $p\left(z(t^\ast) \mid\bve{Z}, \vesub{\theta}{g}\right)$ or the corresponding $\ve{g}(t^\ast)$ can be estimated. The details are given in Appendix \ref{SM-sec:predict_new_z_g}.

\section{Theoretical Results}\label{sec:Theory}

In this section, we first show that the model defined in Section \ref{sec:The Model} is identifiable under some conditions. Following that, we consider the consistency of the model structure which involves three crucial aspects: the estimated parameters $\sigma^2$ and $\vesub{\theta}{ z}$, the predicted responses $\hve{y}$, and the estimated hidden state sequence $\hve{z}$ by the curve-based sampling scheme. All the proofs are provided in the appendices.
\subsection{Model Identifiability}

Identifiability is an important problem for most HMMs. \citet{gassiat2016nonpara} considers the identifiability of nonparametric finite translation HMMs; \citet{gassiat2016inference} and \citet{alexandrovich2016nonparametric} prove the identifiability of finite state space nonparametric HMMs; \citet{gassiat2020identifiability} studies the identifiability of nonparametric translation HMMs with general state space. 
Below we show that the model defined in Section \ref{sec:The Model} is identifiable under some conditions. To the best of our knowledge, the identifiability of nonparametric models with continuous-time HMM has not been considered in the literature.
\begin{theorem}
	Assume that the number of states $M$ is known and the transition rate matrix $\ve{Q}$ is irreducible and aperiodic. If the functions $\{f_m(\cdot)\}_{m=1}^M$ are differentiable and with the support $\mathcal{T} \subset \Re$, satisfying the following condition: 
	for any two functions $f_m(t)$ and $f_h(t)$, where $m,h=1,\cdots,M$ and $h\neq m$, 
	\begin{equation}
		\Vert f_m(t)-f_h(t)\Vert^2 + \Vert f'_m(t)-f'_h(t)\Vert^2 \not= 0, \text{ for any } t\in\mathcal{T}, \label{Con_transversal}
	\end{equation}
	then the FRVS model defined in Section \ref{sec:The Model} is identifiable, up to label swapping of the hidden states. 
\end{theorem}
The proof of the theorem is given in Appendix \ref{SM-sec:Identifiability}.

\subsection{Consistency of Parameters $\sigma^2$ and $\vesub{\theta}{z}$}

The prior distribution for the noise variance $\sigma^2$ is specified as an inverse-Gamma distribution with hyper-parameters $\alpha$ and $\beta$. As the sample size $n$ increases, the posterior mean of $\sigma^2$ converges in probability to the true value $\sigma_0^2$, demonstrating asymptotic consistency. 
\begin{theorem}
	Under the inverse-Gamma prior for $\sigma^2$ with hyper-parameters $\alpha$ and $\beta$, the posterior mean $\E\left(\sigma^2 \mid \ve{y}\right)$ converges in probability to the true value $\sigma_0^2$ as $n\rightarrow +\infty$. Specifically, $\E\left(\sigma^2 \mid \ve{y}\right)\stackrel{p}{\rightarrow} \frac{n \sigma_0^2}{n+2\alpha-2} + \frac{2\beta}{n+2\alpha-2}\rightarrow \sigma_0^2$ as $n\rightarrow +\infty$. \end{theorem} The detailed derivation is provided in Appendix \ref{SM-sec:Consistency-sigma}.

For the parameter $\vesub{\theta}{z}=\{q_{mh}\}_{m,h=1,m\neq h}^M$, no informative prior distribution is specified. Hence we consider the maximum likelihood estimation based on the logarithm of the marginal distribution $p\left(\ve{y}\mid\vesub{\theta}{z} \right).$ Given the number of states $M$ and the initial state $z_0=1$ with an initial state probability $\pi_1=1$ at time point $t_0=0$, we have that $$p\left(\ve{y}\mid \vesub{\theta}{z}\right)=\sum_{\ve{z}} \left\{ p\left(\ve{y}\mid \ve{z}\right)\prod_{i=1}^n \exp\big[\ve{Q} (t_i-t_{i-1}) \big]_{z_{i-1}\, z_i} \right\}. $$ For the summation of each row in the transition probability matrix $\exp\big[\ve{Q} t \big]$ equals 1, we can express the marginal likelihood as \begin{align*}
	p\left(\ve{y}\mid \vesub{\theta}{z}\right)=\sum_{\ve{z}} \left\{ p\left(\ve{y}\mid \ve{z}\right)\left(\prod_{i=1}^n \vess{e}{z_{i-1}}{T} \exp\big[\ve{Q} (t_i-t_{i-1}) \big]\vesub{e}{z_i} \right) \underbrace{\vess{e}{z_n}{T} \exp\big[\ve{Q} (1-t_{n}) \big]\vesub{1}{M} }_{= 1} \right\},
\end{align*} where $\vesub{e}{m}$ is an $M$-dimensional column vector with the $m$-th component equal to one and all the others zero, and $\vesub{1}{M}$ is an $M$-dimensional column vector of ones.
\begin{theorem}
	Let $\vesub{\theta}{z_0}$ denote the true transition parameters. Under conditions (C1) and (C2) in     Appendix \ref{SM-sec:Consistency-theta_z}, the marginal likelihood estimator $\vesub{\hat\theta}{z}$ satisfies $\vesub{\hat\theta}{z} \stackrel{p}{\rightarrow}\vesub{\theta}{z_0}$ as $n \rightarrow +\infty$. 
\end{theorem} The proof is given in Appendix \ref{SM-sec:Consistency-theta_z}.

\subsection{Information Consistency of the Prediction $\hve{y}$}
We examine the information consistency of the prediction $\hve{y}=\left(\hat{y}_1,\cdots, \hat{y}_n\right)^T$, in relation to the actual observations $\ve{y}$. Given $f_z(t)$, the observations $y_i$ ($i=1,\cdots,n$) are conditionally independent and follow a Gaussian distribution with mean $f_{z_i}(t_i)$. The combined function $\ve{f}(t)=(f_1(t),\cdots, f_M(t))^T$ adheres to the convolution GP framework as discussed in Section \ref{sec:The Model}. Given the variables $\fall$ and the hidden state sequence $\ve{z}$, the expected value of the observations is given by $\E\left( \ve{y}\mid \fall,\ve{z}\right)=\vesub{f}{\bs z}$. 

Let $\hvesub{\theta}{f}$ denote the empirical Bayesian estimator of the hyper-parameters $\vesub{\theta}{f}$ pertaining to the covariance structure. Let $\vesub{f}{0}(\cdot)=\left(f_{10}(\cdot),\cdots, f_{M0}(\cdot)\right)^T$ be the true underlying function and $\vesub{f}{\bs{z} 0}=\left(f_{z_1 0}(t_1),\cdots, f_{z_n 0}(t_n) \right)^T$ represent its realizations for the given $\{t_i\}_{i=1}^n$. Then the true mean of each observation $y_i$ is given by $f_{z_i 0}(t_i)$. Define $p_{cgp}\left(\ve{y}\mid \ve{z}\right)=\int{p\left(\ve{y}\mid \vesub{f}{\bs z},\ve{z}\right)p(\vesub{f}{\bs z} \mid \ve{z} )d\vesub{f}{\bs z} } $ and $p_0\left(\ve{y}\mid \ve{z}\right)=p\left(\ve{y}\mid\vesub{f}{\bs z 0}, \ve{z}\right) $, where $p_{cgp}\left(\ve{y}\mid \ve{z}\right)$ is the Bayesian predictive distribution of $\ve{y}$ based on the convolution GP model. We say that $p_{cgp}\left(\ve{y}\mid \ve{z}\right)$ achieves information consistency if $\frac{1}{n}\E_{\bs{t}}\left(\text{KL}\left[p_0\left(\ve{y}\mid \ve{z}\right), p_{cgp}\left(\ve{y}\mid \ve{z}\right)\right]\right)\rightarrow 0 \text{ as } n\rightarrow +\infty,$ where $\E_{\bs{t}}$ denotes the expectation under the distribution of $\{t_i\}_{i=1}^n$ and $\text{KL}\left[p_0\left(\ve{y}\mid \ve{z}\right), p_{cgp}\left(\ve{y}\mid \ve{z}\right)\right]$ is the Kullback-Leibler divergence between $p_0\left(\cdot\mid \ve{z}\right)$ and $ p_{cgp}\left(\cdot\mid \ve{z}\right)$, defined as follows:
\begin{equation*}\text{KL}\left[p_0\left(\ve{y}\mid \ve{z}\right), p_{cgp}\left(\ve{y}\mid \ve{z}\right)\right]=\int{p_0\left(\ve{y}\mid \ve{z}\right)\log\frac{p_0\left(\ve{y}\mid \ve{z}\right)}{ p_{cgp}\left(\ve{y}\mid \ve{z}\right)} }d \ve{y}.\end{equation*}
\begin{theorem}\label{consist_y}
	Under the model defined in Equation \eqref{eq:Model-part1} with specified $\mu_{z(t)}(\ve{x}(t))=0$ and the covariate $\ve{x}(t)$ being the time point $t$, and assuming that $\hvesub{\theta}{f}\rightarrow \vesub{\theta}{f}$ almost surely as $n \rightarrow +\infty$, the predictions $\hve{y}$ are information consistent if the RKHS norm $\|\vesub{f}{0} \|_{\bs{K}}^2$ is bounded and the expected regret term $\E_{\bs{t}}\left( \log|(\delta+ \frac{1}{2 {\sigma}^2}) \vesub{K}{\bs{z}}+ \vesub{I}{n} | \right)=o(n)$. 
\end{theorem} 
The proof is given in Appendix \ref{SM-sec:Consistency-y}. The error bound is specified in Equation \eqref{eq:bound} there.
\begin{remark}
 {Theorem \ref{consist_y} assumes the consistency of the empirical Bayesian estimator of the hyper-parameters. This assumption is commonly adopted in Gaussian process regression when the hyper-parameters are estimated via marginal likelihood maximization. Such consistency results typically require that the covariance kernel is correctly specified and that the true underlying function belongs to the RKHS induced by the kernel. In addition, the observation points should become sufficiently dense in the domain as the sample size increases. Under these conditions, the marginal likelihood estimator of the kernel hyper-parameters can be shown to be consistent; see, for example, \citet{Yi2011Penalized} and \citet{wang2014generalized}. }
\end{remark}

\subsection{Information Consistency of the Estimated Hidden State Sequence $\hve{z}$}
We now show that the estimated hidden state sequence $\hve{z}=\left( \hat{z}_1,\cdots, \hat{z}_n\right)^T$ by the curve-based sampling scheme is information consistent as $n\rightarrow\infty$. This is equivalent to demonstrating the consistency of $\hve{g}(t)=\left( \hat{g}_1(t),\cdots, \hat{g}_{M-1}(t)\right)^T$.

Let $\hvesub{\theta}{g}$ denote the empirical Bayesian estimator of the hyper-parameters $\vesub{\theta}{g}$ related to the covariance functions defined in Equation \eqref{eq:Cov_g}. Let $\vesub{g}{ 0}(\cdot)=\left( {g}_{10}(\cdot),\cdots, g_{M\text{-}1, 0}(\cdot)\right)^T$ be the true functions that capture the underlying probabilities for the hidden state process and $\gall_{0}$ represent their realization values at $\{t_i\}_{i=1}^n$, such that $\gall_{0}=\left( \vesub{g}{0}(t_1)^T,\cdots, \vesub{g}{0}(t_n)^T\right)^T$, where $\vesub{g}{0}(t_i)=\left( g_{10}(t_i),\cdots, g_{M\text{-}1, 0}(t_i)\right)^T$. The true probability for the hidden state $z_i$ is given by $\frac{\exp\left( g_{z_i 0}(t_i)\mathbb{I}\left(z_i \neq M\right) \right)}{\sum_{m=1}^{M-1}\exp\left( g_{m0}(t_i)\right)+1}$. Define ${p}_{gp}(\ve{z})=\int p\left( \ve{z}\mid \gall\right){p}(\gall) d \gall$ and $p_0\left( \ve{z}\right)=p\left( \ve{z}\mid \gall_{0}\right)$, where ${p}_{gp}(\ve{z})$ is the Bayesian predictive distribution of $\ve{z}$ based on the GP model. Then ${p}_{gp}(\ve{z})$ achieves information consistency if $\frac{1}{n}\E_{\bs{t}}\left( \text{KL}\left[ p_{0}(\ve{z}), p_{gp}(\ve{z})\right]\right) \rightarrow 0$ as $n\rightarrow +\infty$.
\begin{theorem} \label{consist_z}
	Under the GP model assumption for $\ve{g}(t)$, assuming that $\hvesub{\theta}{g}\rightarrow \vesub{\theta}{g}$ almost surely as $n \rightarrow +\infty$ and the proposal set size $\mathbb{S}$ is sufficiently large, the prediction $\hve{z}$ obtained by the curve-based sampling scheme is information consistent if the RKHS norm $\|\vesub{g}{0}\|_{\bs C}^2$ is bounded and the expected regret term $\E_{\ve{t}}\left( \log\left|\vesub{I}{n(M\text{-}1)}+\mathbb{S}\ve{C} \right| \right) =o(n)$. 
\end{theorem} The proof is given in Appendix \ref{SM-sec:Consistency-z}. The error bound is specified in Equation \eqref{eq:bound-z} there.
\begin{remark}
	The regret terms $\log|(\delta+ \frac{1}{2 {\sigma}^2}) \vesub{K}{\bs{z}}+ \vesub{I}{n} | $
	and $\log\left|\vesub{I}{n(M\text{-}1)}+\mathbb{S}\ve{C} \right|$
	in Theorem \ref{consist_y} and Theorem \ref{consist_z}
	depend on the covariance functions specified in the GP priors and the distribution of $\ve{t}$. It can be shown that for some widely used covariance functions, such as linear, squared exponential and Mat\'{e}rn class, the expected regret terms are of order $o(n)$; see \citet{seeger2008information} for the detailed discussion. 
\end{remark}

\section{Simulation Studies} \label{sec:Simu}

In this section, we present a comprehensive evaluation of the proposed method through simulation studies. Each scenario is designed to assess different aspects of the method's performance, including the hidden state sequence estimation, the recovery of the underlying functional data and the fitting of the observations.  {Several additional scenarios, including the cases of larger noise variance (Appendix \ref{SM-sec:Scenario-largersigma20.02}), multiple state transitions (Appendix \ref{SM-sec:Scenario1-J4}), state space misspecification (Appendix \ref{SM-sec:Scenario-misspec})  and larger sample size (Appendix \ref{SM-sec:Scenario2-n150}), are presented in Appendix \ref{SM-sec:Additional-simus}.} For all methods, we used $1,000$ posterior draws after discarding $4,000$ burn-in samples. We employed parallel processing across $\mathbb{S}$ cores to generate the underlying probability curves $\ve{g}(\cdot)$. The Gelman-Rubin statistics for the parameters were calculated across five chains and were observed to be less than $1.1$, indicating that the MCMC process converged \citep{gelman2007data}.

For the hidden state sequence estimation, we compare the proposed approach with the naive Gibbs sampling (NGS) across different data fluctuation conditions in Scenarios 1 and 2. This comparison assesses the performance of our method under varying levels of data volatility. The NGS refers to the approach where each $z_i$ is sampled based on the previously sampled state sequence $\vesub{z}{\text{-}i}=\ve{z}\backslash z_i$, i.e., $z_i \sim p\left(z_i \mid \fall, \vesub{z}{\text{-}i}, \Theta,\ve{y} \right)$, as discussed in Section \ref{sec:cond_pz}. In Scenario 3, we address the challenge of handling a very large sample size by incorporating the NNGP approach to evaluate the method's scalability and efficiency.

\subsection{Scenario 1}
In Scenario 1, we consider two states, $M=2$. The hidden state process is set as $z(t)=1$ when $t \in[0,0.3) \cup[0.7,1)$, and $z(t)=2$ when $t \in[0.3,0.7).$ The time points $\{t_i\}_{i=1}^n$ are equally spaced in the range $(0,1)$. The realizations for the underlying functions $\{f_m(t),m=1,2\}$ at these time points were generated from a convolution GP with hyper-parameters $\left\{ \upsilon_{10}=0.1,\upsilon_{11}=0.1, \vesub{A}{10}=1, \vesub{A}{11}=0.1\right\}$ and $\left\{ \upsilon_{20}=0.1,\upsilon_{21}=0.5, \vesub{A}{20}=1, \vesub{A}{21}=1\right\}$, respectively. With these parameter settings, $f_1(t)$ exhibits a lower intensity of fluctuations with stronger serial correlations compared to $f_2(t)$, while the curves also exhibit some between correlation. The noise terms $\{\epsilon_i\}_{i=1}^n$ were independently drawn from a Gaussian distribution $N(0,\, 0.01)$, yielding the final observations $\ve{y}$. One realization of this process is depicted in Figure \ref{fig:Scenario1-fz-obs}. By using two sample sizes, $30$ and $60$ for $n$, we tested the effectiveness of the proposed curve-based sampling scheme (CBSS) method employing the independent MH algorithm against the NGS approach. The number of state sequences in $\ve{Z}$, i.e., the dimension of $\mathcal{S}$, denoted as $\mathbb{S}$, was set to $10$ and $20$. 
The accuracy of the unknown variables ($\ve{z}$ and $\vesub{f}{\bs{z}}$) was evaluated and compared.
\begin{figure}
	\centering
	\subfloat[Scenario 1.]{	
		\includegraphics[width=1\linewidth]{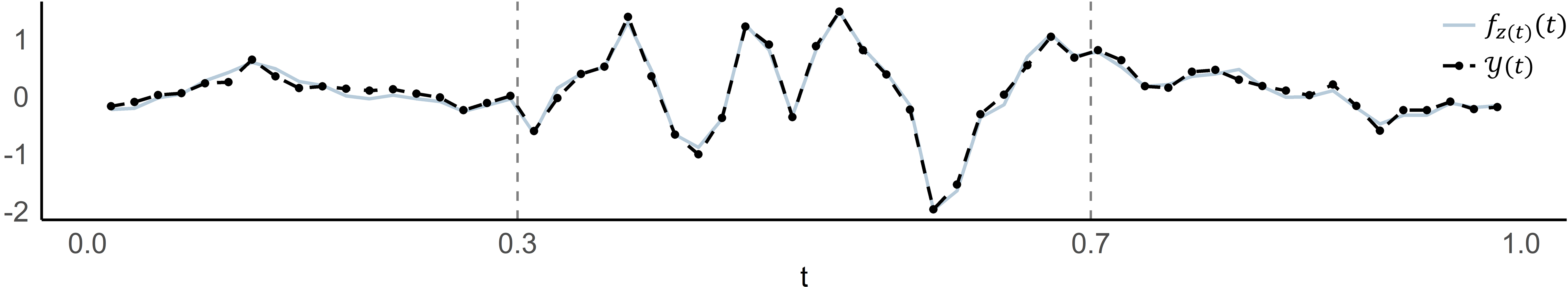}
		\label{fig:Scenario1-fz-obs}
	}\\
	\subfloat[Scenario 2.]{
		\includegraphics[width=1\linewidth]{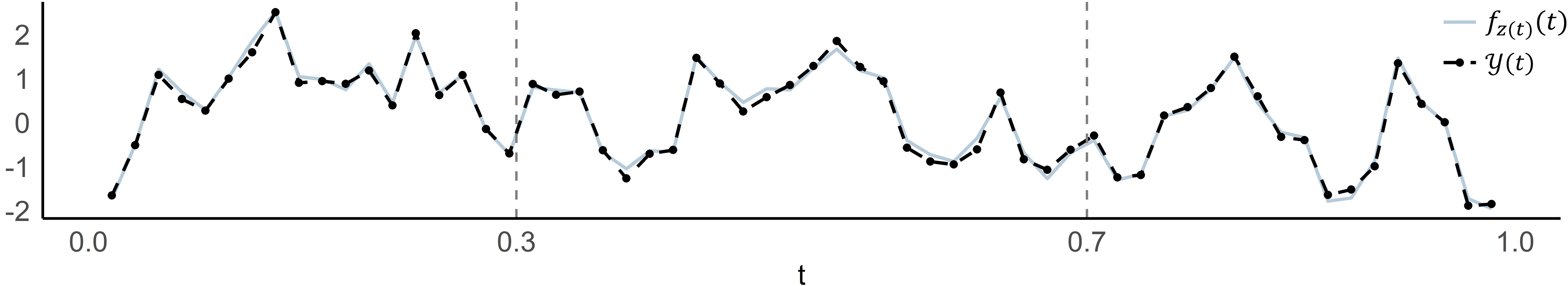}
			\label{fig:scenario2-obs2}}
\caption{The plot for one example of the observations and the underlying functional data, where the hidden state is is $2$ for the data in the range of $[0.3,0.7)$ and $1$ otherwise. (a) Scenario 1.  The variation in the area of $[0.3,0.7)$ is greater than in the other areas. (b) Scenario 2. The amplitudes of the variations are the same, only the dependency between the observations in the range of $[0.3,0.7)$ is greater than in other areas.}
\end{figure}

Table \ref{table:Scenario12-sim} presents the classification performance for estimating the state sequence, goodness of fit measures and computational cost based on fifty replications.  {The average values of the standard classification metrics (accuracy, kappa, precision, specificity, and f1 score) were calculated.} The proposed FRVS model and CBSS method overall outperforms the NGS approach. The difference between the estimated $\ve{f}_{\bs{z}}$ and the observations $\ve{y}$ was assessed using the root mean square error (rmse) and mean absolute error (mae) as measures of goodness-of-fit. The average rmse and mae values based on fifty replications are displayed. It is evident that the fitting performance of CBSS method surpasses that of NGS, exhibiting smaller average rmse and mae values. As expected, performance improves as the sample size increases. 
\begin{table}
	\caption{Comparative analysis of state sequence estimation performance and goodness-of-fit measures and computational cost between BinSeg, CBSS and NGS methods for Scenario 1 and 2. Average metrics are reported based on fifty replications, with varying sample sizes ($n=30, 60$) and proposal set sizes ($\mathbb{S}=10, 20$).}
	\label{table:Scenario12-sim}
	\centering 
    \footnotesize
	\begin{tabular}{llcccccccc}
		\toprule
		& &\multicolumn{4}{c}{$n=30$} &\multicolumn{4}{c}{$n=60$}\\ 
        \cmidrule(r){3-6}\cmidrule(l){7-10}
		& & \multirow{2}{*}{ {BinSeg}} &\multirow{2}{*}{NGS} &CBSS  & CBSS & \multirow{2}{*}{ {BinSeg}}  & \multirow{2}{*}{NGS} &CBSS  & CBSS  \\    
        & & & &($\mathbb{S}=10$) &  ($\mathbb{S}=20$) & & & ($\mathbb{S}=10$) & ($\mathbb{S}=20$) \\   
        \midrule
        \multicolumn{2}{l}{Scenario 1} & \\ \cmidrule{1-2}
		&\multicolumn{9}{l}{\textit{State sequence estimation performance}} \\
		&Accuracy & $0.7400$ & $0.7467$ & $0.9487$ & $0.9420$ & $0.7813$ & $0.8180$ & $0.9717$ & $0.9770$ \\
		&Kappa & $0.4247$ & $0.4816$ & $0.8936$ & $0.8789$ &$0.5132$ & $0.5885$ & $0.9412$ & $0.9522$ \\
		&Precision & $0.8022$ & $0.8546$ & $0.9704$ & $0.9566$ & $0.8185$ & $0.8268$ & $0.9828$ & $0.9839$ \\
		&Specificity & $0.5917$ & $0.7267$ & $0.9517$ & $0.9300$ & $0.6358$ & $0.6400$ & $0.9733$ & $0.9750$ \\
	   &F1 & $0.7951$ & $0.7556$ & $0.9762$ & $0.9517$ & $0.8322$ &$0.8616$ & $0.9782$ & $0.9807$ \\ 
	&	\multicolumn{9}{l}{\textit{Goodness-of-fit}} \\
	&	Rmse & - & $0.1231$ & $0.0700$ & $0.0788$ & -& $0.1792$ &$0.0739$ & $0.0728$ \\
	&	Mae  & - & $0.0850$ & $0.0530$ & $0.0591$ & - & $0.0971$ &$0.0546$ & $0.0541$ \\ 
    &	\multicolumn{9}{l}{\textit{ {Computational cost (min)}}} \\
	&	& $7.1$e-6 & $0.1030$ & $0.1406$ & $0.1992$ &$8.2$e-6 & $0.2855$ &$0.3602$ & $0.4690$ \\
    \cmidrule{1-10}
		\multicolumn{2}{l}{Scenario 2} & \\ \cmidrule{1-2}
	&	\multicolumn{9}{l}{\textit{State sequence estimation performance}} \\
	&Accuracy &$0.6160$ & $0.6480$ & $0.9380$ & $0.9473$ & $0.6180$ & $0.7070$ & $0.9730$ & $0.9743$ \\
	&Kappa & $0.1097$ & $0.2851$ & $0.8728$ & $0.8919$ & $0.1199$ & $0.3706$ & $0.9440$ & $0.9469$ \\
	&Precision & $0.6681$ & $0.8304$ & $0.9692$ & $0.9746$ & $0.6915$& $0.8354$ & $0.9818$ & $0.9860$ \\
	&Specificity& $0.2717$ & $0.6617$ & $0.9517$ & $0.9600$ & $0.3050$& $0.6183$ & $0.9717$ & $0.9783$ \\
	&F1 &$0.7258$ & $0.6669$ & $0.9464$ & $0.9545$ & $0.7159$ &$0.7272$ & $0.9773$ & $0.9783$ \\
	&\multicolumn{9}{l}{\textit{Goodness-of-fit}} \\
	&Rmse & - &$0.4606$ &$0.3460$ & $0.2994$  & -& $0.4292$ & $0.3085$ & $0.2980$ \\
	&Mae & - &$0.3419$ &$0.2793$ & $0.2362$  & -& $0.3172$ & $0.2434$ & $0.2359$ \\
        &	\multicolumn{9}{l}{\textit{ {Computational cost (min)}}} \\
	&	& $7.4$e-6 & $0.1041$ & $0.1423$ & $0.2080$ & $8.4$e-6 & $0.2822$ &$0.3665$ & $0.4723$ \\
    \bottomrule
	\end{tabular}
\end{table}

 {We also compare the performance of the proposed method with a baseline change point detection approach based on Binary Segmentation (BinSeg) \citep{scott1974cluster, fryzlewicz2014wild}. The same data is analyzed and change points in mean and variance are identified using BinSeg with a normal test statistic and manual penalty, allowing up to 10 segments as implemented in the R package \texttt{changepoint} (version 2.2.4). From the analysis, the two primary change points are extracted for temporal regime analysis. The comparative results indicate that the proposed method achieves substantially improved hidden-state recovery than the BinSeg in these settings, highlighting its effectiveness in capturing the underlying structural dynamics of the data.}

Additionally, we consider a more complex setting for hidden state sequences with four transitions, where $z(t)=1$ for $t \in [0,0.3) \cup [0.4,0.7)\cup[0.8,1)$, and $z(t)=2$ for $t \in [0.3,0.4)\cup[0.7,0.8)$. The comparative results between the CBSS method and the NGS approach exhibit similar patterns to those observed with only two transitions. The detailed results are provided in Appendix \ref{SM-sec:Scenario1-J4}. 

\subsection{Scenario 2}
In the second scenario, we evaluate the performance when the underlying functional data $\{f_m(t),m=1,2\}$ differs only in their internal series correlation while maintaining the same variance amplitude. The hyper-parameters $\vesub{\theta}{f}$ are set as $\{ \upsilon_{10}=0.5,\upsilon_{11}=0.5, \vesub{A}{10}=1, \vesub{A}{11}=0.1\}$ and $\{\upsilon_{20}=0.5,\upsilon_{21}=0.5, \vesub{A}{20}=1, \vesub{A}{21}=1\}$, respectively. The other settings are identical to those in Scenario 1. One realization of this process is depicted in Figure \ref{fig:scenario2-obs2}. Comparing to Scenario 1, this scenario appears more challenging due to the lack of obvious difference in fluctuation amplitude, making the hidden state sequence more difficult to determine. Despite these complexities, as shown below, our FRVS model and CBSS method successfully discerned the hidden state sequences, achieving admirable classification metrics. Comparatively, the NGS method struggled to achieve the same level of accuracy.

Table \ref{table:Scenario12-sim} presents the classification performance metrics for estimating the state sequence and goodness-of-fit measures based on fifty replications. In this data generation case, CBSS method successfully identifies the hidden states. With a sample size of $n=60$, our method achieves an accuracy of $0.9743$ ($\mathbb{S}=20$), which is much better than the NGS approach having an accuracy of $0.7070$ only. Other metrics such as kappa, precision, specificity, and f1 score consistently show superior performance of the CBSS method for both sample sizes. The fitting performance of the CBSS method consistently surpass those of the NGS approach and demonstrate improvement as the sample size increases. With $n=60$, our method ($\mathbb{S}=20$) achieves much smaller rmse and mae of compared to the NGS approach's results. When the sample size is increased to $150$ ({as shown in Table \ref{table:Scenario2-Appen-sim} of the Appendix \ref{SM-sec:Scenario2-n150}}), the hidden state classification accuracy for the CBSS method reaches approximately $0.99$, with fitting measures rmse and mae decreased to $0.0805$ and $0.0564$, respectively. 

\subsection{Scenario 3}
In the third scenario, we examine the case where the sample size $n$ is extremely large. As previously discussed in Section \ref{sect:bayesian_inf}, such large sample sizes significantly increase the computational demands when estimating $\gall$ in each iteration. To address this challenge, we use the method discussed of the end of Section \ref{sec:update thetag and g} and select a smaller subset of $\{t_i\}_{i=1}^n$ with $\underline{n}$ time points (set to $100$ in this study) and estimate the functions $g_m(t)$, $m=1,\cdots, M-1$ for $t\in \bve{t}$. By leveraging the power of Gaussian process modeling, we can easily predict the functions for the remaining time points.

Moreover, large sample sizes pose significant challenges on sampling for the underlying functions $\fall$ as well. To tackle this problem, we employ the NNGP method with a neighbor size $\mathbb{N}=10$ and compare it with our standard procedure. The simulation setup is the same as in Scenario 1, featuring two states ($M=2$). We vary the sample size $n$ from $500$ to $10,000$ and randomly partition the samples into a training set ($75\%$) and a testing set ($25\%$). The training set is used for model fitting, while the testing set evaluates the accuracy of predictions for the responses.

Figure \ref{fig:scenario3} illustrates a comparison of the proposed FRVS model and CBSS method with and without using the NNGP approach and varying proposal set sizes $(\mathbb{S}=10,20)$. All results are based on five replications. The model training results demonstrate that while using the NNGP method yields slightly lower estimation and prediction accuracy compared to the approach without using NNGP, it significantly reduces computation time. Notably, the performance gap between two approaches narrows as the sample size increases. We therefore emphasize the advantages of the NNGP approach in large sample scenarios. As $n$ increases, the accuracy of hidden state classification approaches $1$, accompanied by a substantial decrease in both fitting and prediction errors, as evidenced by the smaller rmse and mae values.
\begin{figure}
	\centerline{
	\includegraphics[width=1\linewidth]{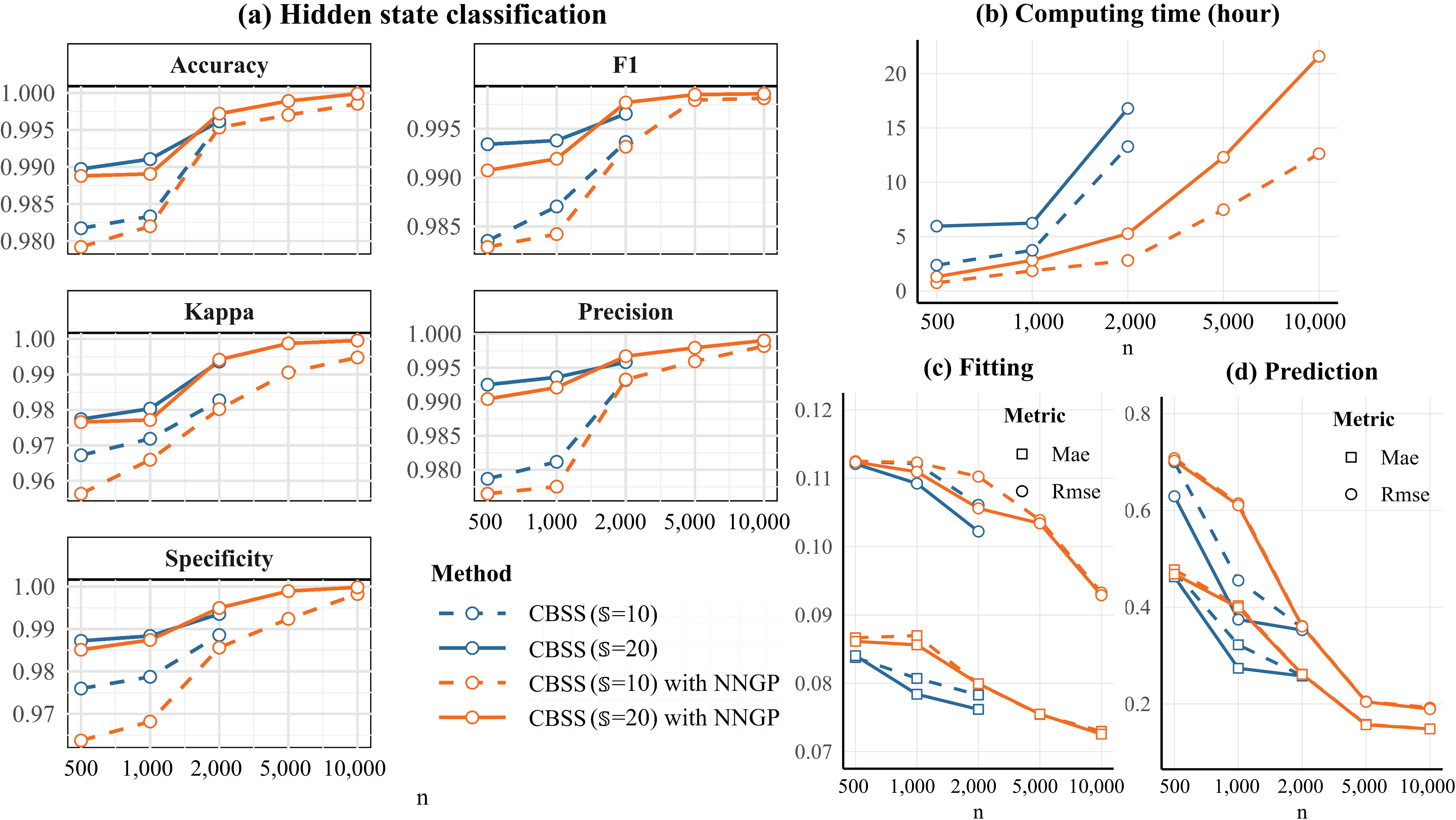}}
	\caption{Comparative analysis of fitting and prediction performance for the unknown variables across varying sample size $n$. Results are averaged over five replications. (a): Average classification performance metrics for the hidden state sequence estimation using different methods. (b): Average computing time using different methods. (c): Average rmse and mae between the estimated $\hve{y}$ and the true values for the training set. (d): Average rmse and mae between the estimated $\hve{y}$ and the true values for the testing set. When $n$ is larger than $2,000$, it exceeds our server's capacity when NNGP is not used, resulting in the omission of the corresponding results.}
	\label{fig:scenario3}
\end{figure}

\section{Applications} \label{sec:App}
In this section, we analyze two distinct signal datasets, with a particular emphasis on the precise segmentation of anomalous periods. These high-frequency real-world datasets encompass thousands of data points, prompting us to employ the NNGP approach for sampling the underlying functions, thereby effectively addressing computational demands. Throughout the data collection process, timestamps were recorded. To align with our FRVS model structure, we standardized the timestamps by setting the initial time point to zero and normalizing the entire range to span from $0$ to $1$. To evaluate the model's performance, we randomly selected $75\%$ of the data points for training and fitting, while reserving the remaining $25\%$ for testing and assessing prediction accuracy.

In this section, we focused on the one-dimensional response variable. However, our model can easily be extended to a multi-output scenario. The detailed multi-output model structure and the covariance and cross-covariance of these functions can be defined using convolution GP, which is outlined in Appendix \ref{SM-sec:CP-MO}. All inference and implementation procedures remain similar to those in the single-output scenario. The applications to complete tri-axial acceleration data of these two datasets are presented and compared against their single-output modeling in Appendix \ref{SM-sec:CP-MO}.

\subsection{Jumping Exercise Segmentation}
The first dataset comprises vertical acceleration signals collected using an Axivity 6 device, with programmed to $\pm8$ g, $250$ dps, and $100$ Hz \citep{powell2021investigating}. The device was worn by the participant on the fifth lumbar vertebra, approximating the center of mass \citep{van2015assessing,del2019gait}, in a free-living environment. This $110$-second recording corresponds to $11,000$ data points, during which the participant engaged in a continuous jumping exercise for an unspecified duration, preceded and followed by phases of normal walking. The data is illustrated in the top panel of Figure \ref{fig:MO-real}, with our objective being to accurately identify the unknown timestamps of transition into and out of the jumping exercise and capture the curve pattern effectively. This application is particularly relevant to gait analysis, where the segmentation of walking intervals poses a challenge \citep{hickey2016detecting}. Existing methodologies primarily utilize data-driven techniques, such as wavelet transformations, to detect frequency domain alterations \citep{LYONS2005497,maclean2023walking}. Those methods are quite efficient on identifying the phase of normal walking, but not so good on identifying other movements.

We fit the training data with the proposed model with two hidden states.
After model training, the estimated parameters are: $\sigma^2=0.0161$, $\vesub{\theta}{z}=\{q_{12}=4.4495, q_{21}=1.4112\}$. The hyper-parameters $\vesub{\theta}{f}$ are estimated as $\vesub{\theta}{f}=\big\{ \upsilon_{10}=0.0488,\upsilon_{11}=0.0289, \vesub{A}{10}=8.0777, \vesub{A}{11}=0.1044, \upsilon_{20}=0.0410,\upsilon_{21}=0.1945, \vesub{A}{20}=5.7948, \vesub{A}{21}=79.8046 \big\}$. The jumping exercise condition exhibits larger $\upsilon_{21}$ and $\vesub{A}{21}$ indicating higher fluctuation and lower inner correlation compared to those of normal walking, aligning well with the data pattern. The fitting results yield an rmse of $0.1198$ and an mae of $0.0503$. The top panel in Figure \ref{fig:Jump-est} presents an estimation of the underlying function for the training set, effectively capturing the trends and patterns within the data. The jumping period is segmented from $t=0.1504$ to $0.9011$. The bottom panel in Figure \ref{fig:Jump-est} illustrates the predictions for the testing set, yielding an rmse of $0.1386$ and an mae of $0.0714$, which, while slightly higher than the fitting results, remains satisfactory.
\begin{figure}
	\centering
	\subfloat[Jumping exercise segmentation.]{	
		\includegraphics[width=1\linewidth]{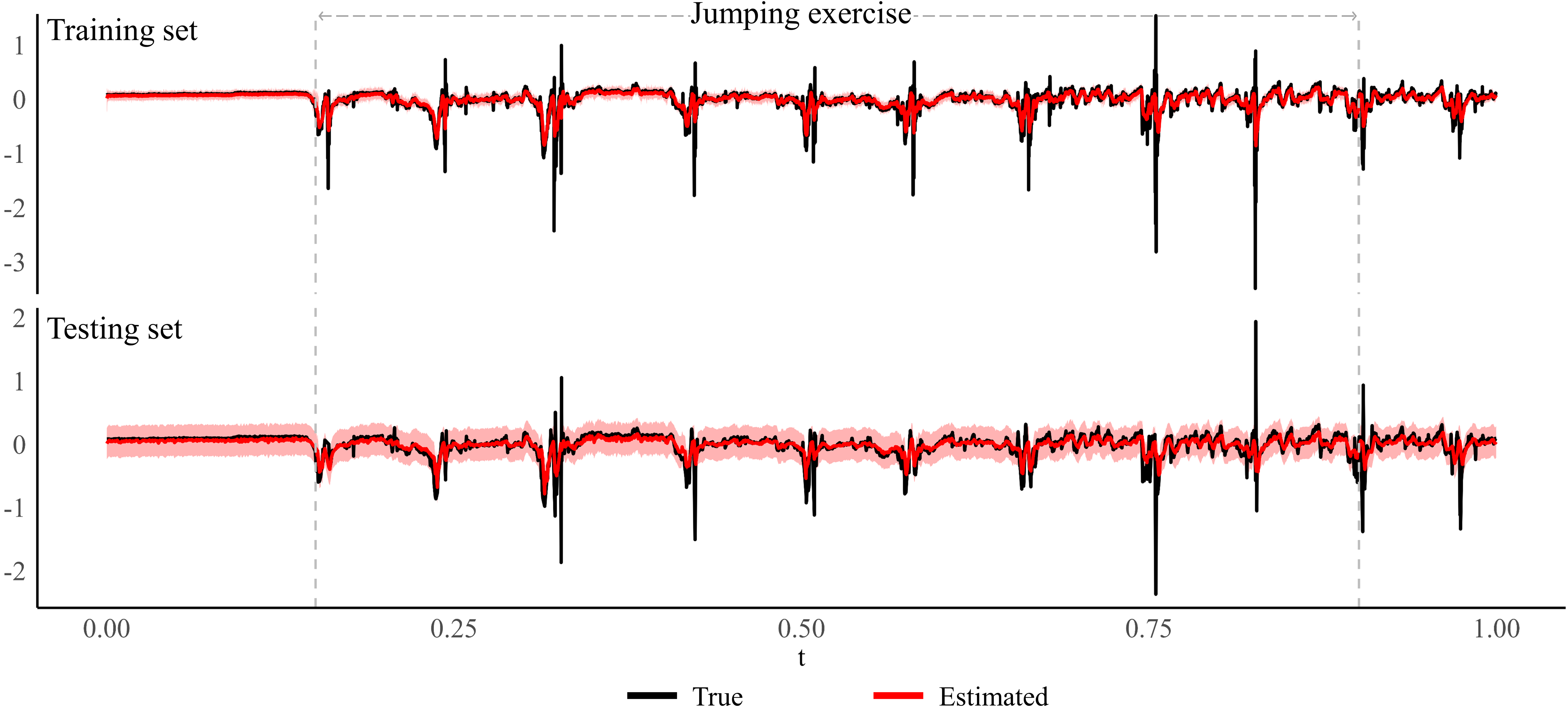}
		\label{fig:Jump-est}
	}\\
	\subfloat[FoG detection.]{
		\includegraphics[width=1\linewidth]{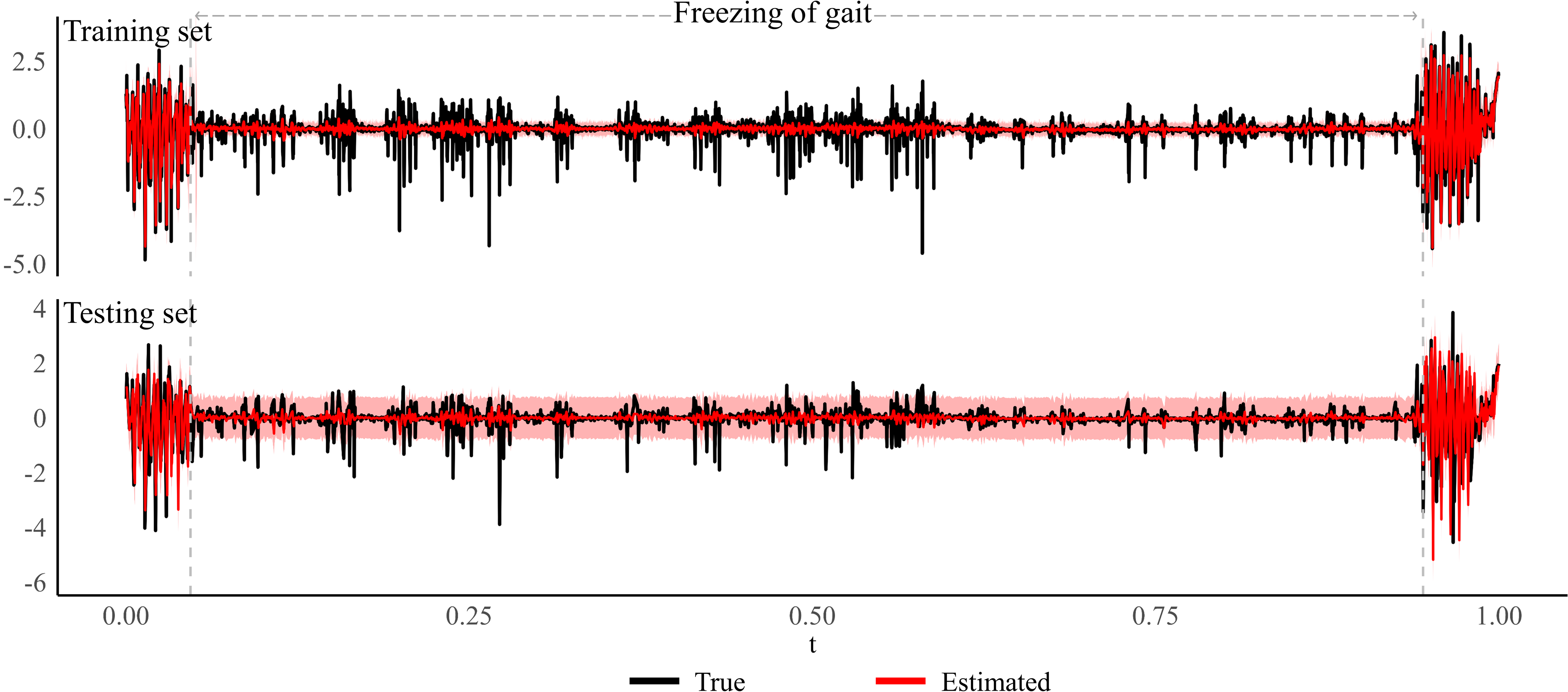}
			\label{fig:FoG-est}}
	\caption{Estimation results for the training and testing sets, with indications of the state transition period. The shaded region indicating the sampling range and confidence interval for the training and testing sets, respectively. }
\end{figure}


\subsection{Freezing of Gait Detection}
The second dataset consists of vertical acceleration data collected in a controlled laboratory environment from a patient diagnosed with Parkinson's disease during a timed-up-and-go test \citep{nonnekes2020freezing}. This test requires the patient to walk in a straight line from a starting position, perform a turn, and return to the starting point. During this $105$-second recording, the patient experienced an episode of freezing of gait (FoG). This symptom is characterized by a sudden and temporary inability to initiate or maintain walking. As depicted in the bottom panel of Figure \ref{fig:MO-real}, the fluctuation during the FoG episode is significantly reduced compared to those at the beginning and end of the recording. Detecting FoG poses challenges for gait analysis. Recent advancements in technology, particularly wearable sensors and deep brain stimulation, offer avenues for managing and mitigating this complex symptom \citep{gilat2018freezing,mancini2019clinical}. In a controlled laboratory setting, where video recording was utilized, this dataset was labeled by training professionals with onset and termination times for the FoG episode. However, accurately identifying these occurrences in free-living settings remain difficult.

With our model training, the estimated parameters are: $\sigma^2=0.1506$, $\vesub{\theta}{z}=\{q_{12}=8.3697, q_{21}=1.7491\}$. The hyper-parameters $\vesub{\theta}{f}$ are estimated as $\big\{ \upsilon_{10}=0.0787,\upsilon_{11}=4.7367, \vesub{A}{10}=83.7847, \vesub{A}{11}=112.1777, \upsilon_{20}=0.3017,\upsilon_{21}=0.0974, \vesub{A}{20}=63.1094, \vesub{A}{21}=3.2331 \big\}$. The FoG episode exhibits extremely smaller $\upsilon_{21}$ and $\vesub{A}{21}$ for lower fluctuation and larger inner correlation compared to those for normal walking, which still aligns well with the data pattern. The fitting results yield an rmse of $0.3637$ and an mae of $0.2005$. The top panel in Figure \ref{fig:FoG-est} presents the estimation of the underlying function for the training set, capturing the trends and shapes of the data. Since this dataset is labeled, we compare the estimated hidden state sequence with the actual one, resulting in classification metrics of accuracy = $0.9985$, kappa = $0.9923$, precision = $1$, specificity = $1$, and f1 = $0.9931$. The bottom panel in Figure \ref{fig:FoG-est} illustrates the prediction results for the testing set, yielding an rmse of $0.4980$ and an mae of $0.2556$. 

\section{Discussion} \label{sec:Discussion}

This study proposes a novel methodological approach for analyzing functional data characterized by varying covariance structures, with a particular emphasis on the curve-based sampling scheme implemented within a continuous-time HMM framework. In contract to most existing works in the HMM literature, which assume independence of response variables given the hidden states, our FRVS model recognizes and incorporates the dependencies among response variables conditioned on the hidden states. 

The primary contributions of this research are multifaceted, with the curve-based sampling scheme being central to addressing the sampling complexities associated with conditionally dependent functional data. By modeling the underlying probability curves of the hidden state process, our approach circumvents the computational challenges related to enumerating various state sequence combinations. This issue is particularly critical in functional data analysis, such as the multiple testing. The work by \citet{xu2019automatic} requires the choice of representative points; in contrast, our curve-based sampling scheme may enhance efficiency by preserving the integrity of the functional data. Using a curve-based sampling scheme may lead to a breakthrough in addressing the intractable two-sided test problems in multiple testing for functional data.

Furthermore, our emphasis on identifying variations in covariance structure is crucial for many problems in functional data analysis. In the context of our accelerometer signal data, our method excels in detecting changes in covariance structure that distinguish different walking patterns. The utilization of the NNGP approach facilitates the efficient handling of large-scale functional data while maintaining model accuracy. The detailed Bayesian inference procedure presented herein provides a robust framework for estimating parameters and hidden states, enhancing the model's applicability and precision across diverse fields. The examination of asymptotic properties ensures both identifiability and consistency in model structure, providing theoretical guarantees that support the reliability of our approach.

 {In this work, the observed data form a single densely sampled trajectory of a continuous-time process. While conventional functional data analysis typically focuses on the analysis of multiple functional observations (i.e., a sample of curves), similar modeling tools are frequently used when analyzing a single functional trajectory observed over a continuous domain. Our analysis is based on a single functional trajectory as the aim is to characterize the structural variation of the observed response curve, rather than to make population-level inferences (such as to estimate the mean function).} In this context, the substantial within-curve variation across the domain provides adequate information for identifying the functional relationship, allowing consistent estimation without the need for repeated response curves; see \citet{wang2014generalized} for discussion on the consistency of a single functional trajectory. However, the Bayesian inference procedure proposed in this paper can be applied to repeated curves without technical difficulties but notational sophistication.

Despite the promising results of our method, several areas worth a further investigation. One potential direction is to address the complexities of selecting the number of hidden states.  {In this study, the number of hidden states $M$ is assumed to be known and such number is selected based on domain knowledge in the applications. In practical applications where $M$ is unknown, several strategies can be considered. A straightforward approach is to fit the proposed FRVS model with different candidate values of $M$ and select the optimal model using model selection criteria such as BIC, DIC, or marginal likelihood. Alternatively, Bayesian methods that allow the number of states to vary, such as reversible jump MCMC or Bayesian nonparametric approaches including hierarchical Dirichlet process hidden Markov models, could be adopted to infer the number of states directly from the data. Once $M$ is specified at a given iteration, the proposed curve-based sampling scheme for updating the hidden state sequence can be applied without modification. Investigating these extensions would be an interesting direction for future research.} Moreover, further exploration is needed to evaluate the model's effectiveness in handling complex covariance structures and multi-output scenarios. While the analysis of data heterogeneity is generally applicable, careful consideration is necessary when determining whether to incorporate distinct features into the covariate vector $\ve{x}$ or within the CTMM model itself. For multi-output modeling, the convolution process we utilized represents a viable methodology; yet further research is required to optimize its performance across various contexts.

\backmatter
 {The code for simulation studies will be made available at \url{https://github.com/yuyian/FRVS} (R Statistical Software, v4.5.1). The data used in the real-data applications are available from the authors upon reasonable request.}

\bmhead{Acknowledgements}

Shi's work is supported by funds from the National Key R\&D Program of China (2023YFA1011400), the National Natural Science Foundation of China (No. 12271239), and Shenzhen Fundamental Research Program JCYJ20220818100602005 (No. 20220111). Wang thanks the University of Leicester for the support of academic study leave during which part of this work was carried out.

\begin{appendices}

\section{Convolution Process for GP}\label{SM-sec:CPGP}
	For illustration purposes, we take $M=2$ as an example to describe the convolution process. We begin by defining three independent Gaussian white noise processes: $\varepsilon_0(\cdot)$, $\varepsilon_1(\cdot)$, and $\varepsilon_2(\cdot)$. We denote the convolution product by the symbol $``\star"$. Let ${\kappa}_{m 0}(\cdot)$ and ${\kappa}_{m 1}(\cdot)$ be any smoothing kernel functions for $m=1,2$. By convolving these white noise processes with the smoothing kernel functions, we construct the following GPs: \begin{equation*}
		\begin{aligned}
			\xi_m(\cdot)={\kappa}_{m0}(\cdot)\star \varepsilon_0(\cdot),\ 
			\eta_m(\cdot)={\kappa}_{m1}(\cdot) \star \varepsilon_m(\cdot),\ \text{for } m=1, 2. 
		\end{aligned} \label{SP-eq:conv-2parts}
	\end{equation*}Here, $\xi_1(\cdot)$ and $\xi_2(\cdot)$ are dependent since both processes are constructed from the same Gaussian white noise $\varepsilon_0(\cdot)$, while they are independent of $\eta_1(\cdot)$ and $\eta_2(\cdot)$. So are the cases for $\eta_1(\cdot)$ and $\eta_2(\cdot)$.
	
	Next, we define $f_m(\ve{x}(t))=\xi_m(\ve{x}(t))+\eta_m(\ve{x}(t))$ and specify the smoothing kernel functions with positive parameters $\vesub{\theta}{f}=\{ \upsilon_{m0},\upsilon_{m1}, \vesub{A}{m0}, \vesub{A}{m1}\}_{m=1}^2$ as follows:
	\begin{equation*}
		\begin{aligned}
			{\kappa}_{m 0}(\ve{x}(t))&=\upsilon_{m0}\exp\left\{ -\frac{1}{2}\ve{x}(t)^{T} \vesub{A}{m0}\, \ve{x}(t)\right\},\\
			{\kappa}_{m 1}(\ve{x}(t))&=\upsilon_{m1}\exp\left\{ -\frac{1}{2} \ve{x}(t)^{T} \vesub{A}{m1}\, \ve{x}(t) \right\},\ m=1,2.
		\end{aligned}
	\end{equation*}
	Then $\left(f_1(\ve{x}(t)), f_2(\ve{x}(t)) \right)^T$ defines a dependent bivariate GP with zero means and the following kernel covariance functions:
	\begin{equation}
		\begin{aligned}
			&k_{mm}(\vesub{x}{i},\vesub{x}{j})\\
			=&\pi^{d/2}\upsilon_{m0}^2\left|\vesub{A}{m0} \right|^{-1/2} \exp\left\{ -\frac{1}{4}(\vesub{x}{i}-\vesub{x}{j})^T\vesub{A}{m0}(\vesub{x}{i}-\vesub{x}{j}) \right\} \\
			& \qquad + 
			\pi^{d/2}\upsilon_{m1}^2\left|\vesub{A}{m1} \right|^{-1/2} \exp\left\{ -\frac{1}{4}(\vesub{x}{i}-\vesub{x}{j})^T\vesub{A}{m1}(\vesub{x}{i}-\vesub{x}{j}) \right\}, \ m=1,2, \\
			&k_{12}(\vesub{x}{i},\vesub{x}{j})=k_{21}(\vesub{x}{j},\vesub{x}{i})\\
			=&(2\pi)^{d/2}
			\upsilon_{10}\upsilon_{20}\left|\vesub{A}{10}+\vesub{A}{20} \right|^{-1/2} \exp\left\{ -\frac{1}{2}(\vesub{x}{i}-\vesub{x}{j})^T\vesub{A}{10}\left(\vesub{A}{10}+\vesub{A}{20}\right)^{-1} \vesub{A}{20}(\vesub{x}{i}-\vesub{x}{j}) \right\}.
		\end{aligned} \label{SM-eq:cov-cross}
	\end{equation}

\section{NNGP Method for Sampling $\fall$ }\label{SM-sec:NNGP}
	As discussed in Remark 1, the standard sampling procedure for $\fall$ is infeasible when the sample size $n$ is large. To avoid the computational costs for solving the inverse of high-dimensional matrices, we employ the NNGP method. This approach utilizes sparse Gaussian distributions on directed acyclic graphs, providing a well-defined, sparsity-inducing prior for modeling large spatial and spatio-temporal datasets, yielding legitimate finite-dimensional Gaussian densities with sparse precision matrices, thereby efficiently facilitating computational algorithms without requiring the storage or decomposition of large matrices \citep{datta2016hierarchical,datta2016nearest}. 
	
	In practice, the convolution GP prior assumption for the unknown functions $\{f_m(\cdot)\}_{m=1}^M$ implies that $\fall\sim N\left(\vesub{0}{nM} \, ,\, \ve{K}\right)$. Expressing the joint density of $\fall$ as the product of conditional densities, we can write \begin{align*}
		p\left(\fall\right)=p\left(\ve{f}(t_1)\right)\prod_{i=2}^{n}p\left(\ve{f}(t_{i}) \mid \ve{f}(t_{i-1}),\cdots, \ve{f}(t_1)\right).
	\end{align*}
	The NNGP method builds upon the idea of replacing the larger conditioning sets on the right-hand side with smaller, carefully selected sets of size at most $\mathbb{N}$, where $\mathbb{N}\ll n$. Specifically, by selecting the smaller conditioning sets $\mathcal{N}_1=\emptyset$ and $\{\mathcal{N}_i \subset \{t_1,\cdots, t_{i-1}\}\}_{i=2}^n$ with $\max_{i}|\mathcal{N}_i|=\mathbb{N}$, it ensures a direct acyclic graph. A proper approximate density is derived from the density ${p}\left(\fall\right)$ and constructed as follows: \begin{align*}
		&\tilde{p}\left(\fall\right)=\prod_{i=1}^{n} p\left(\ve{f}(t_{i}) \mid \fall_{\mathcal{N}_i}\right)\\
		=&\prod_{i=1}^{n}N\left(\ve{f}(t_i) \mid  \ve{\mathrm{K}}(t_i,\mathcal{N}_i)\ve{\mathrm{K}}(\mathcal{N}_i,\mathcal{N}_i)^{-1}\fall_{\mathcal{N}_i}\ ,\ \ve{\mathrm{K}}(t_i,t_i)-\ve{\mathrm{K}}(t_i,\mathcal{N}_i)\ve{\mathrm{K}}(\mathcal{N}_i,\mathcal{N}_i)^{-1} \ve{\mathrm{K}}(\mathcal{N}_i,t_i) \right),
	\end{align*} where $\fall_{\mathcal{N}_i} =\left\{\ve{f}{(t_j)},t_j \in\mathcal{N}_i \right\}$ is the vector formed by stacking the realizations of $\ve{f}(t)$ at $\mathcal{N}_i$. Such approximation $\tilde{p}\left(\fall\right)$ is a multivariate Gaussian density with its covariance matrix $\tve{K}$, differing from $\ve{K}$. Since we assume $\mathcal{N}_i$ has at most $\mathbb{N}$ members for each $t_i$, the inverse $\tvesup{K}{-1}$ is sparse with at most $n\mathbb{N}(\mathbb{N}+1)M^2/2$ nonzero entries, as discussed and proved in \citet{datta2016hierarchical}. 
	
	Turning to the conditioning sets, known as neighbor sets, several methods are considered for choosing $\mathcal{N}_i$ \citep{vecchia1988estimation,datta2016nonseparable}. All the choices depend on the ordering of the points. In our case, since the time points are naturally ordered, we designate $\mathcal{N}_i$ as the $\mathbb{N}$ nearest neighbors of $t_i$ among $t_1,\cdots, t_{i-1}$. Since $\tilde{p}\left(\fall\right)$ ultimately relies on the information borrowed from the neighbors, its effectiveness is often determined by the number of neighbors specified.

	Given a hidden state sequence $\ve{z}$, the approximate distribution for $\vesub{f}{\bs{z}}$ is multivariate Gaussian with covariance matrix $\tvesub{K}{\bs z }$. The matrix $\tvesub{K}{\bs z }$ is a submatrix of $\tve{K}$, and the $(i,j)$-th entry in $\tvesub{K}{\bs z }$ corresponds to the $\left((i-1)M +z_i\, ,\, (j-1)M +z_j\right)$-th entry in $\tve{K}$. The inverse matrix $( \tvesub{K}{\bs z } + \sigma^2 \vesub{I}{n})^{-1}$ retains the same sparsity as $\tvess{K}{\bs z}{-1}$, for $( \tvesub{K}{\bs z } + \sigma^2 \vesub{I}{n})^{-1}=\sigma^{-2}\tvess{K}{\bs z}{-1}(\tvess{K}{\bs z}{-1} + \sigma^{-2}\vesub{I}{n})^{-1}$. A sparse Cholesky factorization of $( \tvess{K}{\bs z }{-1} + \sigma^{-2} \vesub{I}{n})$ will efficiently yield the Cholesky factors of $( \tvess{K}{\bs z }{-1} + \sigma^{-2} \vesub{I}{n})^{-1}$ and facilitate the updates of $\vesub{f}{z}$. However, the eigenvalues of $(\tvesub{K}{\bs{z}}+\sigma^2\vesub{I}{n})$ may be less than $\sigma^2$. Since $$\left(\vesub{f}{\bs{z}}\mid \ve{z},\Theta,\ve{y}\right) \sim N\left( \vesub{K}{\bs{z}}\left(\vesub{K}{\bs{z}}+\sigma^2\vesub{I}{n}\right)^{-1}\ve{y} \, ,\, \sigma^2 \vesub{K}{\bs{z}} \left(\vesub{K}{\bs{z}}+\sigma^2\vesub{I}{n}\right)^{-1} \right),$$ the covariance matrix $\Cov\left(\vesub{f}{\bs{z}}\mid \ve{z},\Theta,\ve{y}\right)=\sigma^2 \tvesub{K}{\bs{z}} \left(\tvesub{K}{\bs{z}}+\sigma^2\vesub{I}{n}\right)^{-1}=\sigma^2\vesub{I}{n}-\sigma^4(\tvesub{K}{\bs{z}}+\sigma^2\vesub{I}{n})^{-1} $ may not be positive definite, resulting in the probability $p\left(\vesub{f}{\bs{z}}  \mid  \ve{z},\Theta,\ve{y} \right)$ no longer being well-defined with the NNGP assumption. Using such a sparse matrix to update $\fall$ requires an efficient sparse Cholesky solver for $( \tvess{K}{\bs z }{-1} + \sigma^{-2} \vesub{I}{n})$. However, the computational time for Gibbs sampling in this marginalized model depends on the sparse structure of $\tvesup{K}{-1}$ and may sometimes heavily exceed the linear usage achieved by the unmarginalized model, which updates $\ve{f}(t_i)$ individually, conditional on $y_i$ and the neighbors $\vesub{y}{\mathcal{N}_i}=\left\{y_j,t_j \in\mathcal{N}_i \right\}$. 
	
	In our sampling procedure, since $\ve{f}(t_i)=\left(f_{z_i}(t_i),\vesub{f}{\text{-}z_i}(t_i)^T\right)^T $, we first update $f_{z_i}(t_i)$ from $p\big( f_{z_i}(t_i) \mid \ve{z}, \sigma^2, \vesub{\theta}{f},y_i, \vesub{y}{\mathcal{N}_i}\big)$ and subsequently update $\vesub{f}{\text{-}z_i}(t_i)$ from $p\big(\vesub{f}{\text{-}z_i}(t_i)  \mid f_{z_i}(t_i), \ve{z}, \sigma^2, \vesub{\theta}{f}, \vesub{y}{\mathcal{N}_i}\big)$. Both updates stem from Gaussian distributions, represented as follows:
	\begin{align*}
		\left( {f}_{z_i}(t_i) \mid \ve{z}, \sigma^2, \vesub{\theta}{f},{y}_i, \vesub{y}{\mathcal{N}_i}\right) & \sim N\left(\vesub{\mu}{z_i}\, , \, \vesub{\Sigma}{z_i} \right), \\
		\left(\vesub{f}{\text{-}z_i}(t_i)  \mid {f}_{z_i}(t_i), \ve{z}, \sigma^2, \vesub{\theta}{f}, \vesub{y}{\mathcal{N}_i}\right) & \sim N\left(\vesub{\mu}{\text{-}z_i} \, , \, \vesub{\Sigma}{\text{-}z_i} \right), 
	\end{align*}
	where \begin{align*}
		\vesub{\mu}{z_i}&=\ve{A}\begin{pmatrix}{y}_i \\ \vesub{y}{\mathcal{N}_i} \end{pmatrix},\ \vesub{\Sigma}{z_i}= {k}_{z_i z_i}(t_i,t_i)-\ve{A}\begin{pmatrix}{k}_{z_i z_i}(t_i,t_i) \\ {k}_{ \bs{z}_{\mathcal{N}_i}z_i}(\mathcal{N}_i,t_i) \end{pmatrix}, \\
		\ve{A}&=\begin{pmatrix}{k}_{z_i z_i}(t_i,t_i) & {k}_{z_i \bs{z}_{\mathcal{N}_i}}(\mathcal{N}_i,t_i) \end{pmatrix} \begin{pmatrix}{k}_{z_i z_i}(t_i,t_i)+\sigma^2 & {k}_{z_i \bs{z}_{\mathcal{N}_i}}(t_i,\mathcal{N}_i) \\
			{k}_{ \bs{z}_{\mathcal{N}_i}z_i}(\mathcal{N}_i,t_i)& {k}_{\bs{z}_{\mathcal{N}_i} \bs{z}_{\mathcal{N}_i}}(\mathcal{N}_i,\mathcal{N}_i)+\sigma^2 \vesub{I}{|\mathcal{N}_i|} \end{pmatrix}^{-1},
	\end{align*} and \begin{align*}
		\vesub{\mu}{\text{-}z_i}&=\ve{B}\begin{pmatrix}{f}_{z_i}(t_i) \\ \vesub{y}{\mathcal{N}_i} \end{pmatrix},\ \vesub{\Sigma}{\text{-}z_i}= {k}_{\text{-}z_i \text{-}z_i}(t_i,t_i)-\ve{B}\begin{pmatrix}{k}_{z_i \text{-}z_i}(t_i,t_i) \\ {k}_{ \bs{z}_{\mathcal{N}_i}\text{-}z_i}(\mathcal{N}_i,t_i) \end{pmatrix}, \\
		\ve{B}&=\begin{pmatrix}{k}_{\text{-}z_i z_i}(t_i,t_i) & {k}_{\text{-}z_i \bs{z}_{\mathcal{N}_i}}(\mathcal{N}_i,t_i) \end{pmatrix} \begin{pmatrix}{k}_{z_i z_i}(t_i,t_i) & {k}_{z_i \bs{z}_{\mathcal{N}_i}}(t_i,\mathcal{N}_i) \\
			{k}_{ \bs{z}_{\mathcal{N}_i}z_i}(\mathcal{N}_i,t_i)& {k}_{\bs{z}_{\mathcal{N}_i} \bs{z}_{\mathcal{N}_i}}(\mathcal{N}_i,\mathcal{N}_i)+\sigma^2 \vesub{I}{|\mathcal{N}_i|} \end{pmatrix}^{-1}. 
	\end{align*}

\section{Prediction}
	After the model has been trained and all unknown variables ($\fall$ and $\ve{z}$) and parameters ($\Theta$) have been estimated, we can calculate the prediction of $y^*$ at a new time point $t^*$. For the sake of notations, we still use $\fall$, $\ve{z}$, and $\Theta$ to denote their estimates. We first estimate the hidden state by:
	\begin{align*}
		z^\ast=&z(t^\ast)=\argmax_{m}p\left(z(t^\ast)=m \mid  \ve{z},\Theta\right)\\
		=&\begin{cases}
			\argmax_{z \in \{1,\cdots,M\}}\vesub{P}{z_{i}\, z}\left(t^\ast-t_{i}\right)\vesub{P}{ z\, z_{\text{i+1}} }\left(t_{i+1}-t^\ast\right), \, &\text{if }\ t_{i}< t^\ast < t_{i+1}, \\
			\argmax_{z\in\{1,\cdots, M\}}\vesub{P}{z_{n}\, z}\left(t^\ast-t_{n}\right),\, &\text{if }\ t^\ast > t_n.
		\end{cases}
	\end{align*} 
	The posterior distribution $p\left(y^\ast \mid  \fall,\ve{z},z^\ast,\Theta,\ve{y}\right)=p\left(y^\ast \mid  \fall,z^\ast,\Theta\right)$ is used to predict $y^\ast$. Denoting the covariance between $f_{z^\ast}(t^\ast)$ and $\fall$ as $\vess{k}{z^\ast}{\ast}$, then the posterior mean and variance are derived as follows:\begin{align*}
		\E\left( y^\ast \mid \fall,z^\ast,\Theta \right)=\vess{k}{z^\ast}{\ast }\vesup{K}{-1}\fall,\
		\Var\left( y^\ast \mid \fall,z^\ast,\Theta \right)=k_{z^\ast z^\ast}(t^\ast,t^\ast)+\sigma^2-\vess{k}{z^\ast}{\ast }\vesup{K}{-1}\vess{k}{z^\ast}{\ast T}.
	\end{align*} With the NNGP assumption for $\fall$, we can calculate the mean and variance from the posterior distribution $p\left(y^\ast \mid  \fall_{\mathcal{N}^\ast},z^\ast,\Theta\right)$ with $\mathcal{N}^\ast$ being the conditioning set for the time point $t^\ast$. Here $\mathcal{N}^\ast$ is much smaller than the sample size $n$, resulting in an efficient algorithm.

\section{Update $\gall$ and $\vesub{\theta}{g}$}\label{SM-sec:predict_new_z_g}
In our curve-based sampling scheme, both $\gall$ and $\ve{z}$ are updated in Step 3. After the burn-in period, the sample of $\gall$ produced in the previous iteration closely resembles its posterior distribution (in the left panel of Figure \ref{fig:curve-based}). Similarly, the state sequence $\ve{z}$ generated from this sample of $\ve{g}(t)$ is also very close to its posterior distribution (in the middle panel of Figure \ref{fig:curve-based}), making it easier to accept in the MH algorithm. Here, we will give a detail discuss about how to update $\gall$ and the related hyper-parameters.
\begin{figure}
	\centering
	\includegraphics[width=1\linewidth]{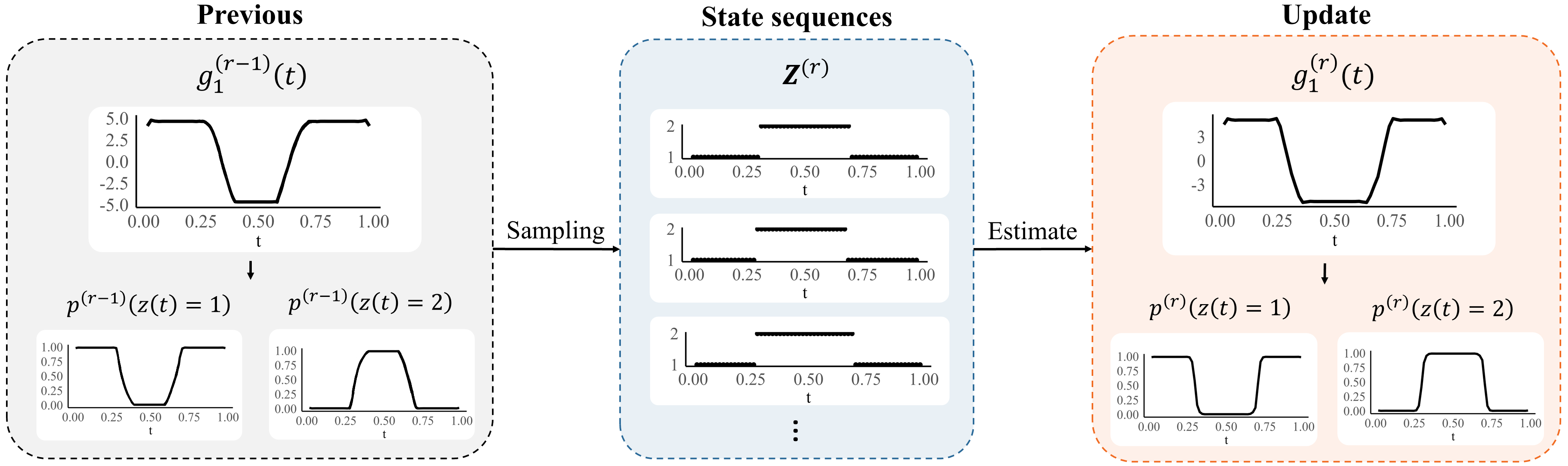}
	\caption{ Curve-based sampling scheme for updating the hidden state sequence for the case of $M=2$. (Left panel): A set of possible state sequences is sampled based on the estimated $\ve{g}(t)$ from the previous iteration; (Middle panel): Samples of the state sequences are generated by MH algorithm; (Right panel): $\ve{g}(t)$ is updated.}
	\label{fig:curve-based}
\end{figure}

After the $S$ samples of state sequences $\ve{Z}=\{\vesub{z}{(s)}\}_{s=1}^S$ are obtained. The hyper-parameters $\vesub{\theta}{g}$ are estimated by maximizing the marginal density function $	p\left(\ve{Z}  \mid \vesub{\theta}{g} \right)$. However, as introduced in Section 3.2, the integral involved in the marginal density is analytically intractable and we approximate $p\left(\ve{Z}  \mid \vesub{\theta}{g} \right)$ by $
\tilde{p}\left(\ve{Z}  \mid \vesub{\theta}{g} \right)\triangleq \left.\frac{p\left(\ve{Z},\gall  \mid \vesub{\theta}{g} \right)}{	\tilde{p}_G\left(\gall  \mid \ve{Z},\vesub{\theta}{g} \right)}\right|_{\gall=\tilde{\gall}},
$ where $\tilde{p}_G\left(\gall  \mid \ve{Z},\vesub{\theta}{g} \right)$ is the Gaussian approximation to the conditional density $p\left(\gall  \mid \ve{Z},\vesub{\theta}{g}\right)$, and $\tilde{\gall}$ is the mode of the conditional density of $\gall$, also being the estimation for $\gall$. Here, $$ p\left(\ve{Z},\gall \mid  \vesub{\theta}{g}\right)=p\left(\ve{Z}  \mid \gall \right)p\left(\gall |\vesub{\theta}{g} \right)
=\exp\left\{ \sum_{s\in\mathcal{S}} \sum_{i=1}^n \log p\left(z_{(v)_i}  \mid \ve{g}(t_i) \right) + \log p\left(\gall |\vesub{\theta}{g} \right) \right\}, $$ and $\ell_{z_{(s)_i}}\left(\ve{g}(t_i)\right)\triangleq \log p\left(z_{(s)_i}  \mid \ve{g}(t_i) \right)$ is approximated by a Taylor expansion to the second order around $\tve{g}(t_i)=\big(\tilde{g}_{1}(t_i),\cdots, \tilde{g}_{M-1}(t_i)
\big)^T$, such that \begin{align*}
	\ell_{z_{(s)_i}}\left(\ve{g}(t_i)\right)\simeq& \ell_{z_{(s)_i}}\left(\tve{g}(t_i)\right) +\left(\ve{g}(t_i)-\tve{g}(t_i) \right)^T \left. \frac{\partial \ell_{z_{(s)_i}}\left(\ve{g}(t_i)\right)}{\partial \ve{g}(t_i)}\right|_{\ve{g}(t_i)=\tve{g}(t_i) }\\
	&-\frac{1}{2}\left(\ve{g}(t_i)-\tve{g}(t_i) \right)^T\left(- \left. \frac{\partial^2 \ell_{z_{(s)_i}}\left(\ve{g}(t_i)\right)}{\partial \ve{g}(t_i)\partial \ve{g}(t_i)^T}\right|_{\ve{g}(t_i)=\tve{g}(t_i) }\right) \left(\ve{g}(t_i)-\tve{g}(t_i) \right).\end{align*}
Denote \begin{equation*}\small
	\begin{aligned}
		& \ve{D}(t_i,t_i) =\left\{\left. -\sum_{s\in\mathcal{S}}\frac{ \partial^2 \ell_{z_{(s)_i}} \left(\ve{g}(t_i)\right) }{\partial g_m(t_i)\partial g_h(t_i)}\right|_{\ve{g}(t_i)=\tve{g}(t_i)} \right\}_{m,h =1}^{M-1},\, \ve{D} =\diag\big( \ve{D}(t_1,t_1),\cdots,\ve{D}(t_n,t_n) \big), \\
		& \vesub{d}{i} = \left(\left. \sum_{s\in\mathcal{S}}\frac{ \partial \ell_{z_{(s)_i}} \left(\ve{g}(t_i)\right) }{\partial g_1(t_i)}\right|_{\ve{g}(t_i)=\tve{g}(t_i)},\cdots, \left. \sum_{s\in\mathcal{S}}\frac{\partial \ell_{z_{(s)_i}} \left(\ve{g}(t_i)\right) }{\partial g_{M-1}(t_i)}\right|_{\ve{g}(t_i)=\tve{g}(t_i)} \right),\,
		\ve{d} =\left( \vesub{d}{1} ,\cdots, \vesub{d}{n} \right),
	\end{aligned}
\end{equation*}
where
\begin{align*}
	\left. \frac{ \partial \ell_{z_{(s)_i}} \left(\ve{g}(t_i)\right) }{\partial g_m(t_i)}\right|_{\ve{g}(t_i)=\tve{g}(t_i)}&=\mathbb{I}\left(z_{(s)_i} =m\right) - \frac{\exp\left( \tilde{g}_{m}(t_i)\right)}{\sum_{{m}^{\ast} =1}^{M-1}\exp\left( \tilde{g}_{{m}^{\ast} }(t_i)\right)+1},\\
	\left. \frac{ \partial^2 \ell_{z_{(s)_i}} \left(\ve{g}(t_i)\right) }{\partial g_m(t_i)\partial g_h(t_i) }\right|_{\ve{g}(t_i)=\tve{g}(t_i)}&=\left\{
	\begin{aligned}
		&-\frac{\exp\left( \tilde{g}_{m}(t_i)\right)\left[\sum_{{{m}^{\ast} }=1, {{m}^{\ast} }\neq m}^{M-1}\exp\left( \tilde{g}_{{m}^{\ast} }(t_i)\right)+1 \right]}{\left(\sum_{{{m}^{\ast} }=1}^{M-1}\exp\left( \tilde{g}_{{m}^{\ast} }(t_i)\right)+1\right)^2},&\ \text{if}\ h=m,\\
		&\frac{\exp\left( \tilde{g}_{m}(t_i)\right)\exp\left(\tilde{g}_{h}(t_i)\right)}{\left(\sum_{{m}^{\ast} =1}^{M-1}\exp\left( \tilde{g}_{{m}^{\ast} }(t_i)\right)+1\right)^2},&\ \text{if}\ h \neq m.
	\end{aligned} \right.
\end{align*}

Then the probability density function $p\left( \gall \mid  \ve{Z}, \vesub{\theta}{g}\right)$ is proportional to $p\left(\ve{Z},\gall \mid  \vesub{\theta}{g}\right)$, given by \begin{align*}
	p\left( \gall \mid  \ve{Z}, \vesub{\theta}{g}\right)	\propto p\left(\ve{Z},\gall \mid  \vesub{\theta}{g}\right)&\propto \exp\left\{\ve{d}\left(\gall-\tilde{\gall}\right) -\frac{1}{2}\left(\gall-\tilde{\gall}\right)^T \ve{D}\left(\gall-\tilde{\gall}\right) -\frac{1}{2}\gall^{T} \vesup{C}{-1} \gall\right\}
	\\ &\propto \exp\left\{ -\frac{1}{2}\gall^{T} \left(\vesup{C}{-1}+ \ve{D} \right)\gall+\left(\ve{d}+\tilde{\gall}^{T}\ve{D} \right) \gall \right\}, 
\end{align*}

The unknown collections of the augmented curves $\gall$ and the hyper-parameters $\vesub{\theta}{g}$ are obtained using the Fisher scoring algorithm, detailed as follows.

\begin{algorithm}
	\caption{Estimate $\vesub{\theta}{g}$ and $\gall$ based on $\ve{Z}$ with the Fisher scoring algorithm.}\label{alg::thetag_g-estimation}
	\begin{algorithmic}[1]
	   \State Initialize $\vesub{\theta}{g}$. Set $w=0$ and initialize $\tilde{\gall}_{(0)}$. Calculate the first and second order derivatives $\vesub{d}{(0)} $, $\vesub{D}{(0)}$ with respect to $\ve{Z}$ and $\tilde{\gall}_{(0)}$.
		\State {\bf For} $w=1,2,\cdots,$  \begin{itemize}
			\item[] Obtain $\tilde{\gall}_{(w)}$ by finding the solution from $\left(\vesup{C}{-1}+ \vesub{D}{(w-1)} \right)\tilde{\gall}_{(w)}= \left(\vesub{d}{(w-1)}+\tilde{\gall}_{(w\text{-}1)}^{T}\vesub{D}{(w-1)}\right)^T$.
			\item[] Calculate the first and second order derivatives $\vesub{d}{(w)} $ and $\vesub{D}{(w)}$ with respect to $\ve{Z}$ and $\tilde{\gall}_{(w)}$.
 \end{itemize}
		  \State {\bf If} $\tilde{\gall}_{(w)}$ is converged, say at $\tilde{\gall}$ {\bf Break} 
		\State Update $\vesub{\theta}{g}=\argmax_{\vesub{\theta}{g}}\tilde{p}\left(\ve{Z}  \mid \vesub{\theta}{g} \right).$
		\State Update $\gall=\tilde{\gall}$.
	\end{algorithmic}
\end{algorithm}

In practice, estimating $\gall$ becomes computationally challenging when the sample size $n$ is large. To address the computational demand associated with this large dataset, we propose selecting a considerably smaller subset of time points $\bve{t}$ from the full set $\{t_i\}_{i=1}^n$ with $\underline{n}$ ($\underline{n}\ll n$) time points, and utilizing the sampled set of state sequences $\ve{Z}$ at these tunable time points ($\bve{Z}$) to estimate hyper-parameters $\vesub{\theta}{g}$ and the corresponding $\underline{\gall}$ at the selected time points. Then predict the functions at the remaining time points. Thus, for any time point $t^\ast$ that are not included in $\bve{t}$, $p\left(z(t^\ast)  \mid \bve{Z}, \vesub{\theta}{g}\right)$ or the corresponding $\ve{g}(t^\ast)$ can be estimated.

With the GP model trained, we predict the probability $p\left(z(t^\ast)  \mid \bve{Z}, \vesub{\theta}{g}\right)$ using the expectation $ \E\big[\E\big(z(t^\ast)  \mid \ve{g}({t^\ast}), \bve{Z}, \vesub{\theta}{g} \big) \big]$ for the time points $t^\ast$ that are not included in $\bve{t}$, where $\ve{g}({t^\ast})=\left(g_1(t^\ast),\cdots, g_{M-1}(t^\ast)\right)^T$. 

It follows that 
\begin{align*}
	p\left(z(t^\ast)  \mid \bve{Z},\vesub{\theta}{g}\right) & =\E\left[\E\left(z(t^\ast)  \mid \ve{g}({t^\ast}), \bve{Z}, \vesub{\theta}{g} \right) \right] \\
	&	= \int{\frac{\exp\left( g_{z(t^\ast) }(t^\ast )\mathbb{I}\left(z(t^\ast) \neq M\right) \right)}{\sum_{h=1}^{M-1}\exp\left( g_h(t^\ast)\right)+1}p\left(\ve{g}({t^\ast})  \mid \bve{Z},\vesub{\theta}{g} \right)} d \ve{g}({t^\ast}).
\end{align*}

Approximating $p\left(\ve{g}({t^\ast})  \mid \bve{Z},\vesub{\theta}{g} \right)$ using the Gaussian approximation $\tilde{p}_G( \ve{g}({t^\ast})  \mid \bve{Z},\vesub{\theta}{g})$ as discussed in Section 3.2 is a simple method, that is, 
\begin{align*}
	p\left(\ve{g}({t^\ast})  \mid \bve{Z},\vesub{\theta}{g} \right)&=\int{ 	p\left(\ve{g}({t^\ast})  \mid \underline{\gall},\vesub{\theta}{g} \right)	p\left(\underline{\gall}  \mid \bve{Z},\vesub{\theta}{g} \right) } d \underline{\gall} \\
	&\simeq \int{ 	p\left(\ve{g}({t^\ast})  \mid \underline{\gall},\vesub{\theta}{g} \right)	\tilde{p}_G\left(\underline{\gall}  \mid \bve{Z},\vesub{\theta}{g} \right) } d \underline{\gall}. 
\end{align*}

Since it is assumed that both $\ve{g}(t^\ast)$ and $\underline{\gall}$ come from the same GP, we have \begin{align*}
	\begin{pmatrix}
		\underline{\gall}\\ \ve{g}(t^\ast)
	\end{pmatrix}\sim N\left( \vesub{0}{(M-1)(\underline{n}+1)}\, ,\, \begin{pmatrix}
		\bve{C}& \bvesup{C}{\ast}\\
		\bvesup{C}{\ast T}& \vesup{C}{\ast}
	\end{pmatrix}\right),
\end{align*}
where $	\bve{C}$ is the covariance matrix of $\underline{\gall}$, $ \bvesup{C}{\ast}$ is an $(M-1)\underline{n} \times( M-1)$ matrix representing the covariance between $\underline{\gall}$ and $\ve{g}(t^\ast)$, and $\vesup{C}{\ast}=\diag\left(c_1(t^\ast, t^\ast),\cdots, c_{M-1}(t^\ast, t^\ast)\right)$. Thus, we have $ \left(\ve{g}(t^\ast)  \mid  \underline{\gall}\right) \sim N\left( \ve{B}\,\underline{\gall}\, ,\, \vesub{\Sigma}{g} \right),$ where $\ve{B}=\bvesup{C}{\ast T}\bvesup{C}{-1}$ and $ \vesub{\Sigma}{g}=\vesup{C}{\ast}-\bvesup{C}{\ast T}\bvesup{C}{-1}\bvesup{C}{\ast}$. From the discussion above, we have $\tilde{p}_G\left(\underline{\gall}  \mid \bve{Z},\vesub{\theta}{g} \right)=N\left(\underline{\tilde{\gall}}\, ,\, \left(\bvesup{C}{-1}+\underline{\ve{ D}} \right)^{-1}\right)$. The integrand for $	p\left(\ve{g}({t^\ast})  \mid \bve{Z},\vesub{\theta}{g} \right)$ is therefore the product of two Gaussian density functions and still gives out a Gaussian density function 
\begin{align*}
	p\left(\ve{g}({t^\ast})  \mid \bve{Z},\vesub{\theta}{g} \right)=N\left(\ve{B}\underline{\tilde{\gall}}\, ,\, \ve{B}\left(\bvesup{C}{-1}+\underline{\ve{ D}} \right)^{-1}\vesup{B}{T} +\vesub{\Sigma}{g} \right).
\end{align*}
Then $	p\left(z(t^\ast)  \mid \bve{Z},\vesub{\theta}{g}\right)$ can be evaluated through numerical integration.

\section{Proof of Identifiability}\label{SM-sec:Identifiability}

Define the following subset of $\Re$:
$$ \tilde{\mathcal{T}} =\{ t: f_m(t)=f_h(t) \text{ for some } h \neq m \}. $$
Because of the condition (8) stated in Theorem 1, any point in $\tilde{\mathcal{T}}$ is an isolated point, i.e., $\tilde{\mathcal{T}}$ has no limit point and contains at most countably many points. We hence denote the points in $\tilde{\mathcal{T}}$ as $\tilde{t}_1,\tilde{t}_2,\cdots,$ such that $\tilde{t}_i < \tilde{t}_{i+1}$. All points in $\tilde{\mathcal{T}}$ are isolated; thus, any point that lies within the interval between two neighboring points does not belong to $\tilde{\mathcal{T}}$, i.e., $(\tilde{t}_i, \tilde{t}_{i+1})\cap \tilde{\mathcal{T}} = \emptyset$, $i=1,2,\cdots$.

For any three time points $t_1, t_2, t_3$ that are not in $\tilde{\mathcal{T}}$ and such that $0<t_1<t_2<t_3$, let $y_1,y_2,y_3$ be the corresponding observations and $z_1,z_2,z_3$ be the hidden states. Then we have
$f_m(t_i)\neq f_h(t_i)$ for $m,h=1,\cdots,M$, $h\neq m$, and $i=1,2,3$. 
Given the initial state $z_0=1$ with the initial state probability $\pi_1=1$ and $t_0=0$, the distribution of $(y_1,y_2,y_3)$ is given by
\begin{align*}
	&p\left(y_1,y_2,y_3\right)\\
	=&\sum^M_{z_1,z_2,z_3=1} \left\{ p(y_1,y_2,y_3|z_1,z_2,z_3) 
	\left(\exp\big[\ve{Q}t_1 \big]_{1\,z_1}
	\exp\big[\ve{Q} (t_2-t_1)\big]_{z_1\,z_2}
	\exp\big[\ve{Q} (t_3-t_2)\big]_{z_2\, z_3} \right) \right\} \\
	=&\sum^M_{z_1,z_2,z_3=1} 
	\left\{ \phi_{z_1}(y_1)\phi_{z_2}(y_2)\phi_{z_3}(y_3) 
	\left(\exp\big[\ve{Q}t_1 \big]_{1\,z_1}
	\exp\big[\ve{Q} (t_2-t_1)\big]_{z_1\,z_2}
	\exp\big[\ve{Q} (t_3-t_2)\big]_{z_2\,z_3} \right) \right\} \\	
	=&\sum^M_{z_2=1} 
	\left(\sum^M_{z_1=1} \exp\big[\ve{Q}t_1 \big]_{1\,z_1} \exp\big[\ve{Q} (t_2-t_1)\big]_{z_1\,z_2} \phi_{z_1}(y_1) \right)
	\phi_{z_2}(y_2) \\
	& \qquad\qquad \left(\sum^M_{z_3=1} \exp\big[\ve{Q} (t_3-t_2)\big]_{z_2\,z_3} \phi_{z_3}(y_3) \right) \\
	=&\sum^M_{z_2=1} \exp\big[\ve{Q}t_1 \big]_{1\,z_2}
	\left(\sum^M_{z_1=1} \frac{\exp\big[\ve{Q}t_1 \big]_{1\,z_1}}{\exp\big[\ve{Q}t_1 \big]_{1\,z_2}} \exp\big[\ve{Q} (t_2-t_1)\big]_{z_1\,z_2} \phi_{z_1}(y_1) \right)
	\phi_{z_2}(y_2) \\
	& \qquad\qquad \left(\sum^M_{z_3=1} \exp\big[\ve{Q} (t_3-t_2)\big]_{z_2\,z_3} \phi_{z_3}(y_3) \right),
\end{align*}
where $\phi_{z_i}(y_i)$ is the probability density function of Gaussian distribution with mean $f_{z_i}(t_i)$ and variance $\sigma^2$, $i=1,2,3$.

Suppose $\tilde{\ve{Q}}$ is another transition rate matrix and $\{\tilde{f}_m(\cdot)\}^M_{m=1}$ are another set of nonlinear functions. Similarly we have
\begin{align*}
	p\left(y_1,y_2,y_3\right)
	& =\sum^M_{z_2=1} \exp\big[\tilde{\ve{Q}}t_1 \big]_{1\,z_2}
	\left(\sum^M_{z_1=1} \frac{\exp\big[\tilde{\ve{Q}}t_1 \big]_{1\,z_1}}{\exp\big[\tilde{\ve{Q}}t_1 \big]_{1\,z_2}} \exp\big[\tilde{\ve{Q}} (t_2-t_1)\big]_{z_1\,z_2} \tilde{\phi}_{z_1}(y_1) \right)
	\tilde{\phi}_{z_2}(y_2) \\
	& \qquad\qquad \left(\sum^M_{z_3=1} \exp\big[\tilde{\ve{Q}} (t_3-t_2)\big]_{z_2\,z_3} \tilde{\phi}_{z_3}(y_3) \right) . 	
\end{align*}	

Because $f_m(t_i)\neq f_h(t_i)$ for $h \neq m$, the Gaussian distributions $\{\phi_{z_i}(y_i)\}^M_{z_i=1}$ are linearly independent \citep{titterington1985}.
Since $\ve{Q}$ is irreducible and aperiodic, the transition probability matrices $\exp\big[\ve{Q}t_1 \big]$, $\exp\big[\ve{Q} (t_2-t_1)\big]$, and $\exp\big[\ve{Q} (t_3-t_2)\big]$ have full rank. As a result, both 
$\big\{\sum^M_{z_1=1} \exp\big[\ve{Q}t_1 \big]_{1\,z_1} \exp\big[\ve{Q} (t_2-t_1)\big]_{z_1\,z_2} \phi_{z_1}(y_1)/\exp\big[\ve{Q}t_1 \big]_{1\,z_2} \big\}^M_{z_2=1} $ and $\big\{\exp\big[\ve{Q} (t_3-t_2)\big]_{z_2\,z_3} \phi_{z_3}(y_3)\big\}^M_{z_2=1} $ are linearly independent. Thus, by Theorem 9 of \citet{allman2009identifiability} we have that, there exists a permutation of the set $\{1,\cdots,M\}$, say, $\tau=\{ \tau(1),\cdots,\tau(M) \} $ which depends on $t_1$, $t_2$, and $t_3$, such that, for all $m=1,\cdots,M$,
\begin{align}
	\tilde{\phi}_{m}(y_2) &= \phi_{\tau(m)}(y_2), \label{perm1} \\
	\sum^M_{h=1} \exp\big[\tilde{\ve{Q}} (t_3-t_2)\big]_{m\,h} \tilde{\phi}_{h}(y_3) 
	&= \sum^M_{h=1} \exp\big[\ve{Q} (t_3-t_2)\big]_{\tau(m)\,h} \phi_{h}(y_3) . \label{perm2}
\end{align}

As the variance $\sigma^2$ is constant, Equation \eqref{perm1} implies that $\tilde{f}_{m}(t_2) = f_{\tau(m)}(t_2)$, for all $m=1,\cdots,M$, and $t_2 \notin \tilde{\mathcal{T}}$.

By the same argument as \citet{huang2013nonparametric}, there exists a permutation $\tau$ of the set $\{1,\cdots,M\}$ that is independent of $t$, such that, for all $m=1,\cdots,M$,
$$ \tilde{f}_{m}(t) = f_{\tau(m)}(t) . $$
Hence, we have $\tilde{\phi}_{h}(y_3) = \phi_{\tau(h)}(y_3)$.
It yields from Equation \eqref{perm2} that
$$
\sum^M_{h=1} \exp\big[\tilde{\ve{Q}} (t_3-t_2)\big]_{m\,h} \phi_{\tau(h)}(y_3) 
= \sum^M_{h=1} \exp\big[\ve{Q} (t_3-t_2)\big]_{\tau(m)\,\tau(h)} \phi_{\tau(h)}(y_3) .
$$
By the linear independence of $\{\phi_{h}(y_3)\}^M_{h=1}$ we get that for all $m,h=1,\cdots,M$, 
$$\exp\big[\tilde{\ve{Q}} (t_3-t_2)\big]_{m\,h}
= \exp\big[\ve{Q} (t_3-t_2)\big]_{\tau(m)\,\tau(h)} ,
$$
which implies $\tilde{\ve{Q}}_{m\,h} = \ve{Q}_{\tau(m)\,\tau(h)} $.

\section{Proof of Consistency for $\sigma^2$}\label{SM-sec:Consistency-sigma}

Given the inverse-Gamma distribution prior assumption for $\sigma^2$, the conditional distribution for $(\sigma^2  \mid  \vesub{f}{\bs{z}},\ve{y})$ is expressed as follows:
\begin{align*}
	p\left(\sigma^2  \mid  \vesub{f}{\bs{z}},\ve{y} \right) \propto & p\left(\ve{y}  \mid  \vesub{f}{\bs{z}},\sigma^2 \right)p\left(\sigma^2 \right)\\ \propto&(\sigma^{2})^{-\frac{n}{2}}\exp \left\{-\frac{\left(\ve{y}-\vesub{f}{\bs{z}}\right)^T\left(\ve{y}-\vesub{f}{\bs{z}}\right)}{2\sigma^{2}}\right\}(\sigma^2)^{-(\alpha+1)}\exp\left\{-\frac{\beta}{\sigma^2}\right\} \\
	= & (\sigma^2)^{-(\frac{n}{2}+\alpha+1)}\exp\left\{-\frac{1}{\sigma^2}\left( \frac{\left(\ve{y}-\vesub{f}{\bs{z}}\right)^T\left(\ve{y}-\vesub{f}{\bs{z}}\right)}{2}+\beta \right)\right\}.
\end{align*} 
This results in the posterior distribution for $(\sigma^2  \mid  \vesub{f}{\bs{z}},\ve{y})$ still follows an inverse-Gamma distribution, with the expectation given by \begin{align*}
	\E(\sigma^2  \mid  \vesub{f}{\bs{z}},\ve{y})&=\left( \frac{\left(\ve{y}-\vesub{f}{\bs{z}}\right)^T\left(\ve{y}-\vesub{f}{\bs{z}}\right)}{2}+\beta \right)\Big/\left(\frac{n}{2}+\alpha-1\right)\\
	&=\frac{1}{n+2\alpha-2}\left(\ve{y}-\vesub{f}{\bs{z}}\right)^T\left(\ve{y}-\vesub{f}{\bs{z}}\right)+\frac{2\beta}{n+2\alpha-2}. 
\end{align*}

The posterior mean $\E\left(\sigma^2  \mid  \ve{y}\right)$ is defined as $\E_{( \vesub{f}{\bs{z}}|\ve{y})}\left[\E(\sigma^2  \mid  \vesub{f}{\bs{z}},\ve{y})\right]$, leading to:
\begin{align*}
	\E\left(\sigma^2  \mid  \ve{y}\right)&=\frac{1}{n+2\alpha-2}\int{\left(\ve{y}-\vesub{f}{\bs{z}}\right)^T\left(\ve{y}-\vesub{f}{\bs{z}}\right)p\left(\vesub{f}{\bs{z}} \mid  \ve{y}\right) }d\vesub{f}{\bs{z}} + \frac{2\beta}{n+2\alpha-2} \\
	&=\frac{1}{n+2\alpha-2}\E\left[ \vess{f}{\bs z}{T}\vesub{f}{\bs z}+\vesup{y}{T}\ve{y} -2 \vesup{y}{T}\vesub{f}{\bs z} \mid  \ve{y}\right] + \frac{2\beta}{n+2\alpha-2}.
\end{align*}

Since a multivariate GP prior is posited for $\fall$, conditioning on the hidden state sequence $\ve{z}$, $\vesub{f}{\bs{z}}$ follows a multivariate Gaussian distribution characterized by zero mean and the covariance matrix $\vesub{K}{\bs z}$. Then, the marginal distribution for $\vesub{f}{\bs{z}}$ aligns with a mixture of Gaussian distributions given by $\sum_{\bs {z}}p(\ve{z})N\left(\vesub{0}{n}, \vesub{K}{\bs{z}} \right)$, with $\sum_{\bs{z}}p(\ve{z})=1$. The posterior distribution remains Gaussian mixture distribution: 
$$
\left(\vesub{f}{\bs z} \mid  \ve{y}\right) \sim \sum_{\bs {z}}p(\ve{z})N\left( \vesub{K}{\bs{z}}\left( \vesub{K}{\bs{z}}+\sigma^2_0\vesub{I}{n}\right)^{-1}\ve{y}\ , \ \sigma^2_0 \vesub{K}{\bs{z}}\left( \vesub{K}{\bs{z}}+\sigma^2_0\vesub{I}{n}\right)^{-1}\right).$$ Thus, we obtain: \begin{align*}
	\E\left[ \vesub{f}{\bs z} \mid  \ve{y}\right]&= \sum_{\bs {z}}p(\ve{z}) \vesub{K}{\bs{z}}\left( \vesub{K}{\bs{z}}+\sigma^2_0\vesub{I}{n}\right)^{-1}\ve{y},\\
	\E\left[ \vesub{f}{\bs z}\vess{f}{\bs z}{T} \mid  \ve{y}\right]&=\sum_{\bs z}p(\ve{z})\left\{ \sigma^2_0 \vesub{K}{\bs{z}}\left( \vesub{K}{\bs{z}}+\sigma^2_0\vesub{I}{n}\right)^{-1}+\vesub{K}{\bs{z}}\left( \vesub{K}{\bs{z}}+\sigma^2_0\vesub{I}{n}\right)^{-1}\ve{y}\vesup{y}{T}\left( \vesub{K}{\bs{z}}+\sigma^2_0\vesub{I}{n}\right)^{-1}\vesub{K}{\bs{z}} \right\}.
\end{align*}

Since $	\E\left[ \vess{f}{\bs z}{T}\vesub{f}{\bs z}+\vesup{y}{T}\ve{y} -2 \vesup{y}{T}\vesub{f}{\bs z} \mid  \ve{y}\right]=\tr\left( \E \left[\vesub{f}{\bs z}\vess{f}{\bs z}{T}+\ve{y}\vesup{y}{T} -2 \vesub{f}{\bs z}\vesup{y}{T} \mid  \ve{y}\right] \right)$, we first derive that: \begin{align*}
	&\E\left[ \vess{f}{\bs z}{T}\vesub{f}{\bs z}+\vesup{y}{T}\ve{y} -2 \vesup{y}{T}\vesub{f}{\bs z} \mid  \ve{y}\right]\\
	=&\sum_{\bs z}p(\ve{z}) \Bigg\{ \tr\bigg(\sigma^2_0 \vesub{K}{\bs{z}}\left( \vesub{K}{\bs{z}}+\sigma^2_0\vesub{I}{n}\right)^{-1}+\vesub{K}{\bs{z}}\left( \vesub{K}{\bs{z}}+\sigma^2_0\vesub{I}{n}\right)^{-1}\ve{y}\vesup{y}{T}\left( \vesub{K}{\bs{z}}+\sigma^2_0\vesub{I}{n}\right)^{-1}\vesub{K}{\bs{z}} +\ve{y}\vesup{y}{T}\\
	&\qquad\qquad\qquad\qquad -2\vesub{K}{\bs{z}}\left( \vesub{K}{\bs{z}}+\sigma^2_0\vesub{I}{n}\right)^{-1}\ve{y}\vesup{y}{T} \bigg) \Bigg\}\\
	=& \sum_{\bs z}p(\ve{z}) \Bigg\{ \sigma_0^2 \tr\bigg(\vesub{K}{\bs{z}}\left( \vesub{K}{\bs{z}}+\sigma^2_0\vesub{I}{n}\right)^{-1} \bigg)\\ 
	&\qquad\qquad\qquad\qquad  + \tr\bigg(\left( \vess{K}{\bs{z}}{2}\left( \vesub{K}{\bs{z}}+\sigma^2_0\vesub{I}{n}\right)^{-2}-2\vesub{K}{\bs{z}}\left( \vesub{K}{\bs{z}}+\sigma^2_0\vesub{I}{n}\right)^{-1} + \vesub{I}{n}\right) \ve{y}\vesup{y}{T} \bigg) \Bigg\} \\
	=& \sigma_0^2 \sum_{\bs z}p(\ve{z}) \Bigg\{ \tr\bigg(\vesub{K}{\bs{z}}\left( \vesub{K}{\bs{z}}+\sigma^2_0\vesub{I}{n}\right)^{-1} \bigg) +\tr\bigg(\sigma_0^2\left( \vesub{K}{\bs{z}}+\sigma^2_0\vesub{I}{n}\right)^{-2} \ve{y}\vesup{y}{T} \bigg) \Bigg\}.
\end{align*} 
As the marginal distribution for the observations is a Gaussian mixture distribution, such that \begin{align*}
	\ve{y}\sim \sum_{\bs z}p(\ve{z})N(\vesub{0}{n}\, , \, \vesub{K}{z}+\sigma_0^2\vesub{I}{n}), 
\end{align*}
it holds the convergence property $ \frac{1}{n}\sum_{\bs z}p(\ve{z})\vesup{y}{T}\left( \vesub{K}{z}+\sigma_0^2\vesub{I}{n} \right)^{-1}\ve{y} \stackrel{p}{\rightarrow} 1 $ as $n\rightarrow +\infty$. 

Consequently, we find that, as $n$ is sufficiently large,
\begin{align*}
	& \frac{1}{n} \E\left[ \vess{f}{\bs z}{T}\vesub{f}{\bs z}+\vesup{y}{T}\ve{y} -2 \vesup{y}{T}\vesub{f}{\bs z} \mid  \ve{y}\right] \\
	\approx & 
	\sigma_0^2 \sum_{\bs z}p(\ve{z}) \Bigg\{ \tr\bigg(\vesub{K}{\bs{z}}\left( \vesub{K}{\bs{z}}+\sigma^2_0\vesub{I}{n}\right)^{-1} \bigg) +\tr\bigg(\sigma_0^2\vesub{I}{n}\left( \vesub{K}{z}+\sigma_0^2\vesub{I}{n} \right)^{-1} \bigg) \Bigg\}\\
	=&n\sigma_0^2
\end{align*} 
with probability 1.

Thus, $$\E\left(\sigma^2  \mid  \ve{y}\right)\stackrel{p}{\longrightarrow} \frac{n \sigma_0^2}{n+2\alpha-2} + \frac{2\beta}{n+2\alpha-2},$$ which converges to the true value $\sigma_0^2$, as $n\rightarrow +\infty$.

\section{Proof of Consistency for $\vesub{\theta}{z}$}\label{SM-sec:Consistency-theta_z}

We first prove the following lemma.

\begin{lemma} \label{Lemma:partial_q}
	For any given distribution $\ve{\pi}$ and transition probability matrix $\vesub{P}{1}$, $\vesub{P}{2}$, we have \begin{align*}
		\vesup{\pi}{T} \vesub{P}{1} \frac{\partial \exp \left[\ve{Q} t \right]}{ \partial q_{mh} } \vesub{P}{2} \vesub{1}{M}=0, \text{ for all } m,h =1,\cdots ,M,\ h \neq m. 
	\end{align*}
	
	\begin{proof}[Proof of Lemma \ref{Lemma:partial_q}]
		The first derivative of $\exp \left[\ve{Q} t \right]$ with respect to $q_{mh}$ have the following form \citep{wilcox1967exponential}: \begin{align*}
			\frac{\partial \exp \left[\ve{Q} t \right]}{ \partial q_{mh} }=t\int_{0}^1{\exp\left[ \ve{Q} ut\right]\frac{\partial \ve{Q}}{ \partial q_{mh} } \exp\left[ \ve{Q} (1-u)t\right] } d u.
		\end{align*} Recalling $q_{mm}=-\sum_{h \neq m} q_{mh}$, we have $\frac{\partial \ve{Q}}{ \partial q_{mh} }=\vesub{e}{m}\left(\vesub{e}{h}-\vesub{e}{m} \right)^T$, and then \begin{align*}
			\vesup{\pi}{T} \vesub{P}{1} \frac{\partial \exp \left[\ve{Q} t \right]}{ \partial q_{mh} } \vesub{P}{2} \vesub{1}{M}&=t\vesup{\pi}{T} \vesub{P}{1}\int_{0}^1{\exp\left[ \ve{Q} ut\right]\vesub{e}{m}\left(\vesub{e}{h}-\vesub{e}{m} \right)^T \exp\left[ \ve{Q} (1-u)t\right] } d u \vesub{P}{2} \vesub{1}{M} \\
			&=t\int_{0}^1{\vesup{\pi}{T} \vesub{P}{1}\exp\left[ \ve{Q} ut\right]\vesub{e}{m}\left(\vesub{e}{h}-\vesub{e}{m} \right)^T \exp\left[ \ve{Q} (1-u)t\right] \vesub{P}{2} \vesub{1}{M} }d u. 
		\end{align*}
		By the definition of generator matrix, $\exp\left[ \ve{Q} ut\right]$ and $\exp\left[ \ve{Q} (1-u)t\right] $ are transition probability matrices from time $0$ to $ut$ and $(1-u)t$ respectively. Hence both $\vesub{P}{1}\exp\left[ \ve{Q} ut\right]$ and $\exp\left[ \ve{Q} (1-u)t\right]\vesub{P}{2}$ are transition probability matrices and the summation of each row equals $1$. That is, $$\vess{e}{h}{T} \exp\left[ \ve{Q} (1-u)t\right] \vesub{P}{2} \vesub{1}{M}=\vess{e}{m}{T} \exp\left[ \ve{Q} (1-u)t\right]\vesub{P}{2} \vesub{1}{M}=1,$$ which implies $$\vesup{\pi}{T} \vesub{P}{1}\exp\left[ \ve{Q} ut\right]\vesub{e}{m}\left(\vesub{e}{h}-\vesub{e}{m} \right)^T \exp\left[ \ve{Q} (1-u)t\right] \vesub{P}{2} \vesub{1}{M} =0,$$ and the proof is completed. 
	\end{proof}
\end{lemma}

We impose the following regularity conditions: 
\begin{itemize}
	\item[(C1)] The matrix $\frac{\partial^2 \log p(\ve{y}|\vesub{\theta}{z} )}{\partial \vesub{\theta}{z} \partial \vess{\theta}{z}{T} }$ is continuous in $\vesub{\theta}{z}$ throughout some neighborhood of the true values $\vesub{\theta}{z0}$, and the Fisher information matrix $\vesub{B}{n}=-\E\left[\left. \frac{\partial^2\log p(\ve{y} | \vesub{\theta}{z})}{\partial \vesub{\theta}{z}\partial \vess{\theta}{z}{T}} \right|_{\vesub{\theta}{z}=\vesub{\theta}{z0}}\right]$ is positive definite. 
	\item[(C2)] There exist a $\delta>0$ and some sequence $\{\lambda_n \}$, neither of them depending on $\vesub{\theta}{z}$, such that, as $n\rightarrow +\infty$ and $|\vesub{\theta}{z}-\vesub{\theta}{z0}|< \delta,$
	\begin{equation*}
		p\left(-\lambda_n^{-\frac{1}{2}}(\vesub{\theta}{z}-\vesub{\theta}{z0})^T\vess{B}{n}{-\frac{1}{2}} \left. \frac{\partial^2 \log p(\ve{y}|\vesub{\theta}{z} )}{\partial \vesub{\theta}{z} \partial \vess{\theta}{z}{T}}
		\right|_{\vesub{\theta}{z}=
			\tilde{\ve{\theta}}_{z0}}(\vesub{\theta}{z}-\vesub{\theta}{z0})\geq |\vesub{\theta}{z}-\vesub{\theta}{z0}|^2 \right)\rightarrow 1,
	\end{equation*} 
	where $\tilde{\ve{\theta}}_{z0}= \vesub{\theta}{z0}+a(\vesub{\theta}{z} - \vesub{\theta}{z_0} ) $ for some $0\leq a \leq 1$.
\end{itemize}

In addition to (C1)-(C2), we also need to show that
the probability $p\left(\ve{y} \mid  \vesub{\theta}{z}\right)$, also the likelihood function of the observations, satisfies the following regularity conditions for maximum likelihood estimator \citep{lehmann1983theory,crowder1976maximum}:
\begin{itemize}
	\item[(C3)] The first and second derivatives of the logarithm of the density $p\left(\ve{y} \mid \vesub{\theta}{z} \right)$ satisfy the equations \begin{align*}
		\E\left[\frac{\partial \log p\left(\ve{y} \mid \vesub{\theta}{z} \right)}{\partial q_{(a)}}\right]&=0, \\
		\E\left[\frac{\partial \log p\left(\ve{y} \mid \vesub{\theta}{z} \right)}{\partial q_{(a)}}\frac{\partial \log p\left(\ve{y} \mid \vesub{\theta}{z} \right)}{\partial q_{(b)}}\right]&=-\E\left[\frac{\partial^2 \log p\left(\ve{y} \mid \vesub{\theta}{z} \right)}{\partial q_{(a)} \partial q_{(b)} }\right],
	\end{align*} for $q_{(a)}$, $q_{(b)}$ $\in \vesub{\theta}{z}$.
\end{itemize}

For the proof of (C3), we first note that 
\begin{align*}
	\E\left[\frac{\partial \log p\left(\ve{y} \mid \vesub{\theta}{z} \right)}{\partial q_{(a)}}\right]=\int{ \frac{1}{p\left(\ve{y} \mid \vesub{\theta}{z} \right)}\frac{\partial p\left(\ve{y} \mid \vesub{\theta}{z} \right)}{\partial q_{(a)}} p\left(\ve{y} \mid \vesub{\theta}{z} \right) } d \ve{y}=\int{\frac{\partial p\left(\ve{y} \mid \vesub{\theta}{z} \right)}{\partial q_{(a)}} }d \ve{y},
\end{align*} 
where \begin{align*}
	\frac{\partial p\left(\ve{y} \mid \vesub{\theta}{z} \right)}{\partial q_{(a)}}=&\sum_{\ve{z}} \Bigg\{ p\left(\ve{y} \mid  \ve{z}\right)\sum_{i=1}^n\Bigg[ \Big(\prod_{j=1}^{i-1}\vess{e}{z_{j-1}}{T} \exp\big[\ve{Q} (t_j-t_{j-1}) \big]\vesub{e}{z_j} \Big)\vess{e}{z_{i-1}}{T} \frac{\partial \exp\big[\ve{Q} (t_i-t_{i-1}) \big]}{\partial q_{(a)}}\vesub{e}{z_i}\\
	& \qquad\qquad \Big(\prod_{j=i+1}^{n}\vess{e}{z_{j-1}}{T} \exp\big[\ve{Q} (t_j-t_{j-1}) \big]\vesub{e}{z_j} \Big) \vess{e}{z_n}{T} \exp\big[\ve{Q} (1-t_{n}) \big]\vesub{1}{M} \Bigg] \Bigg\}.
\end{align*} 
With $\int{ p\left(\ve{y} \mid  \ve{z} \right) }d \ve{y}=1 $, we have \begin{align*}
	&\E\left[\frac{\partial \log p\left(\ve{y} \mid \vesub{\theta}{z} \right)}{\partial q_{(a)}}\right]\\
	=&\sum_{\ve{z}} \sum_{i=1}^n \Bigg[ \Big(\prod_{j=1}^{i-1}\vess{e}{z_{j-1}}{T} \exp\big[\ve{Q} (t_j-t_{j-1}) \big]\vesub{e}{z_j} \Big)\vess{e}{z_{i-1}}{T} \frac{\partial \exp\big[\ve{Q} (t_i-t_{i-1}) \big]}{\partial q_{(a)}}\vesub{e}{z_i}\\
	& \qquad\qquad\qquad \Big(\prod_{j=i+1}^{n}\vess{e}{z_{j-1}}{T} \exp\big[\ve{Q} (t_j-t_{j-1}) \big]\vesub{e}{z_j} \Big) \vess{e}{z_n}{T} \exp\big[\ve{Q} (1-t_{n}) \big]\vesub{1}{M} \Bigg]\\
	=& \sum_{i=1}^n\Bigg[ \vess{e}{z_{0}}{T}\Big(\prod_{j=1}^{i-1} \exp\big[\ve{Q} (t_j-t_{j-1}) \big] \Big) \frac{\partial \exp\big[\ve{Q} (t_i-t_{i-1}) \big]}{\partial q_{(a)}} \\
	&\qquad\qquad \Big(\prod_{j=i+1}^{n}\exp\big[\ve{Q} (t_j-t_{j-1}) \big]\Big) \exp\big[\ve{Q} (1-t_{n}) \big]\vesub{1}{M} \Bigg] \\
	=& \sum_{i=1}^n\Bigg[ \vess{e}{z_{0}}{T} \exp\big[\ve{Q} t_{i-1} \big] \frac{\partial \exp\big[\ve{Q} (t_i-t_{i-1}) \big]}{\partial q_{(a)}} \exp\big[\ve{Q} (1-t_{i}) \big]\vesub{1}{M} \Bigg].
\end{align*} 
By Lemma \ref{Lemma:partial_q}, we have that $\E\left[\frac{\partial \log p(\ve{y}|\vesub{\theta}{z} )}{\partial q_{(a)}}\right]=0$ holds for $q_{(a)} \in \vesub{\theta}{z}$. 

Furthermore, taking derivative of its expectation with respect to $q_{(b)} \in \vesub{\theta}{z}$, we obtain that \begin{align*}
	\frac{\partial}{\partial q_{(b)}}\E\left[\frac{\partial \log p\left(\ve{y} \mid \vesub{\theta}{z} \right)}{\partial q_{(a)}}\right]
	=\frac{\partial }{\partial q_{(b)}} \int{ \frac{\partial \log p\left(\ve{y} \mid \vesub{\theta}{z} \right)}{\partial q_{(a)}} p\left(\ve{y} \mid \vesub{\theta}{z} \right) } d \ve{y}=0.
\end{align*} By the boundedness of the derivatives, we can change the order of derivation and integration, and we obtain \begin{align*}
	\E\left[\frac{\partial \log p\left(\ve{y} \mid \vesub{\theta}{z} \right)}{\partial q_{(a)}}\frac{\partial \log p\left(\ve{y} \mid \vesub{\theta}{z} \right)}{\partial q_{(b)}}\right]+\E\left[\frac{\partial^2 \log p\left(\ve{y} \mid \vesub{\theta}{z} \right)}{\partial q_{(a)} \partial q_{(b)} }\right]=0.
\end{align*} Thus the proof for (C3) is completed. 

It yields by (C3) that
$$\E\left(\left| \vess{B}{n}{-\frac{1}{2}}\left. \frac{\partial \log p(\ve{y}|\vesub{\theta}{z} )}{\partial \vesub{\theta}{z}}\right|_{\vesub{\theta}{z}=\vesub{\theta}{z0}} \right|^2 \right)=M(M-1). $$ 
Consequently, we have
\begin{equation}
	\lambda_n^{-\frac{1}{2}}\vess{B}{n}{-\frac{1}{2}}\left. \frac{\partial \log p(\ve{y}|\vesub{\theta}{z} )}{\partial \vesub{\theta}{z}}\right|_{\vesub{\theta}{z}=\vesub{\theta}{z0}} \rightarrow \ve{0},\ \text{when} \ |\vesub{\theta}{z}-\vesub{\theta}{z0}|<\delta. 
	\label{eq .Appen::thetaz}
\end{equation}

By Taylor expansion,
$$\vess{B}{n}{-\frac{1}{2}}\frac{\partial \log p(\ve{y}|\vesub{\theta}{z} )}{\partial \vesub{\theta}{z}}=\vess{B}{n}{-\frac{1}{2}}\left. \frac{\partial \log p(\ve{y}|\vesub{\theta}{z} )}{\partial \vesub{\theta}{z}}\right|_{\vesub{\theta}{z}=\vesub{\theta}{z0}}+\vess{B}{n}{-\frac{1}{2}}\left. \frac{\partial^2 \log p(\ve{y}|\vesub{\theta}{z} )}{\partial \vesub{\theta}{z}\partial \vess{\theta}{z}{T} }\right|_{\vesub{\theta}{z}=\tilde{\ve{\theta}}_{z0}}( \vesub{\theta}{z}-\vesub{\theta}{z0} ).$$ 

Now we define \begin{equation}
	\vesub{a}{n}=\lambda_n^{-\frac{1}{2}}\vess{B}{n}{-\frac{1}{2}}\left. \frac{\partial \log p(\ve{y}|\vesub{\theta}{z} )}{\partial \vesub{\theta}{z}}\right|_{\vesub{\theta}{z}=\vesub{\theta}{z0}} +\lambda_n^{-\frac{1}{2}}\vess{B}{n}{-\frac{1}{2}}\left[ \left. \frac{\partial^2 \log p(\ve{y}|\vesub{\theta}{z} )}{\partial \vesub{\theta}{z} \partial \vess{\theta}{z}{T} }\right|_{\vesub{\theta}{z}=\tilde{\ve{\theta}}_{z0}}\right](\vesub{\theta}{z}-\vesub{\theta}{z_0}) \label{eqAppendx::an}, 
\end{equation} 
Condition (C2) and Equation \eqref{eq .Appen::thetaz} imply that, given $|\vesub{\theta}{z}-\vesub{\theta}{z_0}|<\delta$, 
\begin{align*}
	p\left( (\vesub{\theta}{z}-\vesub{\theta}{z_0})^T\vesub{a}{n}<0\right)
	\rightarrow 1,\ \text{as}\ n \rightarrow +\infty.
\end{align*}
It follows, using an equivalence of Brouwer's fixed point theorem as in \cite{aitchison1958maximum}, that, with probability tending to 1 as $n \rightarrow +\infty $, $\vesub{a}{n}$ has a zero, denoted by $\vesub{\hat\theta}{z}$, with $|\vesub{\hat\theta}{z}-\vesub{\theta}{z0}|<\delta$. Thus there exists a weakly consistent solution $\vesub{\hat\theta}{z}$ of the likelihood equation 
$\partial\log p\left(\ve{y} \mid  \vesub{\theta}{z}\right)/\partial \vesub{\theta}{z} = 0 $.

\section{Proof of Information Consistency for $\hve{y}$} \label{SM-sec:Consistency-y}

In this proof, we use the covariance functions to define a set of functions over the time range $\mathcal{T}$. The space of such functions is known as a reproducing kernel Hilbert space (RKHS). Let $\mathcal{H}$ denote the RKHS associated with the covariance functions $\{k_{mh}(\cdot, \cdot)\}_{m,h=1}^M$, as defined in Equation \eqref{SM-eq:cov-cross}. Let $\mathcal{H}_n$ be the linear span of the collection $\left\{k_{mh}(\cdot,t_i),i=1,\cdots, n \right\}_{m,h=1}^M, $ such that 
\begin{align*} 
	\mathcal{H}_n=\left\{ (f_{1}(\cdot), \cdots, f_{M}(\cdot))^T: f_m(t)=\sum_{i=1}^n\sum_{h=1}^M \alpha_{ih}k_{mh}(t,t_i),\ \forall \alpha_{ih} \in \Re,\ m=1,\cdots, M\right\}. 
\end{align*}
We first assume that the true underlying function $\vesub{f}{0}(\cdot)=(f_{10}(\cdot), \cdots, f_{M0}(\cdot))^T \in \mathcal{H}_n$, then $f_{m 0}(\cdot)$ can be expressed as 
\begin{align*}
	f_{m0}(\cdot)=\sum_{i=1}^n\sum_{h=1}^M \rho_{ih}k_{mh}(\cdot,t_i)=\left[\vesub{k}{m}(\cdot,t_1),\cdots,\vesub{k}{m}(\cdot,t_n) \right]\ve{\rho},
\end{align*}
where $\ve{\rho}=\left(\vess{\rho}{1}{T},\cdots, \vess{\rho}{n}{T}\right)^T$ with $\vesub{\rho}{i}=\left(\rho_{i1},\cdots, \rho_{iM}\right)^T$, and $\vesub{k}{m}(\cdot,t_i)=( k_{m1}(\cdot,t_i),\cdots,k_{mM}(\cdot,t_i))$, for $i=1,\cdots, n$. The $nM\times 1$ vector $\fall_{0}$ denotes the realizations of the true functions $\{f_{m0}(\cdot)\}_{m=1}^M$ at all time points $\{t_i\}_{i=1}^n$, such that $\fall_{0}=\left(\vesub{f}{0}(t_1)^T,\cdots,\vesub{f}{0}(t_n)^T \right)^T$. Referring to the definition for the covariance matrix $\ve{K}$, such that $\ve{K}=\{\ve{K}(t_i,t_j)\}_{i,j=1}^n$, the $(m,h)$-th entry of the block matrix is expressed as $$\ve{K}(t_i,t_j)=\left(\begin{matrix}
	k_{11}(t_i,t_j) & \cdots & 	k_{1M}(t_i,t_j)\\
	\vdots&\ddots&\vdots\\
	k_{M1}(t_i,t_j) & \cdots & 	k_{MM}(t_i,t_j)
\end{matrix}\right)=\left(\begin{matrix}
	\vesub{k}{1}(t_i,t_j)\\
	\vdots\\
	\vesub{k}{M}(t_i,t_j)
\end{matrix}\right).$$ 
By the properties of RKHS, $\|\vesub{f}{0}\|_{\bs K}^2= \vesup{\rho}{T}\ve{K}\ve{\rho}$ and $\fall_{0}=\ve{K}\ve{\rho}$. 

Incorporating information from the hidden state process, the underlying function for the observational data is $f_{z}(\cdot)=f_{z(\cdot)}(\cdot)$, with the true function represented as $f_{z(\cdot)0}(\cdot)$. Given $z(t)=m$, it follows that $f_{z(t)0}(t)=f_{m0}(t)$. Based on the estimated hidden state sequence $\ve{z}$, which captures the realizations of $z(t)$ over the time points $\{t_i\}_{i=1}^n$, such that $\ve{z}=(z_1,\cdots,z_n)^T$ with $z_i=z(t_i)$, we denote
$\boldsymbol{f}_{\bs{z}}=\left( f_{z_1}(t_1),\cdots, f_{z_n}(t_n)\right)^T$ and
$\vesub{f}{\bs{z}0}=\left( f_{z_1 0}(t_1),\cdots, f_{z_n0}(t_n)\right)^T$ as the realizations of the true function across all time points $\{t_i\}_{i=1}^n$. Utilizing the definition for $\vesub{K}{\bs{z}}$, such that $\vesub{K}{\bs{z}}=\{k_{z_i z_j}(t_i,t_j)\}_{i,j=1}^n$, we have $\vesub{f}{\bs{z}0}=\vesub{K}{\bs{z}}\vesub{\rho}{\bs{z}}$, where $\vesub{\rho}{\bs{z}}=\left(\rho_{1z_1},\cdots, \rho_{nz_n}\right)^T$ is the subvector of $\ve{\rho}$.

Let $P$ and $\bar{P}$ be any two measures on $\mathcal{F}=\{f(\cdot): \mathcal{T}\rightarrow \Re\}$, then it yields by the Fenchel-Legendre duality relationship that, for any functional $\phi(\cdot)$ on $\mathcal{F}$, 
$$\E_{\bar P}\left[\phi(f_{z}(\cdot)) \right]\leq \log \E_P\left[ \exp\left(\phi(f_{z}(\cdot))\right)\right]+\text{KL}[\bar{P},P].$$ 

Let us consider the following:
\begin{enumerate}
	\item Set $\phi(f_z(\cdot))$ be $\log p\left( \ve{y} \mid  f_{z}(\cdot),{z}\right)$ for any specified hidden state sequence ${\ve{z}}$ and $f_z(\cdot) \in \mathcal{F} $,
	
	\item Let $P$ be the measure induced by the convolution GP, hence the distribution for $\vesub{f}{\bs z}$ is $\tilde{p}(\vesub{f}{\bs z})=N\left(\vesub{0}{n}, \hvesub{K}{\bs{z}} \right)$. We then have $\E_P\left[ \exp\left(\phi(f_z(\cdot))\right)\right]=\int{ p\left(\ve{y} \mid  \vesub{f}{\bs{z}},\bs{z}\right) }p\left(\vesub{f}{\bs{z}}|\bs{z}\right) d \vesub{f}{\bs{z}}=p_{cgp}\left(\ve{y} \mid  \ve{z}\right)$, where $\hvesub{K}{\bs{z}}$ is defined in the same way as $\vesub{K}{\bs{z}}$ but with $\vesub{\theta}{f}$ being replaced by its estimator $\hvesub{\theta}{f}$,
	
	\item Let $\bar{P}$ be the posterior distribution of $f_z(\cdot)$ on $\mathcal{F}$ which has a convolution GP prior distribution and Gaussian likelihood $\prod_{i=1}^{n}N\left(\hat{y}_i  \mid  f_{z_i}(t_i),\hat{\sigma}^2 \right)$, where \begin{align*}
		\hve{y}=\left(\vesub{K}{\bs{z}}+\hat{\sigma}^2 \vesub{I}{n}\right)\vesub{\rho}{\bs{z}}.
	\end{align*} 
	Here $\hat{\sigma}^2$ is a constant to be specified. Therefore, the posterior of $\vesub{f}{\bs{z}}$ is \begin{equation}
		\begin{aligned}
			\bar{p}\left(\vesub{f}{\bs{z}}\right)
			\stackrel{\triangle}{=} &p\left(\vesub{f}{\bs{z}} \mid  \hve{y}\right)\\
			=& N\left( \vesub{K}{\bs{z}}\left( \vesub{K}{\bs{z}}+\hat{\sigma}^2 \vesub{I}{n}\right)^{-1}\hve{y}\ , \ \vesub{K}{\bs{z}}- \vesub{K}{\bs{z}}\left( \vesub{K}{\bs{z}}+\hat{\sigma}^2 \vesub{I}{n}\right)^{-1}\vesub{K}{\bs{z}} \right)\\
			=& N\left(\vesub{K}{\bs{z}}\vesub{\rho}{\bs{z}}\ , \ \hat{\sigma}^2\vesub{K}{\bs{z}}\left(\vesub{K}{\bs{z}} +\hat{\sigma}^2 \vesub{I}{n}\right)^{-1} \right).
		\end{aligned} \label{eq::barP}
	\end{equation}
\end{enumerate} 

It follows that 
\begin{align*}
	\text{KL}\left[\bar{P}, P\right]=&\int_{\mathcal{F}} \log \frac{d \bar{P}}{d P} d\bar{P}\\
	=&\int_{\Re^n} \bar{p}\left(\vesub{f}{\bs{z}}\right)\log \frac{\bar{p}\left(\vesub{f}{\bs{z}}\right) }{ \tilde{p}\left(\vesub{f}{\bs{z}}\right)} d \vesub{f}{\bs{z}}\\
	=& \frac{1}{2}\Bigg\{ -\log\left|\hvess{K}{\bs{z}}{-1}\vesub{K}{\bs{z}}\right| -\log\left|\hat{\sigma}^2\left( \vesub{K}{\bs{z}}+\hat{\sigma}^2 \vesub{I}{n}\right)^{-1} \right|\\
	&\qquad +\tr\left[\hat{\sigma}^2\hvess{K}{\bs{z}}{-1}\vesub{K}{\bs{z}}\left( \vesub{K}{\bs{z}}+\hat{\sigma}^2 \vesub{I}{n}\right)^{-1} \right] +\left(\vesub{K}{\bs{z}}\vesub{\rho}{\bs{z}}\right)^T\hvess{K}{\bs{z}}{-1}\vesub{K}{\bs{z}}\vesub{\rho}{\bs{z}}-n \Bigg\} \\
	=& \frac{1}{2}\Bigg\{ -\log\left|\hvess{K}{\bs{z}}{-1}\vesub{K}{\bs{z}}\right| +\log\left|\frac{1}{\hat{\sigma}^2}\vesub{K}{\bs{z}}+ \vesub{I}{n} \right|+\tr\left[\hat{\sigma}^2\hvess{K}{\bs{z}}{-1}\vesub{K}{\bs{z}}\left( \vesub{K}{\bs{z}}+\hat{\sigma}^2 \vesub{I}{n}\right)^{-1} \right]\\ 
	&\qquad +\vess{\rho}{\bs{z}}{T}\vesub{K}{\bs{z}}\vesub{\rho}{\bs{z}} +\left(\vesub{K}{\bs{z}}\vesub{\rho}{\bs{z}}\right)^T\left( \hvess{K}{\bs{z}}{-1}\vesub{K}{\bs{z}}-\vesub{I}{n}\right)\vesub{\rho}{\bs{z}}-n \Bigg\}.
\end{align*}

On the other hand, 
$$\E_{\bar{P}}\left[\phi(\vesub{f}{\bs{z}}) \right]=\E_{\bar{P}}\left[\log p\left( \ve{y} \mid \vesub{f}{\bs{z}}\right) \right]=\sum_{i=1}^n\E_{\bar{P}}\left[\log p\left( {y}_i \mid  {f}_{z_i}(t_i)\right) \right] . $$ 
By Taylor's expansion, expanding $\log p\left( {y}_i \mid  {f}_{z_i}(t_i)\right)$ to the second order at the true value $f_{z_i 0}(t_i)$ yields 
\begin{align*}
	\log p\left( {y}_i \mid  {f}_{z_i}(t_i)\right)=&\log p\left( {y}_i \mid  {f}_{z_i0}(t_i)\right)+ \left. \frac{d \log p\left( {y}_i \mid  {f}_{z_i}(t_i)\right)}{ d {f}_{z_i}(t_i)} \right|_{{f}_{z_i}(t_i)={f}_{z_i0}(t_i)}\left({f}_{z_i}(t_i)-{f}_{z_i0}(t_i) \right)\\
	& \qquad + \frac{1}{2} \left. \frac{d^2 \log p\left( {y}_i \mid  {f}_{z_i}(t_i)\right)}{ \left[ d {f}_{z_i}(t_i) \right]^2} \right|_{{f}_{z_i}(t_i)=\tilde{f}_{z_i}(t_i) }\left({f}_{z_i}(t_i)-{f}_{z_i0}(t_i) \right)^2,
\end{align*} where $\tilde{f}_{z_i}(t_i)={f}_{z_i0}(t_i)+a\left[{f}_{z_i}(t_i)-{f}_{z_i0}(t_i) \right]$ for some $0\leq a\leq 1$. 

For $\log p\left( {y}_i \mid  {f}_{z_i}(t_i)\right)=-\frac{1}{2}\log(2\pi {\sigma}^2)-\frac{1}{2{\sigma}^2}\left( {y}_i-{f}_{z_i}(t_i) \right)^2$, we have 
\begin{align*}
	\log p\left( {y}_i \mid  {f}_{z_i}(t_i)\right)=&\log p\left( {y}_i \mid  {f}_{z_i0}(t_i)\right)+\frac{1}{{\sigma}^2}\left(y_i-{f}_{z_i0}(t_i) \right)\left({f}_{z_i}(t_i)-{f}_{z_i0}(t_i) \right) \\
    & - \frac{1}{2{\sigma}^2} \left({f}_{z_i}(t_i)-{f}_{z_i0}(t_i) \right)^2,
\end{align*} 
and it follows that 
\begin{align*}
	\E_{\bar{P}}\left[\log p\left( {y}_i \mid  {f}_{z_i}(t_i)\right) \right]=&\log p\left( {y}_i \mid  {f}_{z_i0}(t_i)\right)+\frac{1}{{\sigma}^2}\left(y_i-{f}_{z_i0}(t_i) \right)\E_{\bar{P}}\left[\left({f}_{z_i}(t_i)-{f}_{z_i0}(t_i) \right) \right]\\
	& \qquad - \frac{1}{2{\sigma}^2}\E_{\bar{P}}\left[\left({f}_{z_i}(t_i)-{f}_{z_i0}(t_i) \right)^2\right]. 
\end{align*}
From Equation \eqref{eq::barP}, we have that $$ \bar{p}\left({f}_{z_i}(t_i)\right)
= N\left( {f}_{z_i0}(t_i)\ , \ \hat{\sigma}^2 \left[\vesub{K}{\bs{z}}\left(\vesub{K}{\bs{z}}+\hat{\sigma}^2 \vesub{I}{n}\right)^{-1}\right]_{z_i\, z_i} \right).$$ Therefore, $\E_{\bar{P}}\left[\left({f}_{z_i}(t_i)-{f}_{z_i0}(t_i) \right) \right]=0$ and $ \E_{\bar{P}}\left[\left({f}_{z_i}(t_i)-{f}_{z_i0}(t_i) \right)^2 \right]=\hat{\sigma}^2 \left[\vesub{K}{\bs{z}}\left(\vesub{K}{\bs{z}}+\hat{\sigma}^2 \vesub{I}{n}\right)^{-1}\right]_{z_i\, z_i} $. Thus we have \begin{align*}
	\sum_{i=1}^n\E_{\bar{P}}\left[\log p\left( {y}_i \mid  {f}_{z_i}(t_i)\right) \right]=& \sum_{i=1}^n \log p\left( {y}_i \mid  {f}_{z_i 0}(t_i)\right)-\frac{\hat{\sigma}^2}{2\sigma^2}\tr\left[\vesub{K}{\bs{z}}\left(\vesub{K}{\bs{z}}+\hat{\sigma}^2 \vesub{I}{n}\right)^{-1} \right] \\
	=& \underbrace{\log p\left( \ve{y} \mid \vesub{f}{\bs{z} 0} \right) }_{\log p_0 ( \bs{y}|\bs{z} )} -\frac{\hat{\sigma}^2}{2\sigma^2}\tr\left[\vesub{K}{\bs{z}}\left(\vesub{K}{\bs{z}}+\hat{\sigma}^2 \vesub{I}{n}\right)^{-1} \right].
\end{align*}

Combining these results yields 
	\begin{align}
		&-\log p_{cgp}\left(\ve{y} \mid  \ve{z}\right)+\log p_0\left( \ve{y} \mid  \ve{z} \right)   \nonumber \\
		\leq& \frac{1}{2}\Bigg\{ -\log\left|\hvess{K}{\bs{z}}{-1}\vesub{K}{\bs{z}}\right| +\log\left|\frac{1}{\hat{\sigma}^2}\vesub{K}{\bs{z}}+ \vesub{I}{n} \right| + \tr\left[\hat{\sigma}^2 \hvess{K}{\bs{z}}{-1}\vesub{K}{\bs{z}}\left( \vesub{K}{\bs{z}}+\hat{\sigma}^2 \vesub{I}{n}\right)^{-1} \right]  \nonumber  \\ 
		&\qquad +\frac{\hat{\sigma}^2}{2 {\sigma}^2}\tr\left[\vesub{K}{\bs{z}}\left( \vesub{K}{\bs{z}}+\hat{\sigma}^2 \vesub{I}{n}\right)^{-1} \right] +\vess{\rho}{\bs{z}}{T}\vesub{K}{\bs{z}}\vesub{\rho}{\bs{z}} +\left(\vesub{K}{\bs{z}}\vesub{\rho}{\bs{z}}\right)^T\left( \hvess{K}{\bs{z}}{-1}\vesub{K}{\bs{z}}-\vesub{I}{n}\right)\vesub{\rho}{\bs{z}}-n \Bigg\} \nonumber \\
		= & \frac{1}{2}\Bigg\{ -\log\left|\hvess{K}{\bs{z}}{-1}\vesub{K}{\bs{z}}\right| +\log\left|\frac{1}{ \hat{\sigma}^2 }\vesub{K}{\bs{z}}+ \vesub{I}{n} \right|+\tr\left[ \big(\hat{\sigma}^2 \hvess{K}{\bs{z}}{-1}\vesub{K}{\bs{z}}+ \frac{\hat{\sigma}^2}{2{\sigma}^2}\vesub{K}{\bs{z}}\big) \left( \vesub{K}{\bs{z}}+\sigma^2 \vesub{I}{n}\right)^{-1} \right]  \nonumber \\ 
		&\qquad + \vess{\rho}{\bs{z}}{T}\vesub{K}{\bs{z}}\vesub{\rho}{\bs{z}} +\left(\vesub{K}{\bs{z}}\vesub{\rho}{\bs{z}}\right)^T\left( \hvess{K}{\bs{z}}{-1}\vesub{K}{\bs{z}}-\vesub{I}{n}\right)\vesub{\rho}{\bs{z}}-n \Bigg\}. \label{eq:-logp_cgp+logp_0}
	\end{align}
Since the covariance function is bounded and continuous in $\vesub{\theta}{f}$, and as the estimator $\hvesub{\theta}{f} \rightarrow \vesub{\theta}{f}$ as $n\rightarrow +\infty$, we have that $\hvess{K}{\bs{z}}{-1}\vesub{K}{\bs{z}} \rightarrow \vesub{I}{n}$ as $n\rightarrow +\infty$. Therefore, there exist some positive constants $\lambda$ and ${\delta}$ such that: \begin{align*}
	&-\log\left|\hvess{K}{\bs{z}}{-1}\vesub{K}{\bs{z}}\right|<\lambda,\qquad \left(\vesub{K}{\bs{z}}\vesub{\rho}{\bs{z}}\right)^T\left( \hvess{K}{\bs{z}}{-1}\vesub{K}{\bs{z}}-\vesub{I}{n}\right)\vesub{\rho}{\bs{z}}<\lambda,\\
	&\tr\left[ \big(\hat{\sigma}^2\hvess{K}{\bs{z}}{-1}\vesub{K}{\bs{z}}+ \frac{\hat{\sigma}^2}{2 {\sigma}^2}\vesub{K}{\bs{z}}\big) \left( \vesub{K}{\bs{z}}+\hat{\sigma}^2 \vesub{I}{n}\right)^{-1} \right]
	< \tr\left[ \big(\hat{\sigma}^2\vesub{I}{n} +\hat{\sigma}^2(\delta+ \frac{1}{2 {\sigma}^2})\vesub{K}{\bs{z}} \big) \left( \vesub{K}{\bs{z}}+\hat{\sigma}^2 \vesub{I}{n}\right)^{-1} \right].
\end{align*} 
Next, consider the right-hand side of Equation \eqref{eq:-logp_cgp+logp_0}: 
\begin{align*}
	\mathrm{RHS} <& \frac{1}{2}\vess{\rho}{\bs{z}}{T}\vesub{K}{\bs{z}}\vesub{\rho}{\bs{z}} + \frac{1}{2}\log\left|\frac{1}{\hat{\sigma}^2}\vesub{K}{\bs{z}}+ \vesub{I}{n} \right|+\lambda  \\
    & + \frac{1}{2}\left\{\tr\left[ \big(\hat{\sigma}^2\vesub{I}{n} +\hat{\sigma}^2(\delta+ \frac{1}{2 {\sigma}^2})\vesub{K}{\bs{z}} \big) \left( \vesub{K}{\bs{z}}+\hat{\sigma}^2 \vesub{I}{n}\right)^{-1} \right] -n \right\} \\
	<& \frac{c}{2}\|\vesub{f}{0} \|_{\bs K}^2 + \frac{1}{2}\log\left|\frac{1}{\hat{\sigma}^2}\vesub{K}{\bs{z}}+ \vesub{I}{n} \right|+\lambda \\
    & + \frac{1}{2}\left\{\tr\left[ \big(\hat{\sigma}^2\vesub{I}{n} +\hat{\sigma}^2(\delta+ \frac{1}{2 {\sigma}^2})\vesub{K}{\bs{z}} \big) \left( \vesub{K}{\bs{z}}+\hat{\sigma}^2 \vesub{I}{n}\right)^{-1} \right] -n \right\},
\end{align*}
by Lemma \ref{Lemma:rhokrho} below.

Note that the above inequality holds for all $\hat{\sigma}^2>0$. By letting $\hat{\sigma}^2=\frac{2\sigma^2}{2\sigma^2\delta+1}$, the right-hand side of Equation \eqref{eq:-logp_cgp+logp_0} simplified to $ \frac{c}{2}\|\vesub{f}{0} \|_{\bs K}^2 + \frac{1}{2}\log\left|(\delta+ \frac{1}{2 {\sigma}^2}) \vesub{K}{\bs{z}}+ \vesub{I}{n} \right|+\lambda$. Taking infimum over $\vesub{f}{0}$ and applying the representer theorem, we obtain 
$$-\log p_{cgp}\left(\ve{y} \mid  \ve{z}\right)+\log p_0\left( \ve{y} \mid  \ve{z} \right) \leq \frac{c}{2}\|\vesub{f}{0} \|_{\bs K}^2 + \frac{1}{2}\log\left|(\delta+ \frac{1}{2 {\sigma}^2}) \vesub{K}{\bs{z}}+ \vesub{I}{n} \right|+\lambda, $$ and \begin{equation}
	\frac{1}{n}\E_{\ve{t}}\left(\text{KL}\left[p_0\left(\ve{y} \mid  \ve{z}\right), p_{cgp}\left(\ve{y} \mid  \ve{z}\right)\right]\right) \leq \frac{c}{2n}\|\vesub{f}{0} \|_{\bs K}^2+ \frac{1}{2n}\E_{\ve{t}}\left( \log\left| (\delta+ \frac{1}{2 {\sigma}^2}) \vesub{K}{\bs{z}}+ \vesub{I}{n} \right| \right)+\frac{\lambda}{n}, \label{eq:bound}
\end{equation} where $c$, $\delta$, and $\lambda$ are positive constants. The proof is complete.

\begin{lemma} \label{Lemma:rhokrho}
	We have $\vess{\rho}{\bs{z}}{T}\vesub{K}{\bs{z}}\vesub{\rho}{\bs{z}}\leq c \|\vesub{f}{0}\|_{\bs K}^2 $ for some positive constant $c$.
	
	\begin{proof}
		As $t \in \mathcal{T}$ and $\mathcal{T}$ is a compact set in $\Re$, given the covariance functions in Equation \eqref{SM-eq:cov-cross}, we have \begin{align*}
			0<c_1 \leq k_{mh}(\cdot,\cdot) \leq c_2,\ m,h=1\cdots, M,
		\end{align*}where $c_1$ and $c_2$ are two constants. The covariance matrix $\ve{K}$ can be eigendecomposited as $\ve{K}=\vesup{\Psi}{T}\ve{\Lambda}\ve{\Psi}$, where $\vesup{\Psi}{T} \ve{\Psi}=\vesub{I}{nM}$ and $\ve{\Lambda}=\diag\left(\lambda_1,\cdots, \lambda_{nM} \right)$ with $\lambda_1\geq \lambda_2\geq \cdots, \geq \lambda_{nM} > 0$. Thus \begin{align*}
			\|\vesub{f}{0}\|_{\bs K}^2=\vesup{\rho}{T}\ve{K}\ve{\rho}=\vesup{\rho}{T}\vesup{\Psi}{T}\ve{\Lambda}\ve{\Psi}\ve{\rho}= \sum_{i=1}^{nM}\lambda_{i} \left(\vess{\psi}{i}{T}\ve{\rho}\right)^2\geq \lambda_{nM}\sum_{i=1}^{nM}\left(\vess{\psi}{i}{T}\ve{\rho}\right)^2 = nM\lambda_{nM} \|\ve{\rho}\|_2^2,
		\end{align*} where $\vesub{\psi}{i}$ is the $i$-th column of $\ve{\Psi}$. Therefore, $\|\ve{\rho}\|_2^2 \leq \frac{1}{nM\lambda_{nM}} \|\vesub{f}{0}\|_{\bs K}^2 $. Since $\vesub{\rho}{\bs z}$ is the subvector of $\ve{\rho}$, we have $\|\vesub{\rho}{\bs z} \|_2^2 \leq \|\ve{\rho}\|_2^2$ and there exist a constant $c_3$, $0<c_3\leq \frac{1}{nM\lambda_{nM}}$, such that $\|\vesub{\rho}{\bs z} \|_2^2 \leq c_3 \|\vesub{f}{0}\|_{\bs K}^2$. We have 
		\begin{equation}
			\vess{\rho}{\bs{z}}{T}\vesub{K}{\bs{z}}\vesub{\rho}{\bs{z}}=\sum_{i,j=1}^n \rho_{i z_i}k_{z_i z_j}(t_i,t_j)\rho_{j z_j} \leq c_2 \|\vesub{\rho}{\bs z} \|_2^2 \leq c \|\vesub{f}{0}\|_{\bs K}^2, \label{eq::normfz}
		\end{equation} where $c$ is a constant and $c=c_2 c_3$.
	\end{proof}
\end{lemma}

\section{Proof of Information Consistency for $\hve{z}$} \label{SM-sec:Consistency-z}

In this proof, we use the same notation $\mathcal{H}$ to be the RKHS, associated with the covariance function $\ve{\mathrm{C}}(\cdot,\cdot)=\diag\left(c_1(\cdot, \cdot), \cdots, c_{M-1}(\cdot,\cdot) \right)$. Let $\mathcal{H}_n$ be the linear span of $\{c_{m}(\cdot,t_i), i=1,\cdots, n, m\in\{1,\cdots, M-1\} \},$ such that 
\begin{align*} 
	\mathcal{H}_n=\left\{ (g_1(\cdot), \cdots, g_{M-1}(\cdot))^T: g_m(t)=\sum_{i=1}^n \alpha_{im}c_{m}(t,t_i),\ \forall\alpha_{im} \in \Re,\ m=1,\cdots, M-1\right\}. 
\end{align*}
We begin by assuming that the true underlying function $\vesub{g}{ 0}(\cdot)=\left( {g}_{10}(\cdot),\cdots, g_{M\text{-}1, 0}(\cdot)\right)^T \in \mathcal{H}_n $. Then it can be expressed as
\begin{align*}
	\vesub{g}{0}(\cdot)=\left[\ve{\mathrm{C}}(\cdot,t_1),\cdots,\ve{\mathrm{C}}(\cdot,t_n) \right]\ve{\rho},
\end{align*}
where $\ve{\rho}=\left(\vess{\rho}{1}{T},\cdots, \vess{\rho}{n}{T}\right)^T$, with $\vesub{\rho}{i}=\left(\rho_{i1},\cdots, \rho_{iM\text{-}1}\right)^T$ for $i=1,\cdots, n$. The matrix $\ve{\mathrm{C}}(\cdot,t_i)$ is defined as $\diag( c_{1}(\cdot,t_i),\cdots, c_{M\text{-}1}(\cdot,t_i))$. Denote the realizations of the true functions $\{g_{m0}(\cdot)\}_{m=1}^{M-1}$ at all time points $\{t_i\}_{i=1}^n$ as $\gall_{0}=\left(\vesub{g}{0}(t_1)^T,\cdots,\vesub{g}{0}(t_n)^T \right)^T$. Referring to the definition for the covariance matrix $\ve{C}$, such that $\ve{C}=\left\{\ve{\mathrm{C}}(t_i,t_j)\right\}_{i,j=1}^n$, and utilizing the properties of RKHS, the norm for $\ve{g}_0(\cdot)$ in this space satisfies $\|\vesub{g}{0} \|_{\bs C}^2 =\vesup{\rho}{T}\ve{C}\ve{\rho}$, and $\gall_{0}=\ve{C}\ve{\rho}$.

Denote $\ve{g}(\cdot)=(g_1(\cdot), \cdots, g_{M-1}(\cdot))^T $ and $\gall=\left(\ve{g}(t_1)^T,\cdots,\ve{g}(t_n)^T \right)^T$. 
Let $P$ and $\bar{P}$ be any two measures defined on $\mathcal{G}=\{\ve{g}(\cdot): \mathcal{T}\rightarrow \Re^{M-1} \}$. Then, by the Fenchel-Legendre duality relationship, we obtain: for any functional $\phi(\cdot)$ on $\mathcal{G}$, $$\E_{\bar P}\left[\phi(\ve{g} ) \right]\leq \log \E_P\left[ \exp\left(\phi(\ve{g})\right)\right]+\text{KL}[\bar{P},P].$$ 

More specifically, 
\begin{enumerate}
	\item Set $\phi(\ve{g})=\log p\left( \ve{z} \mid  \ve{g} \right)$ for $\ve{g} \in \mathcal{G} $. 
	
	\item Let $P$ be the measure induced by GP, then the distribution for $\gall$ is $p(\gall)=N\left(\vesub{0}{n(M\text{-}1)}, \hve{C} \right)$ and $\E_P\left[ \exp\left(\phi(\ve{g})\right)\right]=\int{ p\left(\ve{z} \mid  \ve{g}\right) }p\left(\ve{g}\right) d \ve{g}=p_{gp}\left(\ve{z}\right)$, with $\hve{C}$ defined in the same way as $\ve{C}$ but with $\vesub{\theta}{g}$ being replaced by its estimator $\hvesub{\theta}{g}$.
	
	\item Referring to Section \ref{SM-sec:predict_new_z_g}, $p\left(
	\gall \mid  \ve{Z},\vesub{\theta}{g}\right)\propto \exp\left\{ -\frac{1}{2}\gall^{T} \left(\vesup{C}{-1}+ \ve{D} \right)\gall+\left(\ve{d}+\tilde{\gall}^{T}\ve{D} \right) \gall \right\}$, where \begin{align*}
		\ve{D}=-\left.\frac{\partial^2 \log p( \bs{Z}| \gall) }{\partial \gall\partial \gall^T}\right|_{\gall=\tilde{\gall}}=- \sum_{s \in \mathcal{S}}\left. \frac{\partial^2 \log p(\vesub{z}{(s)} | \gall )}{\partial\gall \partial\gall^{T} } \right|_{\gall=\tilde{\gall}}=\mathbb{S}\ve{D}(\tilde{\gall})
	\end{align*} with $\ve{D}(\tilde{\gall})=-\left. \frac{\partial^2 \log p(\vesub{z}{(s)} | \gall )}{\partial\gall \partial\gall^{T} } \right|_{\gall=\tilde{\gall}}$ concerning solely on the mode $\tilde{\gall}$. Let $\bar{P}$ be the Gaussian approximate distribution for $\gall$ is given by 
	\begin{equation}
		\begin{aligned}
			\bar{p}\left(\gall\right)
			\stackrel{\triangle}{=} & N\left( \gall_0\ , \ \left(\vesup{C}{-1}+\ve{D}\right)^{-1} \right)= N\left(\ve{C}\ve{\rho} \ , \ \left(\vesup{C}{-1}+\ve{D}\right)^{-1} \right).
		\end{aligned} \label{eq::barP-z}
	\end{equation}
\end{enumerate} 

It follows that 
\begin{align*}
	\text{KL}\left[\bar{P}, P\right]=&\int_{\mathcal{G}} \log \frac{d \bar{P}}{d P} d\bar{P}\\
	=&\int_{\Re^{n(M-1)}} \bar{p}\left(\gall\right)\log 
	\frac{\bar{p}\left(\gall\right) }{p\left(\gall\right)} d \gall\\
	=& \frac{1}{2}\Bigg\{ -\log\left|\hvesup{C}{-1}\ve{C}\right | +\log\left| \vesub{I}{n(M\text{-}1)}+\ve{DC}\right|+\tr\left[\hvesup{C}{-1}\ve{C}\left( \vesub{I}{n(M\text{-}1)}+\ve{DC} \right)^{-1} \right]\\
	&\qquad +\vesup{\rho}{T}\ve{C}\hvesup{C}{-1}\ve{C}\ve{\rho}-n(M-1) \Bigg\}.
\end{align*}

On the other hand, we have $\E_{\bar{P}}\left[\phi(\ve{g}) \right]=\E_{\bar{P}}\left[\log p\left( \ve{z} \mid \ve{g}\right) \right]=\sum_{i=1}^n\E_{\bar{P}}\left[\log p\left( {z}_i \mid  \ve{g}(t_i)\right) \right] $. Using Taylor's expansion, expanding $\log p\left( {z}_i \mid  \ve{g}(t_i)\right)$ to the second order at the true value $\vesub{g}{0}(t_i)$ yields \begin{align*}
	\log p\left( {z}_i \mid \ve{g}(t_i)\right)=&\log p\left( {z}_i \mid  \vesub{g}{0}(t_i)\right)+ \left. \frac{\partial \log p\left( {z}_i \mid \ve{g}(t_i)\right) }{ \partial \ve{g}(t_i)^T } \right|_{{\bs g}(t_i)={\bs g}_{0}(t_i)}\left(\ve{g}(t_i)-\vesub{g}{0}(t_i) \right)\\
	& \qquad + \frac{1}{2}\left(\ve{g}(t_i)-\vesub{g}{0}(t_i) \right)^T \left. \frac{\partial \log p\left( {z}_i \mid \ve{g}(t_i)\right) }{ \partial \ve{g}(t_i)\partial \ve{g}(t_i)^T } \right|_{{ \bs g}(t_i)={ \bs g}^\ast(t_i) } \left(\ve{g}(t_i)-\vesub{g}{0}(t_i) \right),
\end{align*} where ${ \bs g}^\ast(t_i)={\bs g}_{0}(t_i)+a({\bs g}(t_i)-{\bs g}_{0}(t_i) )$ for some $0\leq a\leq 1$. Hence, we have \begin{align*}
	&\E_{\bar{P}}\left[\log p\left( {z}_i \mid \ve{g}(t_i)\right) \right]\\
	=&\log p\left( {z}_i \mid  \vesub{g}{0}(t_i)\right)+ \left. \frac{\partial \log p\left( {z}_i \mid \ve{g}(t_i)\right) }{ \partial \ve{g}(t_i)^T } \right|_{{\bs g}(t_i)= {\bs g}_{0}(t_i)}\E_{\bar{P}}\left[\ve{g}(t_i)-\vesub{g}{0}(t_i)\right]\\
	&\quad +\frac{1}{2} \tr \left[ \E_{\bar{P}}\left[\left. \frac{\partial^2 \log p\left( {z}_i \mid \ve{g}(t_i)\right) }{ \partial \ve{g}(t_i)\partial \ve{g}(t_i)^T } \right|_{{\bs g}(t_i)={ \bs g}^\ast(t_i) }\left(\ve{g}(t_i)-\vesub{g}{0}(t_i) \right)^T\left(\ve{g}(t_i)-\vesub{g}{0}(t_i) \right) \right]\right]. 
\end{align*}

For 
$$\log p\left( {z}_i \mid \ve{g}(t_i)\right)=g_{z_i }(t_i )\mathbb{I}\left(z_i \neq M\right) -\log\left( \sum_{m=1}^{M-1}\exp\left( g_m(t_i)\right)+1\right), $$ 
denoting $ A=\sum_{m=1}^{M-1}\exp\left( g_m(t_i)\right)+1 $, we can derive that
\begin{align*}
&	\frac{\partial \log p\left( {z}_i \mid \ve{g}(t_i)\right) }{ \partial \ve{g}(t_i)^T }=\begin{pmatrix}
		\mathbb{I}(z_i =1)-\frac{\exp\left(g_1(t_i) \right)}{A} ,\cdots, \mathbb{I}(z_i =M-1)-\frac{\exp\left(g_{M\text{-}1}(t_i) \right)}{ A}
	\end{pmatrix}, \\
&	\frac{\partial^2 \log p\left( {z}_i \mid \ve{g}(t_i)\right) }{ \partial \ve{g}(t_i)\partial \ve{g}(t_i)^T } \\
&= \frac{1}{A^2}\left(\begin{smallmatrix}
		\exp\left( 2 g_1(t_i)\right)-A\exp\left( g_1(t_i)\right) &\exp\left( g_1(t_i)+g_2(t_i)\right) &\cdots & \exp\left( g_1(t_i)+g_{M\text{-}1}(t_i)\right) \\
		\exp\left( g_2(t_i)+g_1(t_i)\right) & 	\exp\left( 2 g_2(t_i)\right)-A\exp\left( g_2(t_i)\right) & \cdots & \exp\left( g_2(t_i)+g_{M\text{-}1}(t_i)\right) \\
		\vdots&\vdots & \ddots & \vdots \\
		\exp\left( g_{M\text{-}1}(t_i)+g_{1}(t_i)\right) & 	\exp\left( g_{M\text{-}1}(t_i)+g_{2}(t_i)\right)& \cdots & 	\exp\left( 2 g_{M\text{-}1}(t_i)\right)-A\exp\left( g_{M\text{-}1}(t_i)\right) 
	\end{smallmatrix}\right)\\
&	=  \frac{1}{A^2}\exp\left( \ve{g}(t_i)\right) \exp\left(\ve{g}(t_i)\right)^T-\frac{1}{A}\exp\left(\diag\left( \ve{g}(t_i)\right) \right). 
\end{align*} 
It is easy to demonstrate that 
\begin{equation*}
	-\vess{1}{M\text{-}1}{T} \leq \frac{\partial \log p\big( {z}_i \mid \ve{g}(t_i)\big) }{ \partial \ve{g}(t_i)^T }\leq \vess{1}{M\text{-}1}{T} , \end{equation*}
and by incorporating the lower bound for the negative second-order derivative established in Lemma \ref{postive_for_partial2} we have
\begin{equation}
	\bs{0}_{M\text{-}1}< -\frac{\partial^2 \log p\big( {z}_i \mid \ve{g}(t_i)\big) }{ \partial \ve{g}(t_i)\partial \ve{g}(t_i)^T } \leq \frac{1}{A}\exp\left(\diag\left(\ve{g}(t_i)\right)\right) \leq \vesub{I}{M-1} . \label{eq:bound_second_order_deriviative}
\end{equation}

From Equation \eqref{eq::barP-z}, we have \begin{align*}
	\bar{p}\left(\ve{g}(t_i)\right)
	\stackrel{\triangle}{=} & N\left(\vesub{g}{0}(t_i)\ , \ \left(\vesup{C}{-1}+\ve{D}\right)^{-1}_{[i,i]} \right),
\end{align*} 
where $\left(\vesup{C}{-1}+\ve{D}\right)^{-1}_{[i,i]}$ denotes the submatrix of $\left(\vesup{C}{-1}+\ve{D}\right)^{-1}$ that spans the rows and columns from $\left((i-1)(M-1)+1\right)$ to $i(M-1)$.

Utilizing the properties of moment generating functions for multivariate Gaussian distribution and the positive definiteness of the covariance matrix $\ve{C}$, we can derive that
\begin{align*}
	&\E_{\bar{P}}\left[\log p\left( {z}_i \mid  \ve{g}(t_i)\right) \right]\\
	\geq & \log p\left( {z}_i \mid  \vesub{g}{0}(t_i)\right)-\vess{1}{M\text{-}1}{T} \mid \E_{\bar{P}}\left[\ve{g}(t_i)-\vesub{g}{0}(t_i) \right] \mid  -\frac{1}{2} \tr \Bigg[ \E_{\bar{P}}\left[ \left(\ve{g}(t_i)-\vesub{g}{0}(t_i) \right)^T\left(\ve{g}(t_i)-\vesub{g}{0}(t_i) \right) \right] \Bigg] \\
	=& \log p\left( {z}_i \mid  \vesub{g}{0}(t_i)\right)-\frac{1}{2}\tr\left[ \left(\vesup{C}{-1}+\ve{D}\right)^{-1}_{[i,i]} \right] .
\end{align*}

Thus, we have \begin{align*}
	\E_{\bar{P}}\left[\log p\left( \ve{z} \mid \gall\right) \right]
	\geq \log p\left( \ve{z} \mid \gall_{0}\right) -\frac{1}{2} \tr\left[\left(\vesup{C}{-1}+\ve{D}\right)^{-1} \right].
\end{align*}

Combining these results yields
\begin{align}
	&-\log p_{gp}\left(\ve{z}\right)+\log p_0\left( \ve{z} \right) \nonumber \\
	\leq& \frac{1}{2}\Bigg\{ -\log\left|\hvesup{C}{-1}\ve{C}\right | +\log\left| \vesub{I}{n(M\text{-}1)}+\ve{DC}\right| + \vesup{\rho}{T}\ve{C}\hvesup{C}{-1}\ve{C}\ve{\rho} \nonumber \\
	&\qquad +\tr\left[\hvesup{C}{-1}\ve{C}\left( \vesub{I}{n(M\text{-}1)}+\ve{DC} \right)^{-1}+\left( \vesup{C}{-1}+\ve{D} \right)^{-1} \right]-n(M-1) \Bigg\} \nonumber \\
	\leq& 
	\frac{1}{2}\Bigg\{ -\log\left|\hvesup{C}{-1}\ve{C}\right | +\log\left| \vesub{I}{n(M\text{-}1)}+\ve{DC}\right|+\vesup{\rho}{T}\ve{C}\big(\hvesup{C}{-1}\ve{C}-\vesub{I}{n(M\text{-}1)}\big )\ve{\rho}+\|\vesub{g}{0}\|_{\bs C}^2 \nonumber \\
	&\qquad +\tr\left[\hvesup{C}{-1}\ve{C}\left( \vesub{I}{n(M\text{-}1)}+\ve{DC} \right)^{-1}+\ve{C}\left( \vesub{I}{n(M\text{-}1)}+\ve{DC} \right)^{-1} \right]-n(M-1)\Bigg\} .\label{eq:-logp_cgp+logp_0-z} 
\end{align}

Since the covariance function is bounded and continuous in $\vesub{\theta}{g}$, and $\hvesub{\theta}{g} \rightarrow \vesub{\theta}{g}$ as $n\rightarrow +\infty$, we have $\hvesup{C}{-1}\ve{C} \rightarrow \vesub{I}{n(M-1)}$ as $n\rightarrow +\infty$. Therefore, there exist some positive constants $\lambda$ and $\delta$ such that 
\begin{align*}
	&-\log\left|\hvesup{C}{-1}\ve{C}\right|<\lambda,\quad \vesup{\rho}{T}\ve{C}\big(\hvesup{C}{-1}\ve{C}-\vesub{I}{n(M\text{-}1)}\big )\ve{\rho}<\lambda,\\
	&\tr\left[\big(\hvesup{C}{-1}\ve{C}+\ve{C}\big)\big( \vesub{I}{n(M\text{-}1)}+\ve{DC} \big)^{-1} \right]< \tr\left[\big( \vesub{I}{n(M\text{-}1)}+(\delta+1)\ve{C} \big)\big( \vesub{I}{n(M\text{-}1)}+\ve{DC} \big)^{-1} \right].
\end{align*}

The submatrix of $\ve{D}(\tilde{\gall})$ that spans the rows and columns from $((i-1)(M-1)+1)$ to $i(M-1)$ is given by $\ve{D}(\tilde{\gall})_{[i,i]}=-\frac{1}{\tilde{A}^2}\exp\left({\tve{g}}(t_i)\right)\exp\left({\tve{g}}(t_i)\right)^T+\frac{1}{\tilde{A}}\exp\left(\diag\left({\tve{g}}(t_i)\right) \right)$, where $\tilde{A}=\sum_{m=1}^{M-1}\exp\left( \tilde{g}_m(t_i)\right)+1$. Referring to Equation \eqref{eq:bound_second_order_deriviative}, we can derive that $\vesub{0}{n(M-1)}<\ve{D}(\tilde{\gall})\leq \vesub{I}{n(M-1)}$. Thus there exists a constant $c$ with $0 < c<1$ such that $\ve{D}(\tilde{\gall})\geq c \vesub{I}{n(M-1)}>\vesub{0}{n(M-1)}$. Taking the proposal size $\mathbb{S}\geq \frac{\delta+1}{c}$ to be sufficiently large, it can satisfy $\vesub{0}{n(M-1)}\leq \left(\delta+1 \right)\vesub{I}{n(M-1)}\leq \ve{D}$ as $n\rightarrow +\infty$. This leads to $$\tr\left[\left( \vesub{I}{n(M\text{-}1)}+(\delta+1)\ve{C} \right)\left( \vesub{I}{n(M\text{-}1)}+\ve{DC} \right)^{-1} \right] -n(M-1) \leq 0.$$

We thus conclude that 
\begin{align*}
	-\log p_{gp}\left(\ve{z}\right)+\log p_0\left( \ve{z} \right) &\leq \frac{1}{2}\|\vesub{g}{0}\|_{\bs{C}}^2 + \frac{1}{2}\log\left| \vesub{I}{n(M\text{-}1)}+\mathbb{S} \ve{D}(\tilde{\gall})\ve{C} \right|+\lambda\\
	&\leq \frac{1}{2}\|\vesub{g}{0}\|_{\bs{C}}^2 + \frac{1}{2}\log\left| \vesub{I}{n(M\text{-}1)}+\mathbb{S}\ve{C} \right|+\lambda .
\end{align*}

Taking infimum over $\vesub{g}{0}$ and applying the representer theorem, we obtain \begin{equation}
	\frac{1}{n}\E_{\ve{t}}\left(\text{KL}\left[p_0\left(\ve{z}\right), p_{gp}\left(\ve{z}\right)\right]\right) \leq \frac{1}{2n}\|\vesub{g}{0}\|_{\bs{C}}^2+ \frac{1}{2n}\E_{\ve{t}}\left( \log\left|\vesub{I}{n(M\text{-}1)}+\mathbb{S} \ve{C} \right| \right)+\frac{\lambda}{n}, \label{eq:bound-z}
\end{equation} where $\lambda$ is a positive constant and the proposal set size $\mathbb{S}$ is a preset integer. The proof is complete.

\begin{lemma}\label{postive_for_partial2} We have $-\frac{\partial^2 \log p( {z}_i|\bs{g}(t_i)) }{ \partial \bs{g}(t_i)\partial \bs{g}(t_i)^T }>\bs{0}_{M\text{-}1}$. \begin{proof}
		We need to show that for any nonzero vector $\ve{x}=\left(x_1,\cdots, x_{M-1}\right)^T$, the following condition holds: $
		-\vesup{x}{T}\frac{\partial^2 \log p( {z}_i|\bs{g}(t_i)) }{ \partial \bs{g}(t_i)\partial \bs{g}(t_i)^T }\ve{x} >0.$ This is equivalent to demonstrating that $$\vesup{x}{T}\exp\left(\diag\left( \ve{g}(t_i)\right) \right)\ve{x}A- \vesup{x}{T}\exp\left( \ve{g}(t_i)\right) \exp\left(\ve{g}(t_i)\right)^T\ve{x} >0.$$ The left-hand side can be expressed as follows: \begin{align*}
			&\vesup{x}{T}\exp\left(\diag\left( \ve{g}(t_i)\right) \right)\ve{x}A- \vesup{x}{T}\exp\left( \ve{g}(t_i)\right) \exp\left(\ve{g}(t_i)\right)^T\ve{x}\\
			=&\left(\sum_{m=1}^{M-1}x_m^2 \exp\left(g_m(t_i)\right)\right)\left(\sum_{m=1}^{M-1}\exp\left(g_m(t_i)\right) +1 \right)-\left(\sum_{m=1}^{M-1} x_m \exp\left(g_m(t_i)\right) \right)^2.\end{align*}
		
		This can be further simplified to:
		\begin{align*}
			=& \sum_{m=1}^{M-1} x_m^2 \exp\left(2g_m(t_i)\right)+\sum_{m,h=1, h>m}^{M-1} (x_m^2+x_h^2)\exp\left(g_m(t_i)+g_h(t_i)\right) + \sum_{m=1}^{M-1}x_m^2 \exp\left(g_m(t_i)\right)\\
			& \quad - \sum_{m=1}^{M-1}x_m^2 \exp\left(2 g_m(t_i)\right)- \sum_{m,h=1, h>m}^{M-1} 2 x_m x_h\exp\left(g_m(t_i)+g_h(t_i)\right) \\
			=&\sum_{m,h=1, h>m}^{M-1} (x_m-x_h)^2\exp\left(g_m(t_i)+g_h(t_i)\right)+\sum_{m=1}^{M-1}x_m^2 \exp\left(g_m(t_i)\right) \\
			>&0.
		\end{align*} 
	\end{proof}
\end{lemma}

\section{EnMCMC Algorithm and Simulation Comparisons}\label{SM-sec:EnMCMC}
The ensemble MCMC (EnMCMC) algorithm is a method related to the multiple try Metropolis \citep{liu2000multiple} to provide a better exploration of the sample space. The details of the algorithm are given below.

\begin{algorithm}
	\label{alg:EnMCMCg} \caption{ EnMCMC sampling. }
	\begin{algorithmic}[1]
	\State Initialize $\vesub{z}{(0)}=\vesup{z}{(r-1)}$
	\State {\bf For} $s=1,\cdots,S$ \begin{itemize}
		\item[]	Draw $V$ candidates $\left\{\vess{z}{(v)}{\ast},\ v=1,\cdots, V \right\}$ from $\mathcal{Q}(\ve{z})$. Include $\vesub{z}{(s-1)}$ as one candidate and set $\vess{z}{(V+1)}{\ast}=\vesub{z}{(s-1)}$.
		\item[] Set $\vesub{z}{(s)}=\vess{z}{(v^\ast)}{\ast}$ where $v^\ast \in \left\{ 1,\cdots, V+1\right\}$ and is drawn according to the normalized weight $\mathcal{A}\left(\vesub{z}{(s-1)},\vess{z}{(v^\ast)}{\ast}\right)=\frac{\mathcal{P}\big(\vess{z}{(v^\ast)}{\ast} \big)\big/\mathcal{Q}\big(\vess{z}{(v^\ast)}{\ast}\big)}{\sum_{v=1}^{V+1}\mathcal{P}\big(\vess{z}{(v)}{\ast} \big)\big/\mathcal{Q}\big(\vess{z}{(v)}{\ast} \big) }$. 
	\end{itemize}
	\end{algorithmic}
\end{algorithm}

\medskip
The simulation settings are consistent with those utilized in Scenario 1. We evaluate the method's performance over three sample sizes ($n=30$, $60$ and $150$) to assess the effectiveness of the EnMCMC algorithm. Our findings indicate that the EnMCMC algorithm performs comparably to the independent MH algorithm, with varying choices of the proposal set sizes $\mathbb{S}$ (either $\mathbb{S}=10$ or $20$) and ensemble sizes $V$. However, it is worth noting that the EnMCMC algorithm incurs a higher computational cost, which increases with larger values of both $V$ and $n$. When the sample size is relatively small ($n=30$), using a larger ensemble size (for instance, $V = 5$) accelerates the convergence of $\gall$, resulting in a shorter computation time, particularly compared to a smaller ensemble size $V = 3$. Nonetheless, this speedup in convergence becomes less pronounced as the sample size increases. Simulation results based on fifty replications are summarized in Tables \ref{table:Scenario1-EnMCMC-fitting} and \ref{table:Scenario1-EnMCMC-classification}.

\begin{table}[!htbp]
	\caption{Comparative analysis of computational efficiency between the proposed FRVS model and CBSS method employing the independent MH and EnMCMC algorithms for Scenario 1. Average computing time (in minutes) is reported based on fifty replications, with varying sample sizes ($n=30, 60, 150$), ensemble sizes ($V=3, 5$), and proposal set sizes ($\mathbb{S}=10, 20$).}
	\label{table:Scenario1-EnMCMC-fitting}
	\centering 
	\begin{tabular}{lcccccc}
		\toprule
		& \multicolumn{2}{c}{$n=30$} &\multicolumn{2}{c}{$n=60$} &\multicolumn{2}{c}{$n=150$}\\ \cmidrule(r){2-3}\cmidrule(rl){4-5}\cmidrule(l){6-7}
		&$\mathbb{S}=10$ &$\mathbb{S}=20$ & $\mathbb{S}=10$ & $\mathbb{S}=20$ &$\mathbb{S}=10$ & $\mathbb{S}=20$\\ \midrule
		Independent MH &$0.1406$ & $0.1992$ & $0.3602$ & $0.4690$ & $2.4981$ & $3.5992 $ \\
		EnMCMC ($V=3$) &$0.1790$ & $0.2255$ & $0.4393$ & $0.5708 $ & $3.9076$ & $5.8163$ \\
		EnMCMC ($V=5$) &$0.1761$ & $0.2389$ & $0.4480$ & $0.5747$ & $5.1992$ & $8.1965$ \\
		\bottomrule
	\end{tabular}
\end{table}

\begin{table}[!htbp]
	\caption{Performance evaluation of state sequence estimation using the proposed FVRS model and CBSS method employing the EnMCMC algorithm for Scenario 1. Average classification performance metrics are presented based on fifty replications, with varying sample sizes ($n=30, 60,150$), ensemble sizes ($V=3, 5$), and proposal set sizes ($\mathbb{S}=10, 20$). }
	\label{table:Scenario1-EnMCMC-classification}
	\centering 
	\begin{tabular}{lcccccc}
		\toprule
		& \multicolumn{2}{c}{$n=30$} &\multicolumn{2}{c}{$n=60$} &\multicolumn{2}{c}{$n=150$}\\ 
        \cmidrule(r){2-3}\cmidrule(rl){4-5}\cmidrule(l){6-7}
		$(\mathbb{S}, V)$ &$(10,3)$ & $(20,3)$ &$(10,3)$ & $(20,3)$ &$(10,3)$ & $(20,3)$ \\ 
        \hline
		Accuracy & $0.9447$ & $0.9440$ & $0.9740$ & $0.9740$ & $0.9901$ & $0.9900$ \\
		Kappa & $0.8858$ & $0.8849$ & $0.9460$ & $0.9460$ & $0.9795$ & $0.9792$ \\
		Precision & $0.9711$ & $0.9746$ & $0.9845$ & $0.9845$ & $0.9947$ & $0.9947$ \\
		Specificity & $0.9533$ & $0.9600$ & $0.9758$ & $0.9758$ & $0.9920$ & $0.9920$ \\
		F1 & $0.9528$ & $0.9519$ & $0.9781$ & $0.9782$ & $0.9917$ & $0.9916$ \\ 
        \cmidrule(r){2-3}\cmidrule(rl){4-5}\cmidrule(l){6-7}
		$(\mathbb{S}, V)$ &$(10,5)$ & $(20,5)$ &$(10,5)$ & $(20,5)$ &$(10,5)$ & $(20,5)$ \\ 
        \midrule
		Accuracy & $0.9493$ & $0.9487$ & $0.9727$ & $0.9743$ & $0.9903$ & $0.9903$ \\
		Kappa & $0.8950$ & $0.8939$ & $0.9430$ & $0.9465$ & $0.9797$ & $0.9798$ \\
		Precision & $0.9642$ & $0.9660$ & $0.9765$ & $0.9802$ & $0.9932$ & $0.9943$ \\
		Specificity & $0.9433$ & $0.9467$ & $0.9633$ & $0.9692$ & $0.9897$ & $0.9913$ \\
		F1 & $0.9572$ & $0.9564$ & $0.9772$ & $0.9786$ & $0.9919$ & $0.9919$ \\
		\bottomrule
	\end{tabular}
\end{table}

\section{Additional Simulation Results}\label{SM-sec:Additional-simus}

\subsection{{Larger Noise Variance}}
\label{SM-sec:Scenario-largersigma20.02}

In this scenario, the simulation setup is similar to Scenarios 1 and 2, but with a larger noise variance $\sigma^2=0.02$. Table \ref{Appen:tab-largersigma20.02} provides a comparative analysis of state sequence estimation performance, goodness-of-fit measures and computational cost between BinSeg, NGS and the proposed CBSS methods. The results underscore that the CBSS method still performs well under this increased noise level, although some deterioration is observed, as expected.

\begin{table}[htbp]
	\caption{Comparative analysis of state sequence estimation performance, goodness-of-fit measures and computational cost between BinSeg, CBSS and NGS methods for Scenario 1 and 2 with increased $\sigma^2=0.02$. Average metrics are reported based on fifty replications, with varying sample sizes ($n=30, 60$) and proposal set sizes ($\mathbb{S}=10, 20$).}
    \label{Appen:tab-largersigma20.02}
	\centering 
    \footnotesize
	\begin{tabular}{llcccccccc}
		\toprule
		& &\multicolumn{4}{c}{$n=30$} &\multicolumn{4}{c}{$n=60$}\\ 
        \cmidrule(r){3-6}\cmidrule(l){7-10}
		& & \multirow{2}{*}{BinSeg} &\multirow{2}{*}{NGS} &CBSS  & CBSS & \multirow{2}{*}{BinSeg}  & \multirow{2}{*}{NGS} &CBSS  & CBSS  \\    
        & & & &($\mathbb{S}=10$) &  ($\mathbb{S}=20$) & & & ($\mathbb{S}=10$) & ($\mathbb{S}=20$) \\   
        \midrule
        \multicolumn{2}{l}{Scenario 1} & \\ \cmidrule{1-2}
		&\multicolumn{9}{l}{\textit{State sequence estimation performance}} \\
		&Accuracy & $0.7707$ & $0.6873$ & $0.9393$ & $0.9393$ & $0.7733$ & $0.7558$ & $0.9720$ & $0.9723$ \\
		&Kappa & $0.4864$ & $0.3423$ & $0.8735$ & $0.8740$ &$0.5016$ & $0.4319$ & $0.9418$ & $0.9425$ \\
		&Precision & $0.8014$ & $0.8199$ & $0.9619$ & $0.9588$ & $0.8276$ & $0.7797$ & $0.9811$ & $0.9806$ \\
		&Specificity & $0.6000$ & $0.6300$ & $0.9367$ & $0.9333$ & $0.6558$ & $0.4949$ & $0.9708$ & $0.9700$ \\
	   &F1 & $0.8270$ & $0.7110$ & $0.9492$ & $0.9489$ & $0.8212$ &$0.8147$ & $0.9765$ & $0.9768$ \\ 
	&	\multicolumn{9}{l}{\textit{Goodness-of-fit}} \\
	&	Rmse & - & $0.1325$ & $0.0797$ & $0.0790$ & -& $0.1818$ &$0.0894$ & $0.0908$ \\
	&	Mae  & - & $0.0895$ & $0.0615$ & $0.0601$ & - & $0.1193$ &$0.0676$ & $0.0688$ \\ 
    &	\multicolumn{9}{l}{\textit{Computational cost (min)}} \\
	&	& $7.6$e-6 & $0.1030$ & $0.1480$ & $0.1977$ &$7.9$e-6 & $0.2706$ &$0.3549$ & $0.4555$ \\
    \cmidrule{1-10}
		\multicolumn{2}{l}{Scenario 2} & \\ \cmidrule{1-2}
	&	\multicolumn{9}{l}{\textit{State sequence estimation performance}} \\
	&Accuracy &$0.6173$ & $0.6447$ & $0.9233$ & $0.9313$ & $0.5997$ & $0.7323$ & $0.9560$ & $0.9740$ \\
	&Kappa & $0.1369$ & $0.2656$ & $0.8473$ & $0.8592$ & $0.1039$ & $0.4142$ & $0.9084$ & $0.9461$ \\
	&Precision & $0.7096$ & $0.8361$ & $0.9733$ & $0.9723$ & $0.7123$& $0.8247$ & $0.9734$ & $0.9838$ \\
	&Specificity& $0.3833$ & $0.6300$ & $0.9567$ & $0.9550$ & $0.3833$& $0.6050$ & $0.9525$ & $0.9750$ \\
	&F1 &$0.7103$ & $0.6581$ & $0.9428$ & $0.9405$ & $0.6883$ &$0.7862$ & $0.9612$ & $0.9781$ \\
	&\multicolumn{9}{l}{\textit{Goodness-of-fit}} \\
	&Rmse & - &$0.5736$ &$0.4245$ & $0.4233$  & -& $0.4008$ & $0.2790$ & $0.2728$ \\
	&Mae & - &$0.4414$ &$0.3432$ & $0.3425$  & -& $0.2818$ & $0.2224$ & $0.2208$ \\
        &	\multicolumn{9}{l}{\textit{Computational cost (min)}} \\
	&	& $7.8$e-6 & $0.1028$ & $0.1467$ & $0.2176$ & $8.3$e-6 & $0.2755$ &$0.3888$ & $0.4660$ \\
    \bottomrule
	\end{tabular}
\end{table}

\subsection{Multiple State Transitions}\label{SM-sec:Scenario1-J4}
The simulation for this scenario retains the setting of Scenario 1 but modifies the sample size and hidden state sequence configuration. Specifically, the hidden state is defined as follows: $z(t)=1$ for $t \in [0,0.3) \cup [0.4,0.7)\cup[0.8,1)$, and $z(t)=2$ for $t \in [0.3,0.4)\cup[0.7,0.8)$, resulting in four state transitions. The performance of the method is assessed across two different sample sizes: $n=60$ and $n=150$. Table \ref{table:Scenario1-J4-sim-classification} presents the classification performance metrics for state sequence estimation and goodness-of-fit, comparing the IMH algorithm implemented in our CBSS method against the NGS approach based on fifty replications. CBSS method with $\mathbb{S}=20$ achieves the highest accuracy, surpassing the accuracy of the NGS. The other metrics also reflect enhancements attributable to the independent MH algorithm, and the performance disparity further increases for $n=150$.

\begin{table}[htbp]\small
	\caption{Comparative analysis of state sequence estimation performance and goodness-of-fit between CBSS and NGS methods for Scenario 1 with four state transitions. Results show average classification performance metrics based on fifty replications for varying sample sizes ($n=60,150$) and proposal set sizes ($\mathbb{S}=10,20$).}
	\label{table:Scenario1-J4-sim-classification}
	\centering 
    \footnotesize
	\begin{tabular}{lcccccc}
		\toprule
		& \multicolumn{3}{c}{$n=60$} &\multicolumn{3}{c}{$n=150$}\\ \cmidrule(r){2-4}\cmidrule(l){5-7}
		& NGS &CBSS ($\mathbb{S}=10$) & CBSS ($\mathbb{S}=20$) & NGS &CBSS ($\mathbb{S}=10$) & CBSS ($\mathbb{S}=20$) \\ \midrule
		Accuracy & $0.8743$ & $0.9423$ & $0.9433$ & $0.9117$ & $0.9757$ & $0.9771$ \\
		Kappa & $0.5789$ & $0.8230$ & $0.8269$ & $0.7056$ & $0.9238$ & $0.9277$ \\
		Precision & $0.9168$ & $0.9719$ & $0.9721$ & $0.9350$ & $0.9852$ & $0.9851$ \\
		Specificity & $0.6433$ & $0.8850$ & $0.8867$ & $0.7220$ & $0.9400$ & $0.9393$ \\
		F1 & $0.9218$ & $0.9636$ & $0.9642$ & $0.9454$ & $0.9848$ & $0.9857$ \\
		\bottomrule
        		\multicolumn{7}{l}{\textit{ {Goodness-of-fit}}} \\
		Rmse  & $0.1259$ & $0.0831$ & $0.0786$ & $0.1674$ &$0.1461$ & $0.1475$ \\
	Mae  & $0.0844$ & $0.0634$ & $0.0559$ & $0.1130$ &$0.1019$ & $0.1005$ \\ 
	\end{tabular}
\end{table}

\subsection{{Three-State Modeling Misspecification}}
\label{SM-sec:Scenario-misspec}

In this subsection, we study the robustness of the proposed method under state-space misspecification. Data are generated from a three-state latent process, whereas estimation is carried out using a two-state FRVS model and the corresponding CBSS algorithm. The purpose is not to recover the full three-state structure exactly, but to examine whether the fitted two-state model can still detect the dominant regime changes and approximate the main transition pattern when the true latent dynamics are more complex.

Specifically, the true hidden state follows the transition pattern $
1 \rightarrow 2 \rightarrow 3 \rightarrow 2 \rightarrow 1.$ The hidden state process is defined as $$
z(t)=
\begin{cases}
1, & t \in [0,0.2)\cup[0.8,1),\\
2, & t \in [0.2,0.4)\cup[0.7,0.8),\\
3, & t \in [0.4,0.7).
\end{cases}
$$
The time points $\{t_i\}_{i=1}^n$ are equally spaced on $(0,1)$.

The covariance hyper-parameters for states 1 and 2 are taken to be the same as those in Scenario 1: $
\{\nu_{10}=0.1,\ \nu_{11}=0.1,\ \vesub{A}{10}=1,\vesub{A}{11}=0.1\}, $ and $
\{\nu_{20}=0.1,\ \nu_{21}=0.5,\ A_{20}=1,\ A_{21}=1\}. $
For the additional third state, we set
$
\{\nu_{30}=0.1,\ \nu_{31}=0.1,\ A_{30}=1,\ A_{31}=1\},$ which has the same fluctuation magnitude as state 1 but with different within-state correlation structure.  The observations are then generated in the same way as in Scenario 1, with independent noise terms $
\epsilon_i \sim N(0, 0.01). $

Although the true data-generating process contains three hidden states, estimation is still carried out using a two-state FRVS model and the corresponding
CBSS method. Since the fitted model contains only two latent states, exact recovery of the full three-state path is not a meaningful target. Instead, for state-sequence evaluation we collapse the true three-state path into two broader regimes by treating
state 3 as equivalent to state 1. That is, the true path
$ (1,2,3,2,1) $ is relabeled as $(1,2,1,2,1),$ and the classification metrics in Table \ref{Appen:table:3states} are computed
with respect to this coarsened two-state criterion. This choice is natural because the fitted model is designed to distinguish two dominant regimes, and in the present data-generating setting state 3 is closer to state 1 than to state 2 in terms of fluctuation magnitude.

\begin{table}[htbp]
	\caption{Comparative analysis of state sequence estimation performance, and goodness-of-fit measures between BinSeg, CBSS and NGS methods for Scenario 1 and 2 under model misspecification. Average metrics are reported based on fifty replications, with varying sample sizes ($n=30, 60$) and proposal set sizes ($\mathbb{S}=10, 20$).}
	\label{Appen:table:3states}
	\centering 
    \footnotesize
	\begin{tabular}{lcccccccc}
		\toprule
		& \multicolumn{4}{c}{$n=60$} &\multicolumn{4}{c}{$n=150$}\\ 
        \cmidrule(r){2-5}\cmidrule(l){6-9}
		& \multirow{2}{*}{BinSeg} &\multirow{2}{*}{NGS} &CBSS  & CBSS & \multirow{2}{*}{BinSeg}  & \multirow{2}{*}{NGS} &CBSS  & CBSS  \\    
        & & &($\mathbb{S}=10$) &  ($\mathbb{S}=20$) & & & ($\mathbb{S}=10$) & ($\mathbb{S}=20$) \\   
        \midrule
		\multicolumn{9}{l}{\textit{State sequence estimation performance}} \\
		Accuracy & $0.6633$ & $0.7880$ & $0.9430$ & $0.9473$ & $0.7116$ & $0.8315$ & $0.9581$ & $0.9747$ \\
		Kappa & $0.1763$ & $0.1285$ & $0.8214$ & $0.8336$ &$0.3291$ & $0.1930$ & $0.8630$ & $0.9193$ \\
		Precision & $0.7819$ & $0.8295$ & $0.9658$ & $0.9683$ & $0.8187$ & $0.8337$ & $0.9636$ & $0.9808$ \\
		Specificity & $0.4317$ & $0.1900$ & $0.8600$ & $0.8667$ & $0.5613$ & $0.1820$ & $0.8487$ & $0.9213$ \\
	   F1 & $0.7543$ & $0.8705$ & $0.9643$ & $0.9670$ & $0.7836$ &$0.9053$ & $0.9742$ & $0.9843$ \\ 
		\multicolumn{9}{l}{\textit{Goodness-of-fit}} \\
		Rmse & - & $0.1165$ & $0.0831$ & $0.0786$ & -& $0.1674$ &$0.1461$ & $0.1475$ \\
	Mae  & - & $0.0795$ & $0.0634$ & $0.0559$ & - & $0.1130$ &$0.1019$ & $0.1005$ \\ 
    \bottomrule
	\end{tabular}
\end{table}

Under this coarsened evaluation criterion, the proposed method exhibits robustness to moderate misspecification of the number of hidden states. The resulting classification performance is comparable to that reported in Appendix \ref{SM-sec:Scenario1-J4}, where the underlying latent process likewise involves multiple transitions but remains compatible with a fitted two-state specification. In contrast, the goodness-of-fit errors do not diminish as the sample size increases. Since the fitted model is structurally misspecified, additional observations can enhance the identification of the predominant regime changes, but they cannot remove the systematic bias induced by the incorrectly specified dimensionality of the state space.

Moreover, the estimation of a genuinely three-state latent path is substantially more challenging in practical applications. Under realistic conditions, the underlying state process may exhibit multiple short-lived transitions, heterogeneous dwell times, and recurrent visits to covariance regimes with similar characteristics, rendering the latent path considerably more complex than the stylized configuration examined here. Determining whether an intermediate segment should be identified as a distinct third state, as opposed to a manifestation of variability within an existing state, generally necessitates additional model
structure and more sophisticated inferential techniques. The development of robust and accurate methods for estimating such multi-state paths is an important direction for future research.

\subsection{Increased Sample Size}\label{SM-sec:Scenario2-n150}

In this scenario, the simulation setup is similar to Scenario 2, focusing on an increased sample size to $n = 150$. Table \ref{table:Scenario2-Appen-sim} provides a comparative analysis of the performance metrics for the proposed FRVS model and CBSS method using the independent MH algorithm against the NGS approach. The findings underscore that CBSS method, particularly with $\mathbb{S} = 20$, yields superior performance across various criteria, demonstrating enhanced accuracy in both state sequence classification and observation fitting.

\begin{table}[htbp]
	\caption{Comparative analysis of state sequence estimation performance and goodness-of-fit measures between CBSS and NGS methods for Scenario 2 with sample size $n=150$. Results show average classification performance metrics based on fifty replications, with varying proposal set sizes ($\mathbb{S}=10,20$) }
	\label{table:Scenario2-Appen-sim}
	\centering 
	\begin{tabular}{lccc}
		\toprule
		& NGS &CBSS ($\mathbb{S}=10$) & CBSS ($\mathbb{S}=20$) \\ \midrule
		\multicolumn{4}{l}{\textit{State sequence estimation}} \\
		Accuracy & $0.9232$ & $0.9895$ & $0.9892$ \\
		Kappa & $0.8431$ & $0.9781$ & $0.9776$ \\
		Precision & $0.9644$ & $0.9943$ & $0.9943$ \\
		Specificity & $0.9423$ & $0.9913$ & $0.9913$ \\
		F1 & $0.9327$ & $0.9912$ & $0.9910$ \\
		\multicolumn{4}{l}{\textit{Goodness-of-fit}} \\
		Rmse &$0.0940$ &$0.0987$ & $0.0805$ \\
		Mae &$0.0617$ &$0.0601$ & $0.0564$ \\
		\bottomrule
	\end{tabular}
\end{table}

\section{Multi-Output Modeling}\label{SM-sec:CP-MO}
Here, we detail the multi-output modeling where the response variable $\ve{y}(t)$ is $P$-dimensional function, such as $\ve{y}(t)=\left(y_1(t), \cdots, y_P(t)\right)^T$, then the model in Equation (1) defined in Section 2.1 is extended to
\begin{align*}
	y_p(t)&=\mu_{z(t) p}(\ve{x}(t))+f_{z(t)p}(\ve{x}(t))+\epsilon_p(t),\ p=1,\cdots, P,\\
	f_{z(t)p}(\ve{x}(t))&=f_{mp}(\ve{x}(t)), \text{ when } z(t)=m,\ m \in \{1,\cdots, M\},\\
	\epsilon_p(t)&\stackrel{i.i.d.}{\sim} N(0\, ,\,\sigma_p^2). 
\end{align*}

In the context of multi-output GP, the underlying functions $\vesub{f}{m}(\ve{x}(t))=\\ \left(f_{m1}(\ve{x}(t)),\cdots,f_{mP}(\ve{x}(t)) \right)^T\in \Re^P$, for $m=1,\cdots, M$. Careful selection of a multi-output GP prior assumption is necessary to accurately characterize the relationships across the outputs. As 
\begin{align*}
	\Cov\left(y_p(t_i),y_{p^\ast}(t_j)\right)=& \Cov\left(f_{z(t_i)p}(\ve{x}(t_i)),f_{z(t_j)p^\ast}(\ve{x}(t_j)) \right) \\
    & +\sigma_p^2\mathbb{I}(i=j,p=p^\ast),\ \ p,\, p^\ast \in \{1,\cdots, {P}\},
\end{align*} 
denoting the covariates $\vesub{x}{i}=\ve{x}(t_i)$ and $\vesub{x}{j}=\ve{x}(t_j)$, the covariance and cross-covariance among these functions are defined by 
$$	\Cov \left(f_{mp}(\vesub{x}{i}), f_{m^\ast p^\ast }(\vesub{x}{j} )\right)=k_{mm^\ast\, p p^\ast}(\vesub{x}{i},\vesub{x}{j}), \, \text{for}\ m,\, m^\ast \in \{1,\cdots, M\}, \ p,\, p^\ast \in \{1,\cdots, {P}\}. $$ 

The convolution process serves as an effective approach for constructing multi-output GP, as it guarantees that the resulting covariance matrix is at least non-negative definite. To implement this process, we construct three types of GPs using the convolution operator, denoted by $``\star"$, paired with specific kernel functions ${\kappa}(\cdot)$ and independent Gaussian white noise processes $\varepsilon(\cdot)$. These are defined as follows:
\begin{equation*}
	\begin{aligned}
		\xi_{mp}(\cdot)&={\kappa}_{mp0}(\cdot)\star \varepsilon_0(\cdot),\\
		\zeta_{mp}(\cdot)&={\kappa}_{mp1}(\cdot)\star \varepsilon_{m0}(\cdot),\\
		\eta_{mp}(\cdot)&={\kappa}_{mp2}(\cdot) \star \varepsilon_{mp}(\cdot),\ \text{for } m \in\{1,\cdots, M\},\, p \in \{1,\cdots, {P}\}. 
	\end{aligned} \label{eq:conv-2parts}
\end{equation*} Here, the process $\xi_{mp}(\cdot)$ captures the common components shared among all $MP$ curves, $\zeta_{mp}(\cdot)$ characterizes attributes of curves observed in different states, and $\eta_{mp}(\cdot)$ delineates the unique properties associated with each curve. 

Furthermore, we define $f_{mp}(\ve{x}(t))=\xi_{mp}(\ve{x}(t))+	\zeta_{mp}(\ve{x}(t))+\eta_{mp}(\ve{x}(t))$ and specify the smoothing kernel functions with parameters $\vesub{\theta}{f}=\{ \upsilon_{mp0},\upsilon_{mp1},\upsilon_{mp2}, \vesub{A}{mp0}, \vesub{A}{mp1},\vesub{A}{mp2},m=1,\cdots, M,\ p= 1,\cdots, {P}\}$ as follows: 
\begin{equation*}
	\begin{aligned}\begin{cases}
			{\kappa}_{mp 0}(\ve{x}(t))=\upsilon_{mp0}\exp\left\{ -\frac{1}{2}\ve{x}(t)^{T} \vesub{A}{mp0}\ve{x}(t)\right\}\\
			{\kappa}_{mp 1}(\ve{x}(t))=\upsilon_{mp1}\exp\left\{ -\frac{1}{2}\ve{x}(t)^{T} \vesub{A}{mp1}\ve{x}(t)\right\}\\
			{\kappa}_{mp 2}(\ve{x}(t))=\upsilon_{mp2}\exp\left\{ -\frac{1}{2} \ve{x}(t)^{T} \vesub{A}{mp2}\ve{x}(t) \right\}
		\end{cases} 
		,\quad m=1,\cdots, M,\ p= 1,\cdots, P.
	\end{aligned}
\end{equation*}
The combined vector for all outputs is represented as $\ve{\mathrm{f}}(\ve{x}(t))=\left(\vesub{f}{1}(\ve{x}(t))^T,\cdots,\vesub{f}{M}(\ve{x}(t))^T \right)^T$, which defines a dependent $MP$-variate GP with zero means and the kernel covariance functions given by 

$	k_{mm^\ast\, p p^\ast }(\vesub{x}{i},\vesub{x}{j})=
$
\begin{equation*}
	\left\{
	\begin{aligned} &\pi^{d/2}
		\upsilon_{mp0}^2\left|\vesub{A}{mp0} \right|^{-1/2} \exp\left\{ -\frac{1}{4}(\vesub{x}{i}-\vesub{x}{j})^T\vesub{A}{mp0}(\vesub{x}{i}-\vesub{x}{j}) \right\}\\
		& \quad + 
		\pi^{d/2}\upsilon_{mp1}^2\left|\vesub{A}{mp1} \right|^{-1/2} \exp\left\{ -\frac{1}{4}(\vesub{x}{i}-\vesub{x}{j})^T\vesub{A}{mp1}(\vesub{x}{i}-\vesub{x}{j}) \right\}\\
		& \qquad + 
		\pi^{d/2}\upsilon_{mp2}^2\left|\vesub{A}{mp2} \right|^{-1/2} \exp\left\{ -\frac{1}{4}(\vesub{x}{i}-\vesub{x}{j})^T\vesub{A}{mp2}(\vesub{x}{i}-\vesub{x}{j}) \right\}, \ \text{if}\, m^\ast=m\, \text{and}\ p^\ast=p, \\
		&(2\pi)^{d/2}
		\upsilon_{mp0}\upsilon_{m p^\ast 0}\left|\vesub{A}{mp0}+\vesub{A}{m p^\ast0} \right|^{-1/2} \\
        & \quad \exp\bigg\{ -\frac{1}{2}(\vesub{x}{i}-\vesub{x}{j})^T \vesub{A}{mp0}\left(\vesub{A}{mp0}+\vesub{A}{m p^\ast0}\right)^{-1} \vesub{A}{m p^\ast0}(\vesub{x}{i}-\vesub{x}{j}) \bigg\}\\
		&\quad+ (2\pi)^{d/2}
		\upsilon_{mp1}\upsilon_{m p^\ast 1}\left|\vesub{A}{mp1}+\vesub{A}{m p^\ast1} \right|^{-1/2} \exp\bigg\{ -\frac{1}{2}(\vesub{x}{i}-\vesub{x}{j})^T \vesub{A}{mp1}\left(\vesub{A}{mp1}+\vesub{A}{m p^\ast1}\right)^{-1}\\
		&\qquad\qquad\qquad\qquad\qquad\qquad\qquad\qquad\qquad\qquad \vesub{A}{m p^\ast1}(\vesub{x}{i}-\vesub{x}{j}) \bigg\} ,\ \text{if}\ m^\ast=m\, \text{and}\ p^\ast\neq p,\\
		&(2\pi)^{d/2}
		\upsilon_{mp0}\upsilon_{m^\ast p^\ast 0}\left|\vesub{A}{mp0}+\vesub{A}{m^\ast p^\ast0} \right|^{-1/2} \exp\bigg\{ -\frac{1}{2}(\vesub{x}{i}-\vesub{x}{j})^T\vesub{A}{mp0}\left(\vesub{A}{mp0}+\vesub{A}{m^\ast p^\ast0}\right)^{-1} \\
		&\qquad\qquad\qquad\qquad\qquad\qquad\qquad\qquad\qquad\qquad \vesub{A}{m^\ast p^\ast0}(\vesub{x}{i}-\vesub{x}{j}) \bigg\},\ \text{if}\ m^\ast\neq m.
	\end{aligned} \right.
	\label{eq:MO-cov-cross}
\end{equation*}

\subsection{Applications}\label{SM-sec:MOAP}
For application purposes, we analyze the complete tri-axial acceleration data in the jumping exercise and FoG datasets. These datasets correspond to accelerations along three axes: vertical, mediolateral, and anteroposterior directions, as visualized in Figures \ref{fig:MO-jump-real} and \ref{fig:MO-fog-real}., respectively.
\begin{figure}
	\centering
	\subfloat[Jumping exercise.]{	
		\includegraphics[width=1\linewidth]{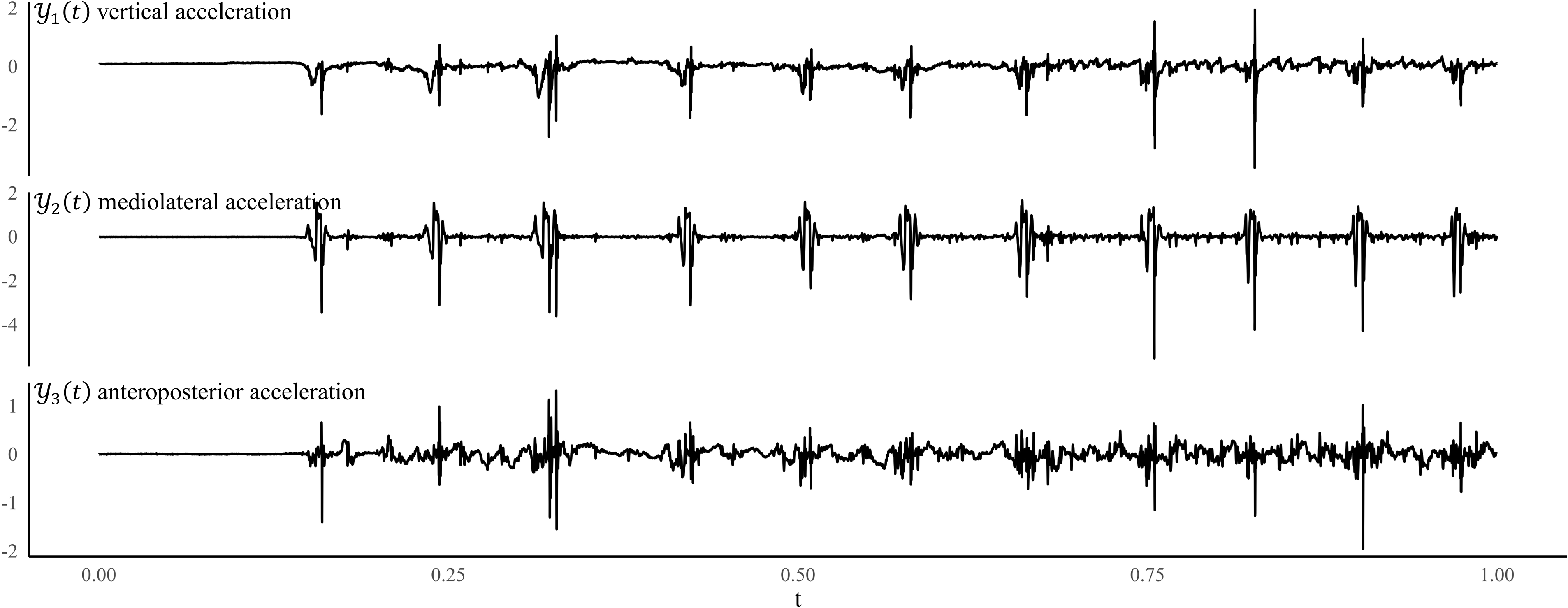}
		\label{fig:MO-jump-real}
	}\\
	\subfloat[FoG, with labeled onset and termination timestamps for the FoG episode.]{
	\includegraphics[width=1\linewidth]{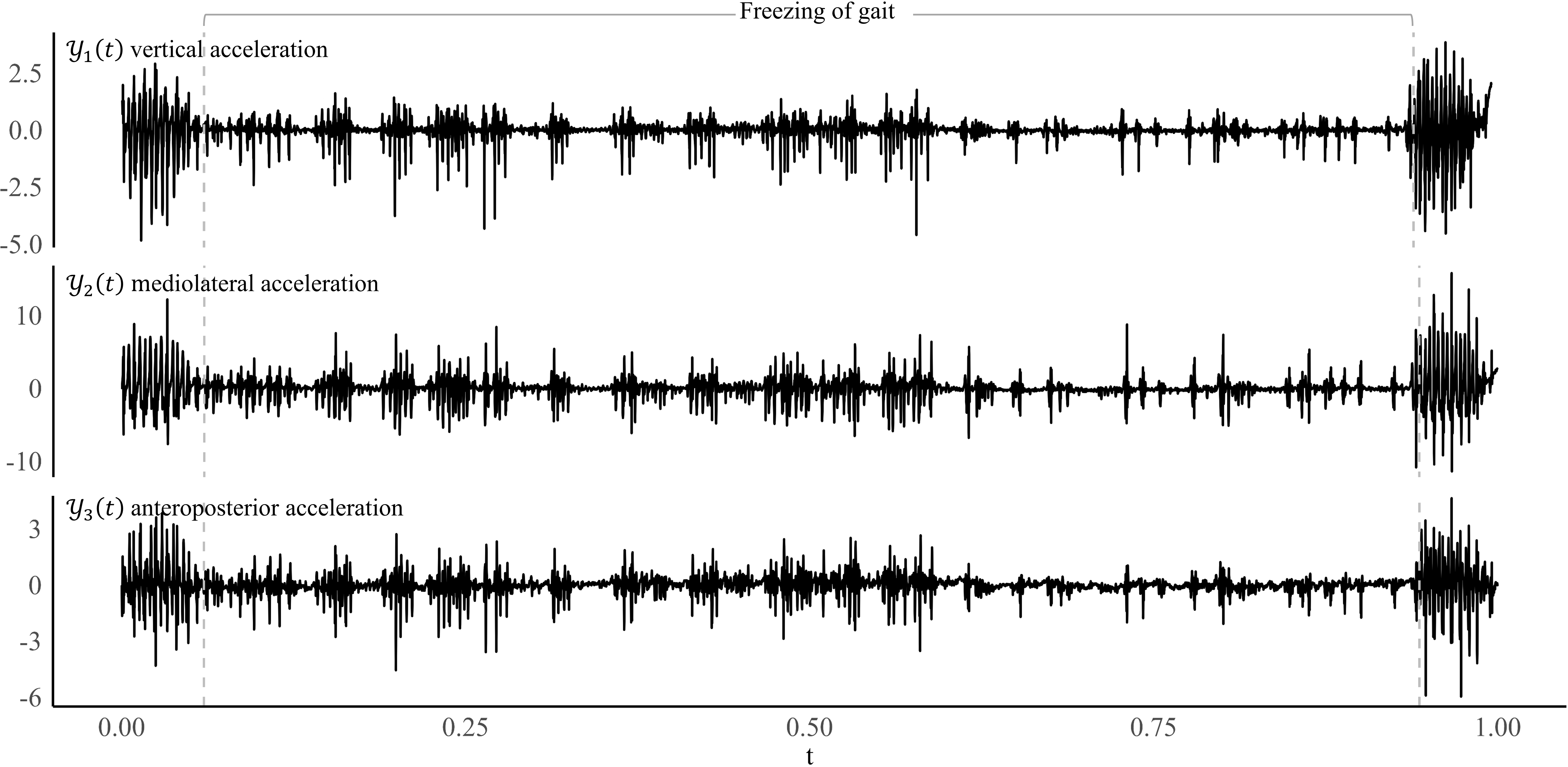}
		\label{fig:MO-fog-real}
		}
\caption{Real world tri-axial acceleration signal data.}
\label{fig:MO-real}
\end{figure}

In the multi-output analysis scenario, the estimation results closely align with those derived from a single-output approach, albeit with minor discrepancies. The fitting and prediction results for the two datasets are illustrated in Figures \ref{fig:MO-jump-ans} and \ref{fig:MO-fog-ans}, respectively. Notably. in the multi-output context, the estimated noise variance was larger than that obtained from a single-output case. This discrepancy can be attributed to the pronounced fluctuations in $y_2(t)$, representing the mediolateral acceleration, compared to the other two axes of accelerations, as evident from the raw data observations. Specifically, for the jumping exercise, the estimated variance is $\sigma^2=0.0166$, while for the FoG dataset, it is $\sigma^2=0.6229$.
\begin{figure}
	\centering
	\subfloat[Jumping exercise segmentation.]{	
		\includegraphics[width=1\linewidth]{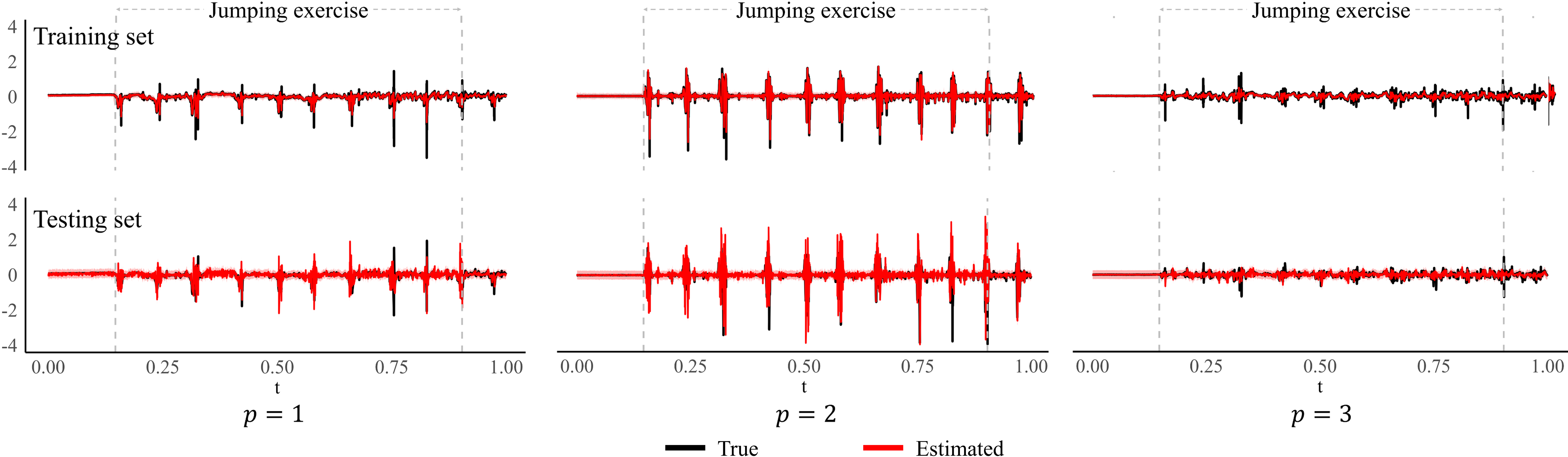}
	\label{fig:MO-jump-ans}
	}\\
	\subfloat[FoG detection.]{
		\includegraphics[width=1\linewidth]{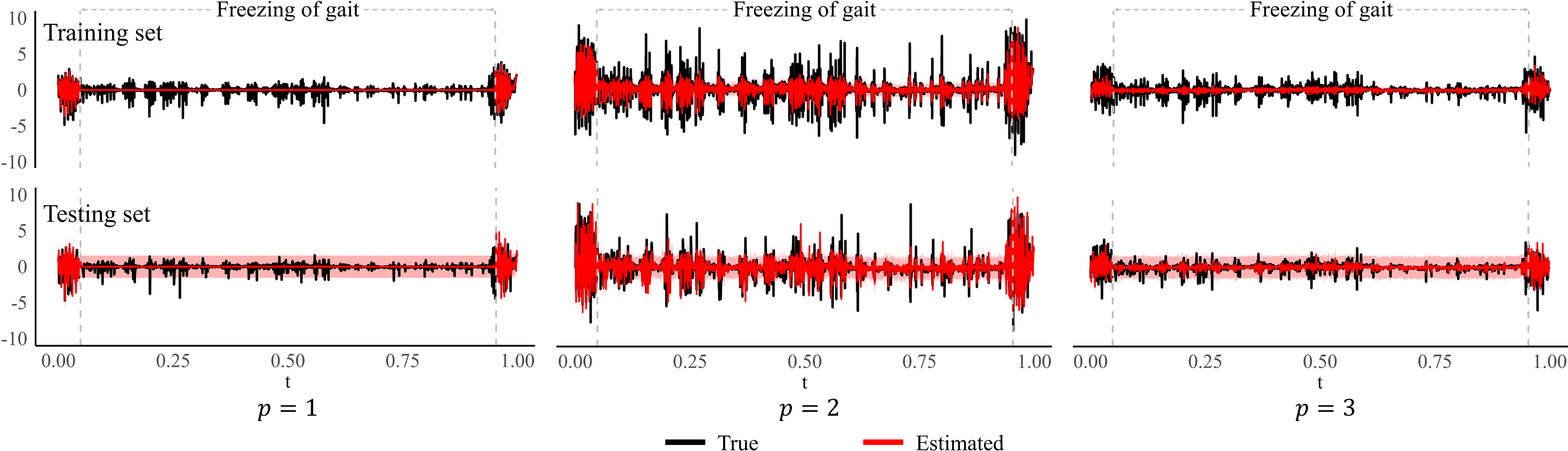} 
		\label{fig:MO-fog-ans}
	}
	\caption{Estimation results for the dataset respectively, with indications of the estimated state transition period. Each column corresponds to the estimation results for one dimension of the multi-output.}
\end{figure}
%
%

Due to the influence of fluctuations in $y_2(t)$, the fitting and prediction errors in this axis are larger than in the others. However, the inclusion of multiple outputs enhances overall information content of the analysis. When specifically evaluating the fitting and prediction of $y_1(t) $ against the results from the single-output scenario, we observe a reduction in error measures, as summarized in Table \ref{table-app}. The jumping exercise duration within the multi-output analysis is segmented from $t = 0.1465$ to $0.9036$, with parameters estimated as $\vesub{\theta}{z} = \{q_{12} = 4.0634, q_{21} = 1.3270\}$. Additionally, the hidden state parameters and classification metrics for the FoG dataset indicate nearly identical results compared to the single-output case, achieving an accuracy of $0.9986$, kappa of $0.9969$, precision of $1$, specificity of $1$, and f1 score of $0.9943$, with parameters estimated as $\vesub{\theta}{z} = \{q_{12} = 8.6381, q_{21} = 1.5997\}$.

\begin{table}[htbp]
	\caption{Fitting and prediction performance results for multi-output modeling applied to the two datasets. Rmse and mae metrics for the training set ($75\%$) and testing set ($25\%$) are reported for each output dimension, where (single) denotes results derived from single-output modeling and this corresponding to the dimension of $p=1$.}
	\label{table-app}
	\centering
	\begin{tabular}{cccccc}
		\toprule
		\multicolumn{2}{c}{Data} & \multicolumn{2}{c}{Training set} & \multicolumn{2}{c}{Testing set} \\ \midrule
		& $p$ & rmse & mae & rmse & mae \\ \cmidrule{2-2}\cmidrule(l){3-6} 
		\multirow{4}{*}{Jumping exercise} & 1 & $\bs{0.0999}$& $\bs{0.0370}$ & $\bs{0.1242}$ & $\bs{0.0656}$ \\
		& 2 & $0.1423$ & $0.0476$ & $0.4803$ & $0.1691$ \\
		& 3 & $0.0909$ & $0.0467$ & $0.1307$ & $0.0705$ \\ 
		& 1 (single) & $0.1198$ & $0.0503$ & $0.1386$ & $0.0714$ \\ \hline
		\multirow{4}{*}{FoG}& 1 & $\bs{0.3444}$ & $\bs{0.1828}$ & $\bs{0.4852}$ & $\bs{0.2298}$ \\
		& 2 & $0.9682$ & $0.5158$ & $1.4171$ & $0.7643$ \\
		& 3 & $0.4716$ & $0.2847$ & $0.6150$ & $0.3537$ \\
		& 1 (single) & $0.3637$ & $0.2005$ & $0.4980$ & $0.2556$ \\
		\bottomrule 
	\end{tabular}
\end{table}

\end{appendices}


\bibliography{bibreference}

@book{ramsay2005functional,
	title = {Functional Data Aanalysis},
	author = {J.O. Ramsay and B.W. Silverman},
	series = {Springer Series in Statistics},
	year = {2005},
	publisher = {Springer},
	address = {New York},
	doi = {https://doi.org/10.1007/b98888},
	isbn_hardcover = {978-0-387-40080-8},
	isbn_softcover = {978-1-4419-2300-4},
	isbn_ebook = {978-0-387-22751-1},
	edition = {2},
	volume = {XIX},
	pages = {429},
	topics = {Statistical Theory and Methods}
}

@article{chamidah2024estimation,
	title={Estimation of multiresponse multipredictor nonparametric regression model using mixed estimator},
	author={Nur Chamidah and Budi Lestari and I Nyoman Budiantara and Dursun Aydin},
	journal={Symmetry},
	volume={16},
	number={4},
	pages={386},
	year={2024},
	publisher={MDPI},
	issn= {2073-8994},
	doi={10.3390/sym16040386}
}

@article{rice1991estimating,
	title={Estimating the mean and covariance structure nonparametrically when the data are curves},
	author={John A. Rice and B. W. Silverman},
	journal={Journal of the Royal Statistical Society. Series B (Methodological)},
	volume={53},
	number={1},
	pages={233--243},
	year={1991},
	publisher={Oxford University Press},
	issn={00359246}
}

@article{Wang2020low,
	author = {Jiayi Wang and Raymond K. W. Wong and Xiaoke Zhang},
	title = {Low-rank covariance function estimation for multidimensional functional data},
	journal = {Journal of the American Statistical Association},
	volume = {117},
	number = {538},
	pages = {809--822},
	year = {2020},
	issn={1537-274X},
	publisher = {Informa UK Limited},
	doi = {10.1080/01621459.2020.1820344}
}

@article{horvath2022change,
	title={Change point analysis of covariance functions: A weighted cumulative sum approach},
	author={Lajos Horv{\'a}th and Gregory Rice and Yuqian Zhao},
	journal={Journal of Multivariate Analysis},
	volume={189},
	pages={104877},
	year={2022},
	publisher={Elsevier},
	issn = {0047-259X},
	doi = {https://doi.org/10.1016/j.jmva.2021.104877}
}

@article{aue2009estimation,
	title={Estimation of a change-point in the mean function of functional data},
	author={Alexander Aue and Robertas Gabrys and Lajos Horv{\'a}th and Piotr Kokoszka},
	journal={Journal of Multivariate Analysis},
	volume={100},
	number={10},
	pages={2254--2269},
	year={2009},
	publisher={Elsevier},
	issn = {0047-259X},
	doi = {https://doi.org/10.1016/j.jmva.2009.04.001}
}

@article{Rabiner1986AnIntro,
	author={L. Rabiner and B. Juang},
	journal={IEEE ASSP Magazine}, 
	title={An introduction to hidden {M}arkov models}, 
	year={1986},
	volume={3},
	number={1},
	pages={4-16},
	doi={10.1109/MASSP.1986.1165342}
}

@book{cappe2005inference,
	title={Inference in Hidden {M}arkov Models},
	author={Olivier Capp{\'e} and Eric Moulines and Tobias Ryd{\'e}n},
	isbn={9780387402642},
	lccn={2005923551},
	series={Springer Series in Statistics},
	year={2005},
	publisher={Springer},
	address = {New York}
}

@article{liu2015efficient,
	title={Efficient learning of continuous-time hidden {M}arkov models for disease progression},
	author={Yu-Ying Liu and Shuang Li and Fuxin Li and Le Song and James M. Rehg},
	journal={Advances in Neural Information Processing Systems},
	volume={28},
	year={2015},
	publisher={Neural information processing systems foundation},
	pages={3600--3608},
	issn={1049-5258}
}

@article{Zhou2020Continuous,
	author = {Jie Zhou and Xinyuan Song and Liuquan Sun},
	year = {2020},
	pages = {104646},
	title = {Continuous time hidden {M}arkov model for longitudinal data},
	volume = {179},
	journal = {Journal of Multivariate Analysis},
	issn = {0047-259X},
	doi = {https://doi.org/10.1016/j.jmva.2020.104646},
	publisher={Elsevier}
}

@article{Rabiner1989Atutorial,
	author={L.R. Rabiner},
	journal={Proceedings of the IEEE}, 
	title={A tutorial on hidden {M}arkov models and selected applications in speech recognition}, 
	year={1989},
	volume={77},
	number={2},
	pages={257-286},
	doi={10.1109/5.18626}
}

@book{ross1996stochastic,
	title={Stochastic Processes},
	author={Sheldon M. Ross},
	isbn={9780471120629},
	lccn={82008619},
	series={Wiley Series in Probability and Mathematical Statistics},
	year={1995},
	publisher={Wiley},
	address={New York}
}

@article{shi2020regression,
	title={Regression analysis for multivariate process data of counts using convolved {G}aussian processes},
	author={A’yunin Sofro and Jian Qing Shi and Chunzheng Cao},
	journal={Journal of Statistical Planning and Inference},
	volume={206},
	pages={57--74},
	year={2020},
	publisher={Elsevier},
	issn={0378-3758},
	doi={https://doi.org/10.1016/j.jspi.2019.09.005}
}

@incollection{bernardo1998regression,
	title={Regression and classification using {G}aussian process priors},
	author={Radford M. Neal},
	booktitle = {Bayesian Statistics 6: Proceedings of the Sixth Valencia International Meeting June 6-10, 1998},
	publisher = {Oxford University Press},
	address = {Oxford},
	year = {1999},
	month = {08},
	isbn = {9780198504856},
	doi = {10.1093/oso/9780198504856.003.0021}
}

@book{williams2005gaussian,
	title={Gaussian Processes for Machine Learning},
	author={Carl Edward Rasmussen and Christopher K. I. Williams},
	isbn={9780262182539},
	lccn={2005053433},
	series={Adaptive Computation and Machine Learning Series},
	year={2005},
	publisher={MIT press},
	address={Cambridge, MA}
}

@book{shi2011gaussian,
	title={Gaussian Process Regression Analysis for Functional Data},
	author={Jian Qing Shi and Taeryon Choi},
	isbn={9781439837733},
	lccn={2011026789},
	year={2011},
	publisher={Chapman \& Hall/CRC},
	address={New York}
}

@article{datta2016hierarchical,
	title={Hierarchical nearest-neighbor {G}aussian process models for large geostatistical datasets},
	author={Abhirup Datta and Sudipto Banerjee and Andrew O. Finley and Alan E. Gelfand},
	journal={Journal of the American Statistical Association},
	volume={111},
	number={514},
	pages={800--812},
	year={2016},
	publisher = {ASA Website},
	doi = {10.1080/01621459.2015.1044091},
	note ={PMID: 29720777}
}

@article{wilson2021pathwise,
	title={Pathwise conditioning of {G}aussian processes},
	author={James T. Wilson and Viacheslav Borovitskiy and Alexander Terenin and Peter Mostowsky and Marc Peter Deisenroth},
	journal={Journal of Machine Learning Research},
	volume={22},
	number={105},
	pages={1--47},
	year={2021}
}

@article{shi2007gaussian,
	title={{G}aussian process functional regression modeling for batch data},
	author={J. Q. Shi and B. Wang and R. Murray-Smith and D. M. Titterington},
	journal={Biometrics},
	volume={63},
	number={3},
	pages={714--723},
	year={2007},
	publisher={Oxford University Press},
	ISSN = {0006341X, 15410420}
}

@article{Yi2011Penalized,
	author = {G. Yi and J. Q. Shi and T. Choi},
	title = {Penalized {G}aussian process regression and classification for high-dimensional nonlinear data},
	journal = {Biometrics},
	volume = {67},
	number = {4},
	pages = {1285-1294},
	year = {2011},
	month = {03},
	issn = {0006-341X},
	doi = {10.1111/j.1541-0420.2011.01576.x}
}

@book{karlin1981second,
	title={A Second Course in Stochastic Processes},
	author={Samuel Karlin and Howard M. Taylor},
	isbn={9780123986504},
	lccn={lc80000533},
	year={1981},
	publisher={Elsevier Science},
	address={New York}
}

@Inbook{robert1999gibbs,
	title={The {G}ibbs {S}ampler},
	author={Christian P. Robert and George Casella},
	bookTitle={Monte Carlo Statistical Methods},
	pages={285--361},
	year={1999},
	publisher={Springer},
	address={New York, NY},
	isbn={978-1-4757-3071-5},
	doi={10.1007/978-1-4757-3071-5_7}
}

@article{gelfand2000gibbs,
	title={{G}ibbs sampling},
	author={Alan E. Gelfand},
	journal={Journal of the American Statistical Association},
	volume={95},
	number={452},
	pages={1300--1304},
	year={2000},
	publisher = {ASA Website},
	doi={10.1080/01621459.2000.10474335}
}

@book{meyn2012markov,
	title={{M}arkov Chains and Stochastic Stability},
	author={S.P. Meyn and R.L. Tweedie},
	isbn={9781447132677},
	series={Communications and Control Engineering},
	year={2012},
	publisher={Springer},
	address={London}
}

@article{metropolis1953equation,
	title={Equation of state calculations by fast computing machines},
	author={Nicholas Metropolis and Arianna W. Rosenbluth and Marshall N. Rosenbluth and Augusta H. Teller and Edward Teller},
	journal={The Journal of Chemical Physics},
	volume={21},
	number={6},
	pages={1087--1092},
	year={1953},
	publisher={American Institute of Physics},
	issn = {0021-9606},
	doi = {10.1063/1.1699114}
}

@article{Hastings1970Monte,
	author = {W. K. Hastings},
	title ={{M}onte {C}arlo sampling methods using {M}arkov chains and their applications},
	journal = {Biometrika},
	volume = {57},
	number = {1},
	pages = {97-109},
	year = {1970},
	issn = {0006-3444},
	doi = {10.1093/biomet/57.1.97}
}

@article{casella2001empirical,
	title={Empirical {B}ayes {G}ibbs sampling},
	author={George Casella},
	journal={Biostatistics},
	volume={2},
	number={4},
	pages={485--500},
	year={2001},
	publisher={Oxford University Press}
}

@book{bishop2006pattern,
	title={Pattern Recognition and Machine Learning},
	author={Christopher M. Bishop},
	isbn={9780387310732},
	lccn={2006922522},
	series={Information Science and Statistics},
	year={2006},
	publisher={Springer},
	address={New York}
}

@book{liu2001monte,
	title={{M}onte {C}arlo Strategies in Scientific Computing},
	author={J.S. Liu},
	isbn={9780387952307},
	lccn={lc00069243},
	series={Springer Series in Statistics},
	year={2001},
	publisher={Springer},
	address={New York}
}

@article{liu2000multiple,
	title={The multiple-try method and local optimization in {M}etropolis sampling},
	author={Jun S. Liu and Faming Liang and Wing Hung Wong},
	journal={Journal of the American Statistical Association},
	volume={95},
	number={449},
	pages={121--134},
	year={2000},
	publisher={ASA Website},
	doi={10.1080/01621459.2000.10473908}
}

@article{wang2014generalized,
	title={Generalized {G}aussian process regression model for non-{G}aussian functional data},
	author={Bo Wang and Jian Qing Shi},
	journal={Journal of the American Statistical Association},
	volume={109},
	number={507},
	pages={1123--1133},
	year={2014},
	publisher = {ASA Website},
	doi={10.1080/01621459.2014.889021}
}

@article{gassiat2016inference,
	title={Inference in finite state space non parametric hidden {M}arkov models and applications},
	author={{\'E}lisabeth Gassiat and Alice Cleynen and Stephane Robin},
	journal={Statistics and Computing},
	volume={26},
	pages={61--71},
	year={2016},
	publisher={Springer},
	issn={1573-1375},
	doi={10.1007/s11222-014-9523-8}
}

@article{gassiat2016nonpara,
	author = {Elisabeth Gassiat and Judith Rousseau},
	title ={Nonparametric finite translation hidden {M}arkov models and extensions},
	volume = {22},
	journal = {Bernoulli},
	number = {1},
	publisher = {Bernoulli Society for Mathematical Statistics and Probability},
	pages = {193 -- 212},
	year = {2016},
	doi = {10.3150/14-BEJ631}
}

@article{gassiat2020identifiability,
	author  = {Elisabeth Gassiat and Sylvain {Le Corff} and Luc Leh{\'e}ricy},
	title   = {Identifiability and consistent estimation of nonparametric translation hidden {M}arkov models with general state space},
	journal = {Journal of Machine Learning Research},
	year    = {2020},
	volume  = {21},
	number  = {115},
	pages   = {1--40},
}

@article{alexandrovich2016nonparametric,
	title={Nonparametric identification and maximum likelihood estimation for hidden {M}arkov models},
	author={Grigory Alexandrovich and Hajo Holzmann and Anna Leister},
	journal={Biometrika},
	volume={103},
	number={2},
	pages={423--434},
	year={2016},
	publisher={Oxford University Press},
	ISSN = {00063444},
}

@article{seeger2008information,
	title={Information consistency of nonparametric {G}aussian process methods},
	author={Matthias W. Seeger and Sham M. Kakade and Dean P. Foster},
	journal={IEEE Transactions on Information Theory},
	volume={54},
	number={5},
	pages={2376--2382},
	year={2008},
	publisher={IEEE},
	doi={10.1109/TIT.2007.915707}
}

@book{gelman2007data,
	title={Data Analysis Using Regression and Multilevel/Hierarchical Models},
	author={A. Gelman and J. Hill},
	isbn={9780521686891},
	lccn={2006040566},
	series={Analytical Methods for Social Research},
	year={2007},
	publisher={Cambridge University Press},
	address={New York}
}

@article{powell2021investigating,
	title={Investigating the AX6 inertial-based wearable for instrumented physical capability assessment of young adults in a low-resource setting},
	author={Dylan Powell and Mina Nouredanesh and Samuel Stuart and Alan Godfrey},
	journal={Smart Health},
	volume={22},
	pages={100220},
	year={2021},
	publisher={Elsevier},
	issn = {2352-6483},
	doi = {https://doi.org/10.1016/j.smhl.2021.100220}
}

@article{van2015assessing,
	title={Assessing physical activity in older adults: Required days of trunk accelerometer measurements for reliable estimation},
	author={Kimberley S. van Schooten and Sietse M. Rispens and Petra J.M. Elders and Paul Lips and Jaap H. van Dieën and Mirjam Pijnappels},
	journal={Journal of Aging and Physical Activity},
	volume={23},
	number={1},
	pages={9-17},
	year={2015},
	publisher={Human Kinetics, Inc.},
	doi={10.1123/JAPA.2013-0103}
}

@article{del2019gait,
	title={Gait analysis with wearables predicts conversion to {P}arkinson disease},
	author={Silvia Del Din and Morad Elshehabi and Brook Galna and Markus A. Hobert and Elke. Warmerdam and Ulrike Suenkel and Kathrin Brockmann and Florian Metzger and Clint Hansen and Daniela Berg and Lynn Rochester and Walter Maetzler},
	journal = {Annals of Neurology},
	volume = {86},
	number = {3},
	pages = {357-367},
	doi = {https://doi.org/10.1002/ana.25548},
	year={2019},
	publisher={Wiley Online Library}
}

@article{hickey2016detecting,
	title={Detecting free-living steps and walking bouts: Validating an algorithm for macro gait analysis},
	author={Aodh{\'a}n Hickey and Silvia Del Din and Lynn Rochester and Alan Godfrey},
	journal={Physiological Measurement},
	volume={38},
	number={1},
	pages={N1},
	year={2016},
	publisher={IOP Publishing},
	doi = {10.1088/1361-6579/38/1/N1}
}

@article{LYONS2005497,
	title = {A description of an accelerometer-based mobility monitoring technique},
	journal = {Medical Engineering \& Physics},
	volume = {27},
	number = {6},
	pages = {497-504},
	year = {2005},
	issn = {1350-4533},
	doi = {https://doi.org/10.1016/j.medengphy.2004.11.006},
	author = {G.M. Lyons and K.M. Culhane and D. Hilton and P.A. Grace and D. Lyons}
}

@article{maclean2023walking,
	title={Walking bout detection for people living in long residential care: A computationally efficient algorithm for a 3-axis accelerometer on the lower back},
	author={Mhairi K. MacLean and Rana Zia Ur Rehman and Ngaire Kerse and Lynne Taylor and Lynn Rochester and Silvia Del Din},
	journal={Sensors},
	volume={23},
	number={21},
	pages={8973},
	year={2023},
	publisher={MDPI},
	issn={1424-8220},
	doi={10.3390/s23218973}
}

@article{nonnekes2020freezing,
	title={Freezing of gait and its levodopa paradox},
	author={Jorik Nonnekes and Matthieu Bereau and Bastiaan R. Bloem},
	journal={JAMA Neurology},
	volume={77},
	number={3},
	pages={287--288},
	year={2020},
	publisher={American Medical Association},
	issn = {2168-6149},
	doi = {10.1001/jamaneurol.2019.4006}
}

@article{gilat2018freezing,
	title={Freezing of gait: Promising avenues for future treatment},
	author={Moran Gilat and Ana L{\'\i}gia Silva de Lima and Bastiaan R. Bloem and James M. Shine and Jorik Nonnekes and Simon J.G. Lewis},
	journal={Parkinsonism \& Related Disorders},
	volume={52},
	pages={7--16},
	year={2018},
	publisher={Elsevier},
	issn={1353--8020},
	doi={10.1016/j.parkreldis.2018.03.009}
}

@article{mancini2019clinical,
	title={Clinical and methodological challenges for assessing freezing of gait: Future perspectives},
	author={Martina Mancini and Bastiaan R. Bloem and Fay B. Horak and Simon J.G. Lewis and Alice Nieuwboer and Jorik Nonnekes},
	journal={Movement Disorders},
	volume={34},
	number={6},
	pages={783--790},
	year={2019},
	publisher={Wiley Online Library},
	doi = {https://doi.org/10.1002/mds.27709}
}

@article{xu2019automatic,
	title={Automatic detection of significant areas for functional data with directional error control},
	author={Peirong Xu and Youngjo Lee and Jian Qing Shi and Janet Eyre},
	journal={Statistics in Medicine},
	volume={38},
	number={3},
	pages={376--397},
	year={2019},
	publisher={Wiley Online Library},
	doi = {https://doi.org/10.1002/sim.7968}
}

@article{datta2016nearest,
	title={On nearest-neighbor {G}aussian process models for massive spatial data},
	author={Abhirup Datta and Sudipto Banerjee and Andrew O. Finley and Alan E. Gelfand},
	journal={Wiley Interdisciplinary Reviews: Computational Statistics},
	volume={8},
	number={5},
	pages={162--171},
	year={2016},
	publisher={Wiley Online Library},
	doi = {https://doi.org/10.1002/wics.1383}
}

@article{vecchia1988estimation,
	title={Estimation and model identification for continuous spatial processes},
	author={Aldo V. Vecchia},
	journal={Journal of the Royal Statistical Society. Series B (Methodological)},
	volume={50},
	number={2},
	pages={297--312},
	year={1988},
	publisher={Oxford University Press},
	ISSN = {00359246}
}

@article{datta2016nonseparable,
	title={Nonseparable dynamic nearest neighbor {G}aussian process models for large spatio-temporal data with an application to particulate matter analysis},
	author={Abhirup Datta and Sudipto Banerjee and Andrew O. Finley and Nicholas A. S. Hamm and Martijn Schaap},
	journal = {The Annals of Applied Statistics},
	volume={10},
	number={3},
	pages={1286--1316},
	year={2016},
	publisher = {Institute of Mathematical Statistics},
	doi = {10.1214/16-AOAS931}
}

@book{titterington1985,
	title={Statistical Analysis of Finite Mixture Distributions},
	author={D. M. Titterington and A. F. M. Smith and U. E. Makov},
	volume={38},
	year={1985},
	publisher={Wiley},
	address = {New York},
	isbn={9780471907633},
	lccn={85006434},
	series={Applied Section}
}

@article{allman2009identifiability,
	title={Identifiability of parameters in latent structure models with many observed variables},
	author={Elizabeth S. Allman and Catherine Matias and John A. Rhodes},
	journal={The Annals of Statistics},
	volume={37},
	number={6A},
	pages={3099--3132},
	year={2009},
	publisher={Institute of Mathematical Statistics},
	ISSN = {00905364, 21688966}
}

@article{huang2013nonparametric,
	title={Nonparametric mixture of regression models},
	author={Mian Huang and Runze Li and Shaoli Wang},
	journal={Journal of the American Statistical Association},
	volume={108},
	number={503},
	pages={929--941},
	year={2013},
	publisher = {ASA Website},
	doi = {10.1080/01621459.2013.772897}
}

@article{wilcox1967exponential,
	title={Exponential operators and parameter differentiation in quantum physics},
	author={Ralph M. Wilcox},
	journal={Journal of Mathematical Physics},
	volume={8},
	number={4},
	pages={962--982},
	year={1967},
	publisher={American Institute of Physics}
}

@book{lehmann1983theory,
	title={Theory of Point Estimation},
	author={E.L. Lehmann},
	isbn={9780471058496},
	lccn={82021881},
	series={Probability and Statistics Series},
	year={1983},
	publisher={Wiley},
	address={University of Minnesota}
}

@article{crowder1976maximum,
	title={Maximum likelihood estimation for dependent observations},
	author={Martin J. Crowder},
	journal={Journal of the Royal Statistical Society. Series B (Methodological)},
	volume={38},
	number={1},
	pages={45--53},
	year={1976},
	publisher={Oxford University Press},
	issn = {0035-9246},
	doi = {10.1111/j.2517-6161.1976.tb01565.x}
}

@article{aitchison1958maximum,
	title={Maximum-likelihood estimation of parameters subject to restraints},
	author={J. Aitchison and S. D. Silvey},
	journal={The Annals of Mathematical Statistics},
	volume={29},
	number={3},
	pages={813--828},
	year={1958},
	publisher={Institute of Mathematical Statistics},
	doi = {10.1214/aoms/1177706538}
}

@article{CaiOptimal2011,
author = {T. Tony Cai and Ming Yuan},
title = {Optimal estimation of the mean function based on discretely sampled functional data: {P}hase transition},
volume = {39},
journal = {The Annals of Statistics},
number = {5},
publisher = {Institute of Mathematical Statistics},
pages = {2330 -- 2355},
year = {2011},
doi = {10.1214/11-AOS898}
}

@article{Li2023HMMGPFR,
	author = {Li, Tao and Ma, Jinwen},
	title= {Hidden {M}arkov mixture of {G}aussian process functional regression: {U}tilizing Multi-Scale Structure for time series forecasting},
	journal = {Mathematics},
	volume= {11},
	year= {2023},
	number= {5},
	issn = {2227-7390},
	doi = {10.3390/math11051259}
}

@misc{jung2020HMMGP,
    author={Yohan Jung and Jinkyoo Park},
	title={Scalable hybrid {HMM} with {G}aussian process emission for sequential time-series data clustering. arXiv preprint, arXiv:2001.01917}, 
    archivePrefixjournal = {arXiv preprint arXiv:2001.01917},
    eprint={2001.01917},
    journal={arXiv preprint arXiv:2001.01917},
	year={2020},
	primaryClass={cs.LG}
}

@article{Hamada2016Modeling,
	author={Hamada, Ryunosuke and Kubo, Takatomi and Ikeda, Kazushi and Zhang, Zujie and Shibata, Tomohiro and Bando, Takashi and Hitomi, Kentarou and Egawa, Masumi},
	journal={IEEE Transactions on Intelligent Vehicles}, 
	title={Modeling and prediction of driving behaviors using a nonparametric {B}ayesian method with {AR} models}, 
	year={2016},
	volume={1},
	number={2},
	pages={131-138},
	doi={10.1109/TIV.2016.2586307}
	}

@article{scott1974cluster,
  title={A cluster analysis method for grouping means in the analysis of variance},
  author={Scott, Andrew Jhon and Knott, Martin},
  journal={Biometrics},
  pages={507--512},
  year={1974},
  publisher={JSTOR}
}

@article{fryzlewicz2014wild,
  title={Wild binary segmentation for multiple change-point detection},
  author={Fryzlewicz, Piotr},
  journal={The Annals of Statistics},
  volume={42},
  number={6},
  pages={2243--2281},
  year={2014}
}
\label{lastpage}

\end{document}